\documentclass[sigconf]{acmart}

\AtBeginDocument{%
  }

\usepackage{tikz}
\usepackage{amsmath}
\usepackage{xurl}
\usepackage{pifont}
\usepackage{colortbl}        %
\usepackage{subcaption}
\usepackage{threeparttable}
\usepackage{booktabs}
\usepackage{makecell}
\usepackage{graphicx} %
\usepackage{stmaryrd}
\usepackage{amsthm} %
\usepackage{enumitem}
\usepackage[normalem]{ulem}
\usepackage{xparse}
\usepackage{xspace}
\usepackage{url}

\NewDocumentCommand{\DeepUL}{O{0.85ex} O{0.3ex} m}{%
  {\begingroup
   \renewcommand\ULdepth{#1}%
   \renewcommand\ULthickness{#2}%
   \uline{#3}%
   \endgroup}%
}
\newcommand{\bnumS}[3]{%
  \tikz[baseline=(char.base)]{
    \node[shape=circle, fill=black, inner sep=#1] (char)
      {\color{white}#2\bfseries #3};
  }%
}
\newcommand{\NA}{\textemdash}
\newcommand{\paragraphNew}[1]{\noindent\textbf{#1.}\xspace}
\newcommand{\cmark}{\textcolor{green!80!black}{\ding{51}}} %
\newcommand{\xmark}{\textcolor{red!80!black}{\ding{55}}}   %
\newtheorem{theorem}{Theorem}[section]

\newtheoremstyle{idealfuncstyle}%
  {}{}%
  {\itshape}%
  {}%
  {\bfseries}%
  {:}%
  { }%
  {}%

\theoremstyle{idealfuncstyle}

\newtheorem{lemma}[theorem]{Lemma}

\theoremstyle{definition}
\newtheorem{definition}[theorem]{Definition}

\newcommand{\ProtocolName}{\textmd{\textsc{\textsf{SEEK}}}}

\begin{document}
\date{}

\title{\ProtocolName{}: \underline{S}ecure and \underline{E}fficient \underline{E}ncrypted \underline{K}eyword Search \\For Privacy-Preserving Messaging Protocols}

\author{Soumyadyuti Ghosh}
\orcid{0000-0002-2015-4192}
\affiliation{%
\department{Center for Cyber Security}
  \institution{New York University Abu Dhabi}
  \city{Abu Dhabi}
  \country{United Arab Emirates}
  }
\email{sg8466@nyu.edu}

\author{Michail Maniatakos}
\orcid{0000-0001-6899-0651}
\affiliation{%
\department{Center for Cyber Security}
  \institution{New York University Abu Dhabi}
  \city{Abu Dhabi}
  \country{United Arab Emirates}
  }
\email{michail.maniatakos@nyu.edu}

\renewcommand{\shorttitle}{\ProtocolName{}: \underline{S}ecure and \underline{E}fficient \underline{E}ncrypted \underline{K}eyword Search For Privacy-Preserving Messaging Protocols}

\begin{abstract}
Encrypted communication protects sensitive user data but can facilitate harmful or unlawful exchanges, creating a trade-off between detecting dangerous messages and preserving end-user privacy. To address this, we propose \ProtocolName{}, a practical and efficient encrypted keyword-search protocol for privacy-preserving messaging that combines homomorphic encryption with secure two-party computation (2PC). \ProtocolName{} first partitions messages into ciphertext fragments with the minimum sufficient overlap, then homomorphically correlates them using encrypted keyword trapdoors. For long messages, this design can reduce sender-side encryption and upload overhead by up to two orders of magnitude over state-of-the-art baselines. It supports ASCII case-insensitive matching with one fixed-size encrypted trapdoor and one homomorphic multiplication per fragment, yielding up to 5.47$\times$ faster correlation computation than the strongest fragmentation-based baselines. \ProtocolName{} then invokes 2PC-based selected decoding, blinded zero testing, and secure aggregation, revealing only the keyword presence-or-absence bit while hiding the keyword, its length, message contents, match counts, and locations. \ProtocolName{} achieves $100\%$ accuracy under case variations that result in exact-matching failures, without requiring additional trapdoors or online communication. We further realize \ProtocolName{} as an end-to-end web and cross-platform mobile application. Prototype evaluation on a weekly messaging history yields an online computation time of $1.92$~s per search, demonstrating the practical feasibility and efficiency of \ProtocolName{}.
\end{abstract}

\ccsdesc[500]{Security and privacy~Privacy-preserving protocols}

\keywords{Privacy-preserving Keyword Search, E2EE Messaging,  Homomorphic Encryption, Secure Two-Party Computation (2PC).}

\maketitle

\section{Introduction}
\label{sec:Introduction}

In the modern digital world, users regularly exchange highly sensitive personal, corporate, and governmental information through cloud-serviced messaging platforms that 
store private conversations beyond the direct control of their users. Without 
cryptographic protection, these communications can be exposed to
untrustworthy cloud service providers ($\mathsf{CSP}$s), 
compromising user privacy at scale. Consequently, End-to-End Encryption (E2EE) has become a de facto requirement for secure messaging,
ensuring that only the sender and recipient can read 
messages, while $\mathsf{CSP}$ handles only encrypted data, 
thus substantially reducing 
associated privacy concerns.

\noindent \ding{110} \textbf{\DeepUL{Motivation.}} While preserving the confidentiality of private communications remains essential, lawfully authorized investigations increasingly require targeted searches of retained messages for evidence of misconduct or crime.

\textit{\textbf{{Searchability Paradox.}}} Retrospective access to communications can have substantial investigative value in serious crime cases. Infiltrations of encrypted criminal platforms such as \textit{EncroChat} and Operation Trojan Shield enabled authorities to analyze messages related to drugs, weapons, money laundering, violence, and corruption, leading to large-scale arrests and asset seizures \cite{N18,N19,N20}. \textit{EncroChat} litigation also produced legal rulings, with UK courts admitting evidence obtained under a targeted equipment-interference warrant and the Court of Justice of the European Union addressing the cross-border transmission and use of such evidence under European Investigation Orders~\cite{N33,N34}. Legal process can likewise compel disclosure when a provider retains readable message content, as occurred in the Nebraska Facebook/Messenger investigation~\cite{N39}. However, approaches built on provider-side plaintext access, broad endpoint-side scanning, or reusable exceptional-access capabilities threaten to transform targeted investigations into scalable surveillance infrastructure, creating substantial risks to privacy, security, accountability, mission creep, and civil liberties \cite{N30,N31,N32}. 

\textit{\textbf{Access-Control Paradox.}} The tension surrounding E2EE has further intensified due to ongoing governmental and regulatory initiatives across jurisdictions aimed at weakening, bypassing, or reversing robust encryption standards. In the United Kingdom, authorities reportedly sought access to Apple users' encrypted cloud data, prompting Apple to withdraw \textit{Advanced Data Protection} for new UK users~\cite{N21,N22}. European Union negotiations over the proposed CSA Regulation continue to consider whether detection orders should extend to E2EE interpersonal communications~\cite{N23,N24}, while Australia's \textit{Assistance and Access} framework enables 
agencies to seek provider assistance when encryption or other technical barriers impede lawful investigations~\cite{N26,N27}. The removal of opt-in  E2EE for private messages by Instagram, along with reports regarding Apple's operations in China, further demonstrates how the availability of encryption, data storage practices, and infrastructure control can shift in response to safety, content moderation, law enforcement, or local compliance requirements~\cite{N28,N29}. Conversely, providers and E2EE services have resisted demands to weaken or bypass encryption, as illustrated by the Facebook Messenger wiretap dispute and WhatsApp's traceability challenge in India~\cite{N41,N42,N25}. These divergent responses demonstrate why the technical interface matters: major providers typically require valid legal process before disclosing stored content, transparency reporting reveals the scale of such requests, and data-minimization measures limit what providers retain and can disclose~\cite{N37,N38,N40}.

\textit{\textbf{Cryptographic Middle Ground.}} Together, these competing imperatives motivate the central design question: \textit{In legally authorized investigations, should access require bulk disclosure of decrypted messages, or be limited to scoped cryptographic queries over encrypted data?} We pursue the latter approach by constraining authorized access through the technical interface rather than 
goodwill or 
general plaintext disclosure. Our aim is therefore to offer a narrower cryptographic middle ground: \textit{a trade-off that can satisfy legitimate law-enforcement needs for targeted message probing while cryptographically preventing mass-scale surveillance of the encrypted corpus.}

\begin{figure}[!t]
    \centering
    \includegraphics[width=\columnwidth]{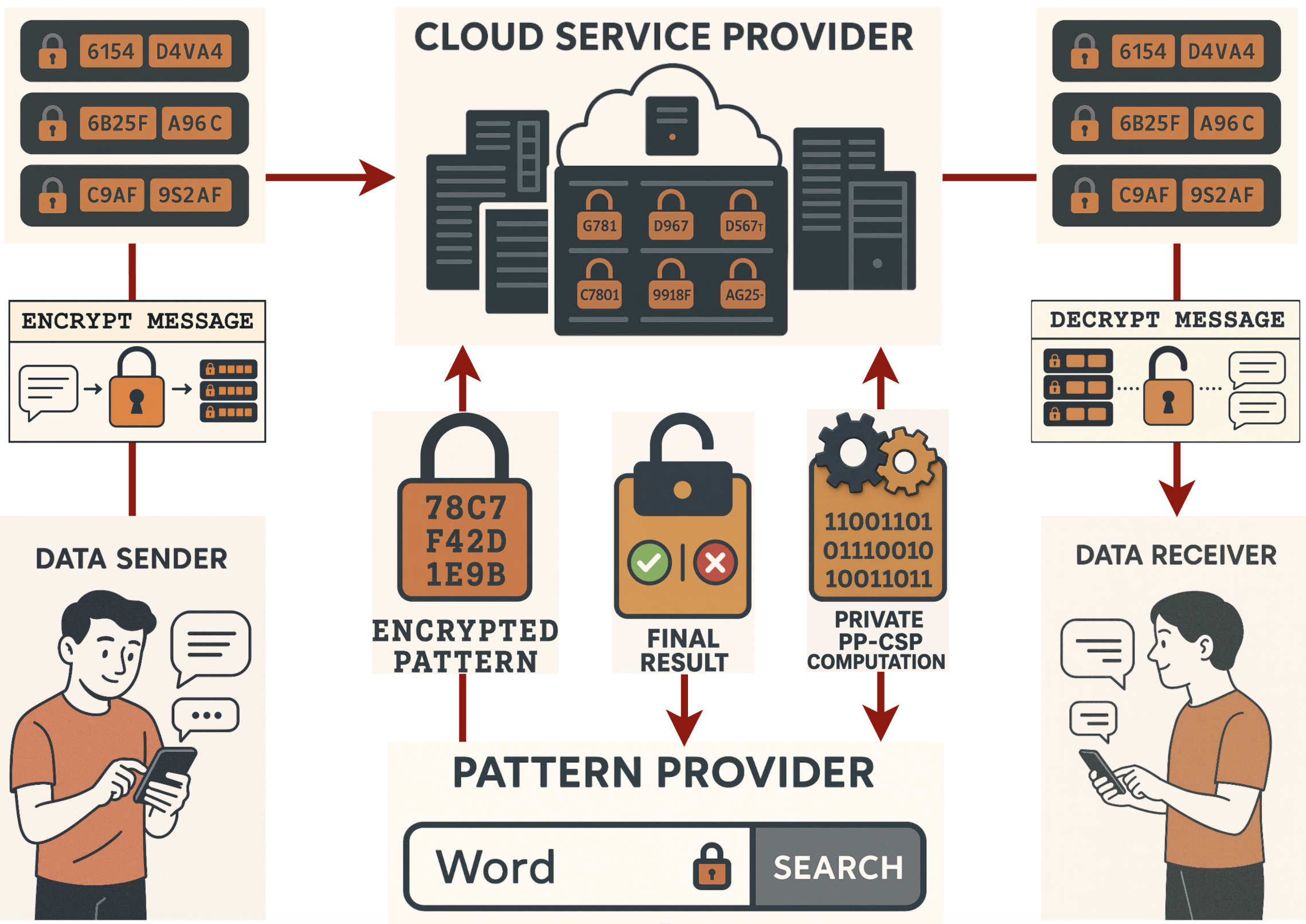}
    \captionsetup{font=small,justification=centering}
    \caption{Fundamental \ProtocolName{} Protocol Architecture.}
    \label{fig:flowdiagram}
\end{figure}

\noindent \ding{42} \textbf{\DeepUL{Design Objectives.}} Realizing this middle ground requires a search interface that does not depend on endpoint cooperation. Our design objectives differ from traditional victim-initiated abuse-reporting mechanisms for E2EE messaging and message-franking-based protocols, in which a cooperative recipient reports an abusive message together with its cryptographic evidence \cite{GrubbsLR17,issa2022hecate}. In contrast, the proposed setting considers users who may intentionally exchange harmful or unlawful messages and fall within the scope of investigation, so neither can be expected to initiate a voluntary abuse report. In a realistic deployment, users communicate via E2EE while $\mathsf{CSP}$ stores the encrypted messages. An authorized pattern provider ($\mathsf{PP}$) encodes a harmful or relevant keyword as an encrypted trapdoor and submits it to the $\mathsf{CSP}$, which executes a privacy-preserving search protocol over the encrypted corpus. The $\mathsf{CSP}$ learns neither the plaintext messages, queried keyword, nor the search outcome, while $\mathsf{PP}$ learns only the presence or absence of the keyword. Fig.~\ref{fig:flowdiagram} illustrates the foundational \ProtocolName{} architecture and its associated workflow.

\subsection{Related Work}
\label{RelatedWork}
Encrypted search over private data spans several domains. Below, we briefly describe each approach and its limitations.

\noindent \bnumS{1.6pt}{\normalsize}{1} \textbf{\DeepUL{Indexed and Tokenized Search.}} Index-based searchable encryption (SE) enables conjunctive queries, updates, access control, multi-user sharing, and forward or backward privacy over encrypted indexes. However, it typically matches only indexed keywords and returns document identifiers~\cite{11,22,24,25,23,20}. Consequently, retrospective arbitrary-substring search requires relevant substrings to be indexed or represented using specialized structures. Encrypted suffix-tree constructions facilitate such searches with linear-size ciphertexts, but they require three communication rounds, reveal prescribed access patterns, and return occurrence indices~\cite{1:CS15}. 
Moreover, practical query-reconstruction attacks exploit the scheme-specific leakage profiles of prominent substring-SSE schemes~\cite{gui2024leakage}.
Encrypted deep packet inspection (DPI) systems such as \textit{BlindBox} and \textit{BlindIDS} tokenize network traffic to enable encrypted equality matching~\cite{1:SLPR15,1:CDK+17}. Supporting variable-length patterns necessitates either scheme-specific tokenization or query decomposition. In contrast, SEST enables arbitrary-length matching using shiftable trapdoors and bilinear-pairing tests, but discloses match locations~\cite{1:DFOS18}. Consequently, retrospective arbitrary search with only a presence bit requires modified search representations, private aggregation, or result obfuscation.

\noindent\bnumS{1.6pt}{\normalsize}{2} \textbf{\DeepUL{Fragmentation-Based Encrypted Search.}}  A distinct line of research divides streams into fixed-size encrypted fragments and introduces overlap or auxiliary boundary instances to preserve cross-boundary matches~\cite{1:BCS21,1:BCC20,1:BCS23}. These approaches employ either bilinear-pairing constructions~\cite{1:BCC20,1:BCS21} or functional-encryption (FE) primitives~\cite{1:BCS23}. In the latter approach, fragmentation avoids position-dependent key material that would otherwise grow linearly with stream length, but the native functionality returns match positions rather than a hidden aggregate. Pairing-based matching also requires computationally costly bilinear-pairing evaluations across the tested instances. Pattern privacy remains limited because~\cite{1:BCS21} does not target pattern-hiding trapdoors, while~\cite{1:BCC20} protects patterns from the $\mathsf{CSP}$ only under a high min-entropy assumption. These limitations position homomorphic encryption (HE) as a complementary design choice for privacy-preserving pattern/keyword search in the encrypted domain.

\begin{table}[!t]
\centering
\setlength{\tabcolsep}{1.05pt}
\renewcommand{\arraystretch}{1.08}
\caption{Qualitative comparison of representative encrypted pattern/keyword-matching approaches with \ProtocolName{}.}
\label{tab:seek-compare-intro}
\begin{threeparttable}
\resizebox{\columnwidth}{!}{%
\begin{tabular}{@{}l@{\hspace{0pt}}cccccccccc@{}}
\toprule
\rowcolor{yellow!10}
\hspace{5mm}\textbf{Approach}
& \makecell{\hspace{1.4pt}\textbf{Any}\hspace{1.4pt}\\\textbf{KW}}
& \makecell{\hspace{1.4pt}\textbf{Any}\hspace{1.4pt}\\\textbf{PL}}
& \makecell{\hspace{1.4pt}\textbf{Single}\hspace{1.4pt}\\\textbf{TD}}
& \makecell{\hspace{1.4pt}\textbf{Dynamic}\hspace{1.4pt}\\\textbf{History}}
& \makecell{\textbf{HE}\\\textbf{Evaluation}}
& \makecell{\textbf{ASCII}\\\textbf{CI}}
& \makecell{\textbf{Exact}\\\textbf{match}}
& \makecell{\textbf{No}\\\textbf{ASP}}
& \makecell{\textbf{Hide}\\\textbf{PL}}
& \makecell{\textbf{Private}\\\textbf{presence}} \\
\midrule

\rowcolor{black!8}
\multicolumn{11}{@{}l}{\hspace{2pt}\textbf{Encrypted Search Paradigms}} \\
\rowcolor{yellow!7}
SE \cite{11,22,24,25}
& \xmark & \textit{part.} & \cmark
& \textit{part.} & \NA
& \xmark & \cmark
& \xmark & \xmark & \xmark \\
\rowcolor{yellow!7}
DPI \cite{1:SLPR15,1:CDK+17}
& \cmark & \xmark & \textit{part.}
& \textit{part.} & \NA
& \xmark & \cmark
& \textit{part.} & \xmark & \xmark \\
\rowcolor{yellow!7}
SEST \cite{1:DFOS18}
& \cmark & \cmark & \cmark
& \cmark & \NA
& \xmark & \cmark
& \xmark & \xmark & \xmark \\
\rowcolor{yellow!7}
Pairing/FE\cite{1:BCS21,1:BCC20,1:BCS23}
& \cmark & \cmark & \textit{dep.}
& \cmark & \NA
& \xmark & \cmark
& \textit{part.} & \xmark & \xmark \\

\addlinespace[1pt]
\rowcolor{black!8}
\multicolumn{11}{@{}l}{\hspace{2pt}\textbf{Homomorphic String Search Constructions}} \\
\rowcolor{yellow!7}
CipherMatch \cite{7}
& \cmark & \cmark & \xmark
& \xmark & Add-only
& \xmark & \cmark
& \cmark & \xmark & \xmark \\
\rowcolor{yellow!7}
RLWE Frag. \cite{1}
& \cmark & \cmark & $1/2$
& \cmark & $2\mathrm{CC}{+}2\mathrm{CP}$
& \cmark & \cmark
& \cmark & \xmark & \xmark \\
\rowcolor{yellow!7}
BGV Search \cite{bonte2020homomorphic}
& \cmark & \cmark & \cmark
& \textit{part.} & HomEQ
& \xmark & \textit{rand.}
& \cmark & \xmark & \xmark \\
\rowcolor{yellow!7}
TFHE Search \cite{narisada2026tfhe}
& \cmark & \cmark & \xmark
& \xmark & PBS/CMux
& \xmark & \cmark
& \cmark & \cmark\textsuperscript{a} & \xmark \\
\rowcolor{yellow!7}
CKKS Search \cite{secoaguirre2026text}
& \cmark & \cmark & \cmark
& \xmark & Poly.+rot.
& \xmark & \textit{approx.}
& \cmark & \xmark & \textit{part.}\textsuperscript{b} \\

\midrule
\rowcolor{blue!16}
\textbf{\ProtocolName{} (Section \ref{sec:seek-presence-protocol})}
& \cmark & \cmark & \cmark
& \cmark & $\mathbf{1CC}$ %
& \cmark & \cmark
& \cmark & \cmark & \cmark \\
\bottomrule
\end{tabular}%
}
\end{threeparttable}

\vspace{2pt}
\begin{minipage}{\columnwidth}
\footnotesize
\noindent\textbf{KW}: Keyword; \textbf{PL}: Pattern Length; \textbf{TD}: Trapdoor; \textbf{CI}: Case-Insensitive; \textbf{ASP}: CSP Access/Search Pattern Leakage; \textbf{CC/CP}: ct--ct/ct--pt Multiplication; \textbf{HomEQ}: Homomorphic Equality; \textbf{PBS}: Programmable Bootstrapping. \textbf{\textit{part.}/\textit{dep.}}: Partial/Model-Dependent. \\
\textbf{Any KW, Any PL, and Single TD} denote post-encryption query choice, length variation without corpus reprocessing, and one fixed-size encrypted trapdoor, respectively. \textbf{Dynamic history} means incremental encrypted ingestion with cross-boundary completeness.
\textbf{Private presence} means only one party (e.g., $\mathsf{PP}$) learns corpus-wide presence bit. \textbf{HE Evaluation} lists dominant query-dependent homomorphic core operations (e.g., correlation score). For~\cite{1}, $1/2$ is the exact/CI TD count and wildcard costs exclude pair formation.
\textbf{\textit{rand.}/\textit{approx.}} denote randomized/approximate equality. 
\\
\textbf{\textsuperscript{a}}Hidden within a public padded bound. \textbf{\textsuperscript{b}}Native only for one slot-bounded target. Longer histories require fragmentation and aggregation.
\end{minipage}
\end{table}

\noindent \bnumS{1.6pt}{\normalsize}{3} \textbf{\DeepUL{Homomorphic Matching.}} Prior HE-based constructions employ several distinct techniques. Unlike the pairing and FE-based schemes that generate overlapping boundary instances~\cite{1:BCS21,1:BCC20,1:BCS23}, the HE design of~\cite{1} uploads non-overlapping fragments and forms adjacent pairs online for cross-boundary matches. Its case-sensitive exact path uses one trapdoor ciphertext, whereas case-insensitive wildcard matching uses two. Excluding pair formation, this variant requires two ciphertext--ciphertext (ct--ct) multiplications, two ciphertext--plaintext (ct--pt) multiplications, and two additions or subtractions per evaluation. While the fragment-length bound in~\cite{1} leaves at least half of each uploaded ciphertext's packing width unused, the constant-depth BGV construction~\cite{bonte2020homomorphic} packs overlapping chunks across slots and applies randomized equality. However, its representation and evaluation process disclose the private query length, and matching remains pattern-length dependent through repeated equality circuits involving multiple multiplications, rotations, and Frobenius maps. Its randomized equality also has a one-sided false-positive probability that accumulates across tested positions. These two constructions further expose richer outputs to the decryptor, namely per-offset correlation scores and locations in~\cite{1} and occurrence locations and counts in~\cite{bonte2020homomorphic}. 

\texttt{CipherMatch}~\cite{7} instead reduces arithmetic cost by packing multiple bits per BFV coefficient and performing addition-only exact equality with shifted encrypted queries. This approach requires multiple query ciphertexts, and its query-dependent shifts can reveal the query length unless padded. Moreover, the TFHE-based search~\cite{narisada2026tfhe} performs binary search over an encrypted suffix array constructed from the complete plaintext, preventing direct support for independently encrypted, incrementally arriving messages. Its proposed encoding uses four LWE ciphertexts per character, while bootstrapping and encrypted lookups scale with the padded pattern length, increasing runtime and memory. CKKS-based comparison~\cite{secoaguirre2026text} instead examines every candidate alignment using rotations and approximate polynomial binarization. It explicitly reveals the pattern length, while near matches can cause false positives unless higher-degree, deeper circuits are used, increasing runtime and memory. Slot-bounded targets also require fragmentation and aggregation for longer histories. Except for the wildcard path of~\cite{1}, these schemes lack native case-insensitive matching, while corpus-wide presence-only release can require further query padding, boundary-aware fragmentation, cross-ciphertext aggregation, result obfuscation, or private aggregation. We retain~\cite{1} as the primary quantitative baseline and compare the remaining HE constructions qualitatively in Appendix~\ref{ComparisonWithHE}.

\smallskip
\noindent \ding{42} \textbf{\DeepUL{Our Contributions.}}  We present \ProtocolName{}, a privacy-preserving keyword-search protocol for practical E2EE messaging systems. Table~\ref{tab:seek-compare-intro} qualitatively compares prior encrypted search paradigms against \ProtocolName{} in our E2EE messaging system. We summarize  our contributions below:

\begin{enumerate}[leftmargin=*]

    \item We introduce a byte-aligned fragmentation strategy that combines \emph{minimal sufficient overlap} with full-width homomorphic packing, reducing the fragment count by up to two orders of magnitude relative to prior designs~\cite{1:BCC20,1} and thereby lowering the $\mathsf{CSP}$-side ciphertext footprint, encryption, and communication overheads.

    \item We design \ProtocolName{}, an ASCII case-insensitive keyword-search protocol using a \emph{single} fixed-size encrypted trapdoor (unlike~\cite{1,7}), and compute the complete within-fragment correlation vector using a \emph{single} query-dependent homomorphic multiplication per fragment (unlike~\cite{1}). To obtain the final search results, $\mathsf{PP}$ and $\mathsf{CSP}$ jointly perform privacy-preserving selected decoding, zero testing, and secure aggregation, revealing only the presence-or-absence bit to $\mathsf{PP}$ while $\mathsf{CSP}$ receives no result and neither party learns match counts, locations, or correlation scores.
    \item We prove the correctness of \ProtocolName{} and establish its privacy against semi-honest adversaries under standard real/ideal-world simulation paradigm. We further analyze the feasibility of exhaustive repeated query attacks against \ProtocolName{}.

    \item We implement \ProtocolName{} using \textsf{Microsoft SEAL} and integrate it into E2EE web/mobile applications through a modified \textsf{node-seal} backend. 
    Across $10{,}000$ encrypted searches over \textsf{SAMSum} messaging corpus scaled up to an estimated five-year history~\cite{gliwa-etal-2019-samsum}, \ProtocolName{} achieves $100\%$ observed accuracy and up to a $5.47\times$ correlation-computation speedup over the evaluated wildcard baseline configuration of~\cite{1}.
\end{enumerate}

\section{Preliminaries}

\subsection{BFV Homomorphic Encryption}
\label{subsec:bfv}

We build \ProtocolName{}'s encrypted-matching core on the Brakerski--Fan--Vercauteren (BFV) leveled homomorphic-encryption (HE) scheme~\cite{brakerski2012fully,fan2012somewhat}, whose security is based on the Ring-Learning-with-Errors (RLWE) assumption. An approved BFV parameter set $\mathsf{params}$ fixes the polynomial-modulus degree $N$, a prime plaintext modulus $t$, and an RNS ciphertext-modulus chain $Q=\prod_{\ell=0}^{\ell_Q-1}q_\ell$. These parameters define the plaintext ring $R_t=\mathbb Z_t[X]/(X^N+1)$ and the ciphertext ring $R_Q=\mathbb Z_Q[X]/(X^N+1)$. Key generation produces $(\mathsf{pk},\mathsf{sk},\mathsf{rlk})\leftarrow\mathsf{KeyGen}(\mathsf{params})$, where $\mathsf{pk}$ and $\mathsf{sk}$ are the public and secret keys and $\mathsf{rlk}$ is the relinearization key, used to relinearize a homomorphic product ciphertext into the standard form $c=(c_0,c_1)\in R_q^2$ at an active RNS modulus level $q\mid Q$. Our privacy analysis uses the adaptive multi-message IND-CPA model for HE, in which the adversary receives all published evaluation material, including $\mathsf{rlk}$~\cite{HomomorphicEncryptionSecurityStandard}. Encryption and decryption are denoted by $c\leftarrow\mathsf{Enc}_{\mathsf{pk}}(m)$ and $m\leftarrow\mathsf{Dec}_{\mathsf{sk}}(c)$, while homomorphic evaluation uses ct-ct operations $\mathsf{Eval}_{\mathsf{add}}$ and $\mathsf{Eval}_{\mathsf{mult}}$ and ct-pt operations $\mathsf{Eval}_{\mathsf{add\text{-}plain}}$ and $\mathsf{Eval}_{\mathsf{mult\text{-}plain}}$. For any modulus $r$, let $\operatorname{can}_r(a)\in\{0,\ldots,r-1\}$ denote the canonical integer representative of $a\in\mathbb Z_r$. For any odd modulus $r$, let $\operatorname{ctr}_r(a)\in\{-(r-1)/2,\ldots,(r-1)/2\}$ denote its centered integer representative, and $[\cdot]_t$ denote reduction modulo $t$. We define coefficient-wise BFV decoding as $\mathsf{DecCoeff}_{q,t}(a)=[\lfloor (t/q)\operatorname{ctr}_q(a)\rceil]_t$. For a relinearized ciphertext $c=(c_0,c_1)$, decryption first forms the noisy plaintext $v=c_0+c_1\mathsf{sk}\pmod q$ and applies $\mathsf{DecCoeff}_{q,t}$ independently to its coefficients. In particular, a coefficient encoding $m_i\in\mathbb Z_t$ is recovered correctly whenever $|(t/q)\operatorname{ctr}_q(v_i)-\operatorname{ctr}_t(m_i)|<1/2$. For an odd active modulus  $Q'\mid Q$, the public operation $\mathsf{ModSwitch}_{Q\rightarrow Q'}$ reduces a ciphertext to the  active modulus $Q'$ while preserving its decoded plaintext whenever the BFV correctness condition remains satisfied. Because \ProtocolName{} uses signed-coefficient arithmetic, elements of $\mathbb Z_t$ are interpreted through their centered representatives: for example, $-1$ is encoded as $t-1$.  Our deployment settings require a prime plaintext modulus $t$, satisfying $t\equiv1\pmod{2N}$, so that $\mathbb Z_t$ is a field and \ProtocolName{} supports NTT-compatible arithmetic.

\subsection{Secure Two-Party Computation (2PC)}
\label{sec:beaver}

\noindent \ding{110}
\paragraphNew{\DeepUL{Beaver Multiplication}}
For modulus $r$ and $x\in\mathbb Z_r$, we write
$\langle x\rangle_r=(x^{\mathsf{1}},x^{\mathsf{2}})$ for an additive sharing
satisfying $x=x^{\mathsf{1}}+x^{\mathsf{2}}\pmod r$. When $r=t$, we omit the subscript
and write $\langle x\rangle$. To multiply shares
$\langle x\rangle$ and $\langle y\rangle$ over $\mathbb Z_t$, the parties
consume a fresh Beaver triple
$(\langle u\rangle,\langle v\rangle,\langle w\rangle)$ satisfying
$w=uv\pmod t$~\cite{beaver1991efficient}. In the online phase, the
parties reconstruct only the masked differences $d=x-u$ and $f=y-v$ and
compute
$z^{\mathsf{1}}=w^{\mathsf{1}}+d\,v^{\mathsf{1}}+f\,u^{\mathsf{1}}+df
\pmod t$ and
$z^{\mathsf{2}}=w^{\mathsf{2}}+d\,v^{\mathsf{2}}+f\,u^{\mathsf{2}}
\pmod t$, yielding
$z^{\mathsf{1}}+z^{\mathsf{2}}=xy\pmod t$. Because $u$ and $v$ are uniform
and remain secret-shared, the opened values $d$ and $f$ reveal no information
about $x$ or $y$ in the semi-honest model. Fresh arithmetic triples over the
prime field $\mathbb Z_t$ can be generated by standard finite-field
preprocessing~\cite{keller2016mascot}. %

\noindent \ding{110}
\paragraphNew{\DeepUL{Oblivious Linear Evaluation (OLE)}}
In an OLE over $\mathbb Z_t$, a sender holding $(a,b)$ and a receiver holding $x$ securely compute $ax+b\pmod t$ for the receiver, without revealing $x$ to the sender or anything beyond the output about $(a,b)$ to the receiver~\cite{baum2020ole}. \ProtocolName{} invokes one OLE instance per retained coefficient to generate the correlated zero-test tokens, with
all instances evaluated as one offline batch.
Section~\ref{subsec:presence-zero-target} defines their
distribution, while Appendix~\ref{app:seldec-realization}
specifies the exact sender and receiver inputs, correctness
invariant, and batching.

\noindent \ding{110}
\paragraphNew{\DeepUL{Boolean Sharing and Mixed-Domain Conversion}} For a bit $b\in\{0,1\}$, we write $\llbracket b\rrbracket_{\mathsf B} =(b^{\mathsf{1}},b^{\mathsf{2}})$ for a Boolean XOR sharing satisfying $b=b^{\mathsf{1}}\oplus b^{\mathsf{2}}$. For a bitstring, this notation is applied componentwise. We use a standard semi-honest GMW protocol to evaluate Boolean circuits over these shares \cite{goldreich1987mental,demmler2015aby}: XOR gates are evaluated locally, whereas each AND gate consumes a fresh Boolean multiplication triple $(\llbracket a\rrbracket_{\mathsf B},\llbracket b\rrbracket_{\mathsf B},\llbracket c\rrbracket_{\mathsf B})$ satisfying $c=a\land b$. To convert Boolean outputs into arithmetic shares over $\mathbb Z_t$, the parties use daBits $(\llbracket\rho\rrbracket_{\mathsf B},\langle\rho\rangle)$, in which the same random bit $\rho$ 
is represented in both domains \cite{rotaru2019marbled,escudero2020mixed}. All Boolean triples and daBits used in an execution are fresh, independent, and consumed once.

\section{Protocol Architecture and Threat Model}
\label{sec:architecture-threat-model}
\label{SystemModel}
\label{subsec:ThreatModel}

Our E2EE messaging and privacy-preserving keyword-search system comprises four primary entities: a sender ($\mathsf{S}$), a receiver ($\mathsf{R}$), a cloud service provider ($\mathsf{CSP}$), and a pattern provider ($\mathsf{PP}$), as illustrated in Fig.~\ref{fig:flowdiagram}. Each party performs a distinct role in the \ProtocolName{} protocol. We outline these roles and trust assumptions associated with each entity as follows.

\ding{108} \paragraphNew{\DeepUL{Sender ($\mathsf{S}$)}} The sender encodes, fragments, and encrypts each message once using $\mathsf{R}$'s public key before uploading the resulting ciphertext fragments to $\mathsf{CSP}$. Although $\mathsf{S}$ and $\mathsf{R}$ may cooperatively exchange harmful or unlawful content, they remain compliant with the protocol. Semantic evasion through code words or out-of-band pre-encryption before the protocol encoding and HE layers remains outside the literal keyword-search guarantee, as receiver-initiated abuse-reporting protocols likewise address reported plaintext rather than deliberately concealed semantics~\cite{GrubbsLR17,issa2022hecate}.

\ding{108} \paragraphNew{\DeepUL{Receiver ($\mathsf{R}$)}} Each receiver generates the BFV keys as in Section~\ref{subsec:bfv}, retains the full secret key locally, publishes the corresponding public and evaluation keys, and decrypts the ciphertext fragments forwarded by $\mathsf{CSP}$. These fragments are retained by $\mathsf{CSP}$ for search, making the delivered ciphertexts the sole searchable representation and eliminating the need for a separate upload. Because $\mathsf{R}$ possesses the full secret key, it has the capability to decrypt the encrypted correlations produced by \ProtocolName{}. Returning these correlations to an implicated $\mathsf{R}$ would reveal query-related scores and locations, allowing it to suppress or misreport the search result. Consequently, both endpoints remain outside the online search process, while $\mathsf{PP}$ and $\mathsf{CSP}$ each hold additive shares of $\mathsf{R}$'s full secret key and use them to generate their respective local selected-decoding shares. The privacy-motivated $\mathsf{R}$ is assumed not to disclose its key or collude with either $\mathsf{PP}$ or $\mathsf{CSP}$, as such actions would fundamentally compromise its own message privacy. Nevertheless, it is essential to validate the functional consistency of the key artifacts and additive shares supplied by $\mathsf{R}$, as inconsistencies could lead to missed detections or biased search outcomes. \ProtocolName{} therefore performs randomized functional-validation checks during registration and before the search pipeline, as detailed in Section~\ref{subsec:functional-validation}, thereby supporting correct and unbiased subsequent search outcomes.

\ding{108} \paragraphNew{\DeepUL{Pattern Provider ($\mathsf{PP}$)}} For an authorized query, $\mathsf{PP}$ encodes the keyword as an encrypted trapdoor and submits it to $\mathsf{CSP}$. 
After $\mathsf{CSP}$ computes the encrypted correlations, $\mathsf{PP}$ uses its secret-key share for selected decoding and jointly performs zero testing and secure aggregation with $\mathsf{CSP}$. Only the presence-or-absence bit is revealed to $\mathsf{PP}$. Beyond its query, it learns no plaintext, match count, location, or correlation score. Locations are withheld as repeated single-character queries (e.g., \texttt{a}--\texttt{z}) could reveal sufficient positional structure to reconstruct substantial message content. Although presence-only disclosure still allows dictionary inference through repeated queries, Section~\ref{sec:eval_repeated_queries} demonstrates the practical infeasibility of such exhaustive repeated-query attacks.

\ding{108} \paragraphNew{\DeepUL{Cloud Service Provider ($\mathsf{CSP}$)}} Upon receiving the trapdoor, $\mathsf{CSP}$ computes correlations over the stored ciphertext fragments and uses its secret-key share to generate its local selected-decoding contribution. In \ProtocolName{}, $\mathsf{CSP}$ and $\mathsf{PP}$ are modeled as static, non-colluding, semi-honest adversaries that adhere to the protocol while attempting to infer information from their individual views. The non-collusion assumption models the organizational separation and potentially divergent incentives between an investigator or authorizing body and the messaging provider, as illustrated by the Facebook Messenger wiretap dispute and WhatsApp's traceability challenge in India~\cite{N41,N42,N25}. Under this model, $\mathsf{CSP}$ observes its local setup state, ciphertexts, trapdoors, intermediate values, and public metadata but learns no additional information about the plaintext messages, query, its length, or the result.

\ProtocolName{} targets efficient encrypted messaging and privacy-preserving keyword search, while message authentication, integrity, replay protection, ratcheting, and post-compromise security must be provided by a complementary messaging security layer that adheres to standard key evolution and lifecycle practices~\cite{N47,N48}. Setup artifacts are subject to deployment-defined rotation, with retained historical shares enabling retrospective search. Legal approval, query admission, signed audit and transparency records~\cite{N43,N44}, and per-authorization budgets~\cite{N45,N46} function as external controls. An \emph{authorized query} therefore refers to a search approved by the appropriate legal authority and admitted under these controls.

\begin{figure*}[!t]
    \centering
    \includegraphics[width=\linewidth]{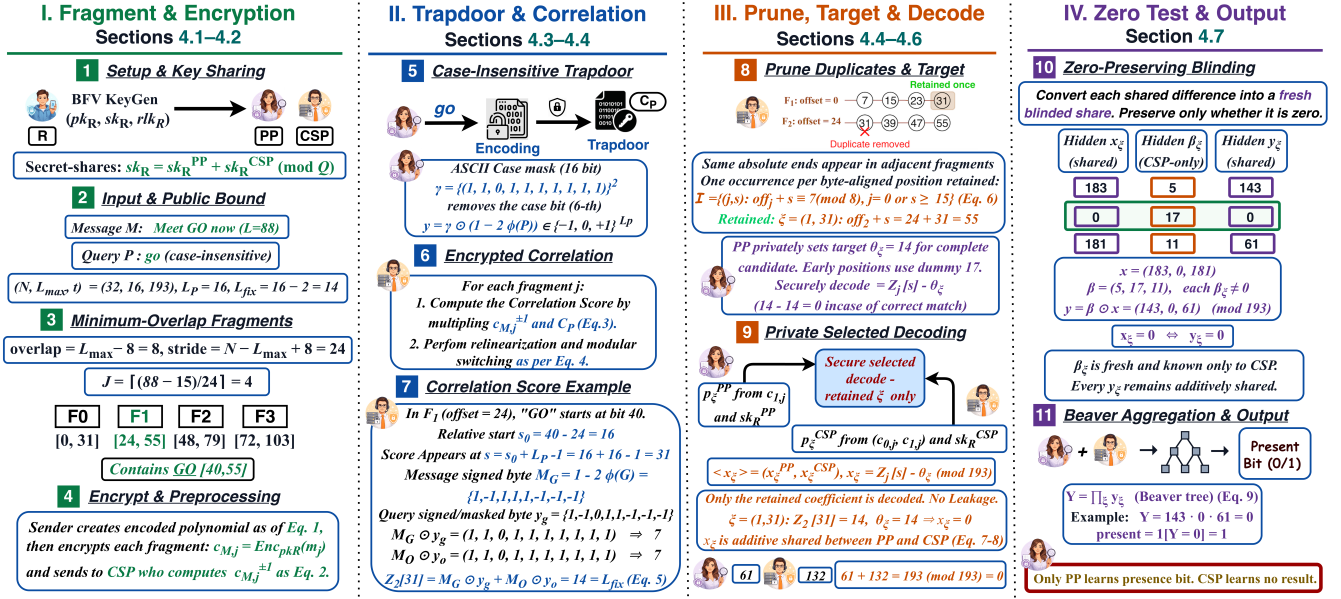}
    \caption{End-to-end overview of \ProtocolName{} with a running example under illustrative parameter settings.}
    \label{fig:Overview}
\end{figure*}

\section{\ProtocolName{}: Privacy-Preserving Encrypted Messaging and Keyword Search Protocol}
\label{sec:seek-presence-protocol}

\noindent \ding{110} \textbf{\DeepUL{Overview.}} In this section, we present \ProtocolName{}, a privacy-preserving case-insensitive keyword-search protocol for E2EE messaging systems. In \ProtocolName{}, the sender $\mathsf{S}$ encodes each message into overlapping fragments, encrypts, and uploads them to $\mathsf{CSP}$ (Section \ref{subsec:fragEncrypt}). The $\mathsf{PP}$ then submits a single encrypted trapdoor (Section \ref{subsec:presence-trapdoor}), which $\mathsf{CSP}$ evaluates against the query-independently preprocessed fragments to obtain encrypted correlation polynomials (Section \ref{subsec:presence-correlation}). After publicly removing overlap-induced duplicate positions, $\mathsf{PP}$ and $\mathsf{CSP}$ privately encode the match targets (Section~\ref{subsec:presence-zero-target}), decode only the retained coefficients into additive shares (Section~\ref{subsec:presence-selected-decoding}), transform them into blinded zero-test values (Section~\ref{subsec:presence-aggregation}), and aggregate them using the Beaver-multiplication protocol (Section \ref{sec:beaver}).  Only one presence-or-absence bit is revealed to $\mathsf{PP}$, while $\mathsf{CSP}$ learns no search result. Fig.~\ref{fig:Overview} illustrates the E2E \ProtocolName{} workflow using a running example.

\subsection{Cryptographic Setup Phase}
\label{subsec:setup}

For each key-lifecycle epoch $e$, \ProtocolName{} first selects the BFV parameters $\mathsf{params}_e=(N,Q,t)$. During the setup phase, $\mathsf{R}$ generates $(\mathsf{pk}_{R,e},\mathsf{sk}_{R,e},\mathsf{rlk}_{R,e}) \leftarrow \mathsf{KeyGen}(\mathsf{params}_e),$ as defined in Section~\ref{subsec:bfv}. These parameters and keys are reused for all uploads and authorized searches within the epoch and replaced only upon deliberate rotation. The receiver $\mathsf{R}$ retains $\mathsf{sk}_{R,e}$ for E2EE message decryption, distributes $(\mathsf{params}_e,\mathsf{pk}_{R,e})$ to $\mathsf{S}$, and $(\mathsf{params}_e,\mathsf{pk}_{R,e},\mathsf{rlk}_{R,e})$ to $\mathsf{PP}$ and $\mathsf{CSP}$. 
Afterwards, $\mathsf{PP}$ samples $\mathsf{sk}_{R,e}^{\mathsf{PP}}\gets R_Q$ uniformly and sends it to $\mathsf R$, who derives $\mathsf{sk}_{R,e}^{\mathsf{CSP}}=\mathsf{sk}_{R,e}-\mathsf{sk}_{R,e}^{\mathsf{PP}}\pmod Q$ and sends this share to $\mathsf{CSP}$ (Step~1 of Fig.~\ref{fig:Overview}). The resulting 2-out-of-2 additive shares reconstruct $\mathsf{sk}_{R,e}$, while each is marginally uniform and statistically independent of the fixed BFV key tuple. 
By the RNS representation of Section~\ref{subsec:bfv}, the same additive sharing relation holds independently coefficient-wise in every modulus limb and remains valid after both parties restrict their shares to any reduced modulus level $Q^{\prime}\mid Q$. Thus, for a relinearized ciphertext $c=(c_0,c_1)\in R_q^2$ at $q\in\{Q,Q^{\prime}\}$:
\begin{equation}
\begin{aligned}
c_0+c_1\mathsf{sk}_{R,e}
=&
\underbrace{
c_1\mathsf{sk}_{R,e}^{\mathsf{PP}}
}_{\substack{\mathsf{PP}\text{-side decryption}}}
+\,\,\,\,\,
\underbrace{
(c_0+c_1\mathsf{sk}_{R,e}^{\mathsf{CSP}})
}_{\substack{\mathsf{CSP}\text{-side decryption}}}
\pmod q
\nonumber
\end{aligned}
\end{equation}
This additive decomposition follows the multiparty BFV paradigm~\cite{mouchet2021multiparty}. $\mathsf{PP}$ and $\mathsf{CSP}$ locally compute their decryption-phase contributions from its secret-key share and supplies them as private inputs to the selected-decoding functionality of Section~\ref{subsec:presence-selected-decoding}, realized using 2PC primitives of Section~\ref{sec:beaver}.

\subsection{Fragmentation and Encryption}
\label{subsec:fragEncrypt}

Let $\Sigma$ denote the messaging alphabet, and $\phi:\Sigma\rightarrow\{0,1\}^{8}$ be a fixed public 8-bit encoding. A plaintext message 
sent from $\mathsf{S}$ to $\mathsf{R}$ is represented as $M_{\mathsf{S}\rightarrow\mathsf{R}}$
and encoded as the bitstream $\mathbf{m}_{\mathsf{S}\rightarrow\mathsf{R}} =\phi(M_{\mathsf{S}\rightarrow\mathsf{R}}) \in\{0,1\}^{L}$.
For a public byte-aligned maximum pattern length
$L_{\max}\in 8\mathbb{N}$ with $8\leq L_{\max}\leq N$, \ProtocolName{}
partitions the bitstream into zero-padded length-$N$ fragments
using the minimal sufficient overlap of $L_{\max}-8$ bits.
Accordingly, the fragment step size is defined as $\mathsf{stride}=N-L_{\max}+8$, and the number of fragments is $\mathcal{J}=\max\{1,\lceil(L-(L_{\max}-8))/(N-L_{\max}+8)\rceil\}$. 
The overlap $L_{\max}-8$ is the minimal sufficient overlap that guarantees
containment of every byte-aligned encoded-keyword window with
$L_P\in8\mathbb{N}$ and $L_P\leq L_{\max}$ (see Lemma \ref{lem:presence-fragment-coverage}).
For each fragment index $j\in\{0,\ldots,\mathcal{J}-1\}$, the starting offset is defined as: $\mathsf{off}_j=j\cdot\mathsf{stride}$ and the logical fragment length $T_j=\min\{N,L-\mathsf{off}_j\}$.  The corresponding padded fragment coefficients $b^{(j)}_i$ are defined as of Eq. \ref{fragmentCoeff} and the resulting coefficient-encoded plaintext polynomial is computed as: $m_j(X)=\sum_{i=0}^{N-1}b_i^{(j)}X^i\in R_t$.
\begin{equation}
b_i^{(j)}
=
\mathbf{1}[i<T_j]\,
\mathbf m_{\mathsf S\rightarrow\mathsf R}[\mathsf{off}_j+i],
\,\, \text{for}\,\, 0\leq i<N
\label{fragmentCoeff}
\end{equation}
The sender then computes $c_{M,j}\leftarrow \mathsf{Enc}_{\mathsf{pk}_{R,e}}(m_j(X))$ and uploads $c_{M,j}$ with the public metadata $(j,\mathsf{off}_j,T_j)$ to $\mathsf{CSP}$.

\noindent \ding{110}
\paragraphNew{\DeepUL{Preprocessing at $\mathsf{CSP}$}} To reduce online latency, $\mathsf{CSP}$ transforms $0/1$-encoded ciphertexts into a $\{\pm 1, 0\}$ representation as a one-time query-independent preprocessing step. For each fragment $j$, $\mathsf{CSP}$ constructs a public plaintext mask $\mathsf{mask}_j(X) = \sum_{i=0}^{{T}_j-1} X^i$, whose coefficients are $1$ in positions $[0,{T}_j-1]$ and $0$ elsewhere. The affine map $b \mapsto 1 - 2b$ is evaluated homomorphically on the stored ciphertext as follows.
\begin{equation}
c^{\pm1}_{M,j} \leftarrow \mathsf{Eval}_{\textsf{add-plain}}( \mathsf{Eval}_{\textsf{mult-plain}}(c_{M,j}, -2),\mathsf{mask}_j(X)) 
\label{preprocessing}
\end{equation}
By construction, the resulting ciphertext $c^{\pm1}_{M,j}$
maps message bits $0/1 \mapsto +1/-1$, and positions outside the logical fragment map to $0$. This preprocessing requires only ct–pt multiplications and additions per fragment, and can be performed asynchronously once before any search evaluation.

\smallskip
\noindent\ding{42} \paragraphNew{\DeepUL{Running Example}} Steps~2-4 of Fig.~\ref{fig:Overview} illustrate fragmentation using message \textsf{``Meet GO now''}. With $L=88$, $N=32$, and $L_{\max}=16$, the 24-bit stride produces $4$ overlapping fragments with one-byte overlap, and the bits encoding \textsf{``GO''} are fully contained in $F_1$, while final fragment is padded to $N$ bits.

\subsection{Case-Insensitive Trapdoor Generation}
\label{subsec:presence-trapdoor}

Capitalization differences are common in messaging and can cause false negatives under exact case-sensitive search. Accordingly, the exact-equality variants in~\cite{1,7} fail to match variants such as \texttt{security}, \texttt{Security}, and \texttt{SECURITY} using a single query. 
Although the wildcard mechanism in~\cite{1} supports case-insensitive matching, it requires two trapdoor ciphertexts and, including adjacent-pair construction, two ct--ct multiplications, three ct--pt multiplications, and three ciphertext additions/subtractions per pair. 
\ProtocolName{} instead realizes case-insensitive matching using one signed-correlation trapdoor ciphertext and one query-dependent ct--ct multiplication per fragment. 
For every authorized search, $\mathsf{PP}$ decides on a secret keyword $P$ of length $\ell_P$ and encodes it as $\mathbf p=\phi(P)\in\{0,1\}^{L_P}$, where $L_P=8\ell_P\leq L_{\max}\leq N$ and $t>4L_{\max}+2$. Let $\mathcal A_{\mathsf{ASCII}}=\{\texttt{A},\ldots,\texttt{Z},\texttt{a},\ldots,\texttt{z}\}$ represent the ASCII alphabetic characters. Because uppercase and lowercase ASCII letters differ only at the $6^{\text{th}}$ LSB bit position $k_{\mathsf{case}}$ (e.g., $\phi(\texttt{A})=01000001$ vs. $\phi(\texttt{a})=01100001$), the pattern provider masks only this case bit for alphabetic characters while keeping non-alphabetic characters fully constrained. For each character position $r\in\{0,\ldots,\ell_P-1\}$ and bit position $k\in\{0,\ldots,7\}$, the flattened bit index is $i=8r+k$. $\mathsf{PP}$ assigns $\gamma_i=0$ when $P[r]\in\mathcal A_{\mathsf{ASCII}}$ and $k=k_{\mathsf{case}}$, and sets $\gamma_i=1$ otherwise. It then computes the signed query coefficients $y_i=\gamma_i(1-2\mathbf p[i])\in\{-1,0,+1\}$. Let $A_P=|\{r:P[r]\in\mathcal A_{\mathsf{ASCII}}\}|$ denote the number of alphabetic characters in $P$. After ignoring one case bit per letter, the number of constrained query bits is $L_{\mathsf{fix}}=L_P-A_P=7\ell_P$, when $P$ consists entirely of ASCII letters.
$\mathsf{PP}$ then forms the reversed polynomial $P_{\mathsf{CI}}(X)=\sum_{i=0}^{L_P-1}\widetilde y_i X^{L_P-1-i}\in R_t$, where $\widetilde y_i$ is the representative of $y_i$ in $\mathbb Z_t$ (Section~\ref{subsec:bfv}). The reversal causes the aligned signed-bit products of each candidate window to accumulate in one correlation coefficient \cite{yasuda2013secure,yasuda2014practical}. Finally, $\mathsf{PP}$ computes the trapdoor $c_P\leftarrow\mathsf{Enc}_{\mathsf{pk}_{R,e}} (P_{\mathsf{CI}}(X))$ and sends it to $\mathsf{CSP}$. Coefficients beyond the encoded query are zero, so every trapdoor occupies one fixed-size ciphertext regardless of $L_P$.

\smallskip
\noindent\ding{42} \paragraphNew{\DeepUL{Running Example}} Step~5 of Fig.~\ref{fig:Overview} encodes \textsf{``go''} as a reversed signed trapdoor. Ignoring two ASCII case bits gives $L_{\mathrm{fix}}=14$, enabling the trapdoor to match \textsf{``GO''}.

\subsection{Encrypted Correlation Computation} %
\label{subsec:presence-correlation}

For each trapdoor $c_P$ and preprocessed ciphertext fragment $c_{M,j}^{\pm1}$ from Eq.~\ref{preprocessing}, where $j\in\{0,\ldots,\mathcal J-1\}$, $\mathsf{CSP}$ computes the encrypted correlation, relinearizes it using $\mathsf{rlk}_{R,e}$, and switches it to the reduced coefficient modulus $Q'$ (refer Section~\ref{subsec:bfv}):
\begin{align}
\widehat c_{\mathsf{corr},j}
&\leftarrow
\mathsf{Eval}_{\mathsf{mult}}(c_{M,j}^{\pm1},c_P)
\label{eq:presence-correlation-ciphertext}
\\
c_{\mathsf{corr},j}^{Q'}
&\leftarrow
\mathsf{ModSwitch}_{Q\rightarrow Q'}(
\mathsf{Relin}_{\mathsf{rlk}_{R,e}}
(\widehat c_{\mathsf{corr},j}))
\label{eq:presence-qprime-correlation}
\end{align}
The chosen $Q'$ ensures coefficient-wise BFV decoding correctness 
for every message, query, and result. A complete byte-aligned candidate window in fragment $j$ begins at $s_0$ satisfying $0\le s_0\le T_j-L_P$ and $(\mathsf{off}_j+s_0)\bmod 8=0$. Because $P_{\mathsf{CI}}(X)$ reverses the query polynomial, its correlation score occurs at $s=s_0+L_P-1$. We denote the candidate window by $\mathsf{Win}_{j,s_0}=(b_{s_0}^{(j)},\ldots,b_{s_0+L_P-1}^{(j)})$ and define the  Hamming distance between $\mathbf a,\mathbf b\in\{0,1\}^{L_P}$ as $\mathsf{Ham}_{\gamma}(\mathbf a,\mathbf b)=\sum_{i=0}^{L_P-1}\gamma_i(\mathbf a[i]\oplus\mathbf b[i])$. 
Since $s\geq L_P-1$ and the unreduced polynomial product has degree at most $N+L_P-2$, no term of degree $s+N$ exists. Therefore, reduction modulo $X^N+1$ introduces no negacyclic contribution at coefficient $s$. Its plaintext correlation value is represented as:
\begin{equation}
z_j[s]
=
\sum_{i=0}^{L_P-1}
y_i(1-2b_{s_0+i}^{(j)})
=
L_{\mathsf{fix}}
-
2\mathsf{Ham}_{\gamma}
(\mathsf{Win}_{j,s_0},\mathbf p)
\label{eq:presence-correlation-semantics}
\end{equation}
Thus, $z_j[s]=L_{\mathsf{fix}}$ iff the candidate window matches $P$ under the case-insensitive predicate. Coefficients with $s<L_P-1$ represent no complete window and may contain negacyclic wraparound, while those outside the logical fragment length $T_j$ correspond to no message position. Because $L_P$ is private, retained early coefficients are not removed using it and instead they receive private dummy targets as in Section~\ref{subsec:presence-zero-target}.

\smallskip
\noindent \ding{110}
\paragraphNew{\DeepUL{Overlap-Duplicate Pruning}}
Fragment overlap guarantees complete-window coverage for byte-aligned candidates, but creates duplicate positions across adjacent correlation polynomials. 
After modulus switching, $\mathsf{CSP}$ removes these duplicates. Since coefficient $s$ of fragment $j$ corresponds to absolute window-end position $\mathsf{off}_j+s$, the parties retain:
\begin{equation}
\begin{aligned}
\mathcal I&=\{\xi=(j,s):\;
0\leq j<\mathcal J,
0\leq s<T_j,\\
&
\underbrace{\mathsf{off}_j+s\equiv7\pmod 8}
_{\text{byte-aligned end position}},
\underbrace{j=0\ \lor\ s\geq L_{\max}-1}
_{\text{exclude repeated overlap positions}}
\}
\end{aligned}
\label{eq:presence-open-list}
\end{equation}
The congruence condition retains byte-aligned absolute end positions, while the final condition removes the initial overlap of each later fragment because those positions occur in its predecessor.
Because $\mathcal I$ depends only on public geometry and $L_{\max}$, it reveals neither the query nor the result. For each fragment contributing to $\mathcal I$, $\mathsf{CSP}$ retains $c_{0,j}^{Q'}$ and sends $(c_{1,j}^{Q'},j,\mathsf{off}_j,T_j,e)$ to $\mathsf{PP}$. This allows $\mathsf{PP}$ to form its local decryption shares for every retained coefficient but not to decode independently, since $c_{0,j}^{Q'}$ and $\mathsf{sk}_{R,e}^{\mathsf{CSP}}$ remain with $\mathsf{CSP}$.

\smallskip
\noindent\ding{42} \paragraphNew{\DeepUL{Running Example}} Steps~6--8 of Fig.~\ref{fig:Overview} show the resulting correlation and public pruning. In $F_1$, the match beginning at relative bit offset $16$ produces score $L_{\mathrm{fix}}=14$ at coefficient $s=31$. Because the overlap causes absolute window-end position $31$ to appear in both $F_1$ and $F_2$, pruning removes the duplicate from $F_2$ and retains the occurrence from $F_1$.

\subsection{Private Zero-Target Encoding and Offline Zero-Test Tokens}
\label{subsec:presence-zero-target}

After overlap pruning, \ProtocolName{} retains public coefficient indices $\mathcal I$, but the private length $L_P$ determines which coefficients represent complete query windows. To hide this distinction from $\mathsf{CSP}$, $\mathsf{PP}$ privately assigns a target $\theta_\xi$ to each $\xi=(j,s)\in\mathcal I$ and defines the corresponding zero-test input $x_\xi$ as:
\[
\theta_\xi
=
\begin{cases}
L_{\mathsf{fix}}, & \!\!\!\!\!\!s\geq L_P-1\\
L_{\max}+1,       & \!\!\!\!\!\!s<L_P-1
\end{cases}
\,\, \text{and} \,\,
x_\xi
=
z_j[s]-\theta_\xi
\!\!\!\!\!\pmod t
\]
For complete candidates, $x_\xi=0$ exactly when the candidate matches $P$ under the case-insensitive predicate. If $s<L_P-1$, the coefficient cannot represent a complete private query window with $|z_j[s]|\leq L_{\mathsf{fix}}\leq L_{\max}$, so the dummy target $L_{\max}+1$ cannot yield zero under the no-wrap condition (Lemma~\ref{lem:presence-no-wrap}). As neither $\theta_\xi$ nor $x_\xi$ is disclosed, $\mathsf{CSP}$ learns neither $L_P$ nor which retained coefficients represent complete windows.

To determine whether any $x_\xi=0$ without reconstructing individual values, \ProtocolName{} constructs one fresh one-sided zero-test token per $\xi\in\mathcal I$ using the OLE primitive of Section~\ref{sec:beaver}~\cite{baum2020ole}. In input-independent preprocessing, $\mathsf{PP}$ and $\mathsf{CSP}$ sample:
$$\mathsf{PP}:\,\alpha_\xi^{\mathsf{PP}}\gets\mathbb Z_t, \,\, \mathsf{CSP}:\,\alpha_\xi^{\mathsf{CSP}},\mu_\xi^{\mathsf{CSP}} \gets\mathbb Z_t, \beta_\xi\gets\mathbb Z_t^*$$

For each $\xi$, $\mathsf{CSP}$ is the OLE sender with $(a_\xi^{\mathsf{OLE}},b_\xi^{\mathsf{OLE}}) = (\beta_\xi,\,\beta_\xi\alpha_\xi^{\mathsf{CSP}}-\mu_\xi^{\mathsf{CSP}})$, while $\mathsf{PP}$ is the receiver with $x_\xi^{\mathsf{OLE}}=\alpha_\xi^{\mathsf{PP}}$. The OLE returns only $\mu_\xi^{\mathsf{PP}}$
to $\mathsf{PP}$, such that:
$$\mu_\xi^{\mathsf{PP}} = a_\xi^{\mathsf{OLE}}x_\xi^{\mathsf{OLE}} +b_\xi^{\mathsf{OLE}}= \beta_\xi \bigl(\alpha_\xi^{\mathsf{PP}}+\alpha_\xi^{\mathsf{CSP}}\bigr) -\mu_\xi^{\mathsf{CSP}}$$
Thus, $\mathsf{PP}$ holds $(\alpha_\xi^{\mathsf{PP}},\mu_\xi^{\mathsf{PP}})$, and $\mathsf{CSP}$ holds $(\alpha_\xi^{\mathsf{CSP}},\mu_\xi^{\mathsf{CSP}},\beta_\xi)$.

The token's base values are sampled freshly and independently across candidates and executions and independently of the corpus, query, result, and other preprocessing instances, and components within a token are correlated only by this invariant. In the same execution, $\alpha_\xi^{\mathsf{PP}}$ also masks the selected-decoding output, enabling Section~\ref{subsec:presence-aggregation} to convert shares of $x_\xi$ into shares of $\beta_\xi x_\xi$.  Because token allocation depends only on the public set $\mathcal I$, the $K=|\mathcal I|$ independent OLE instances are evaluated as one offline batch~\cite{baum2020ole}. Appendix~\ref{app:seldec-realization} specifies the batch binding and one-time lifecycle.

\noindent\ding{42} \paragraphNew{\DeepUL{Running Example}} Step~8 of Fig.~\ref{fig:Overview} compares each retained coefficient with a hidden target: $14=L_{\mathrm{fix}}$ for complete candidates and the dummy target $L_{\max}+1=17$ for early partial-overlap coefficients. Thus, the coefficient aligned with \textsf{``GO''} in $F_1$ yields a zero difference, while early coefficients remain nonzero, hiding target selection and match location.

\subsection{Privacy-Preserving Selected Decoding}
\label{subsec:presence-selected-decoding}

For each retained index $\xi=(j,s)$, the parties derive additive shares of $x_\xi=z_j[s]-\theta_\xi\pmod t$ while ensuring that neither reconstructs $z_j[s]$ or $x_\xi$, $\theta_\xi$ remains hidden from $\mathsf{CSP}$, and both local BFV decryption-phase contributions remain private. $\mathsf{PP}$ masks $\theta_\xi$ using the fresh token mask $\alpha_\xi^{\mathsf{PP}}$ from Section~\ref{subsec:presence-zero-target} by setting $a_\xi=-\theta_\xi-\alpha_\xi^{\mathsf{PP}}\pmod t$. Following the multiparty-BFV paradigm~\cite{mouchet2021multiparty}, the parties locally compute:
$$
\begin{aligned}
p_\xi^{\mathsf{PP}}
&=
[
c_{1,j}^{Q'}
(
\mathsf{sk}_{R,e}^{\mathsf{PP}}\bmod Q'
)
]_s
\pmod{Q'}\\
p_\xi^{\mathsf{CSP}}
&=
[
c_{0,j}^{Q'}
+
c_{1,j}^{Q'}
(
\mathsf{sk}_{R,e}^{\mathsf{CSP}}\bmod Q'
)
]_s
\pmod{Q'}
\end{aligned}
$$

\noindent\ding{110} \textbf{\DeepUL{Batched selected-decoding functionality.}} For public moduli $q,t$, 
and public index set $\mathcal I$, the functionality $\mathcal F_{\mathsf{SD}}^{q,t}$ receives $(p_\xi^{\mathsf{PP}},a_\xi)$ from $\mathsf{PP}$ and $p_\xi^{\mathsf{CSP}}$ from $\mathsf{CSP}$ and computes:
\begin{equation}
\begin{aligned}
v_\xi
&=
p_\xi^{\mathsf{PP}}+p_\xi^{\mathsf{CSP}}
\pmod q,\,\, \forall\,\xi\in\mathcal I\\
w_\xi
&=
\mathsf{DecCoeff}_{q,t}(v_\xi)+a_\xi
\pmod t,\,\, \forall\,\xi\in\mathcal I
\label{eq:presence-masked-decoding1}
\end{aligned}
\end{equation}
It returns only $(w_\xi)_{\xi\in\mathcal I}$ to $\mathsf{CSP}$, revealing neither $v_\xi$ nor the other party's local contribution. For $q=Q'$, we suppress the fixed moduli $(Q',t)$ from the 
notation. Let $\Pi_{\mathsf{SD}}^{\mathsf{pre}}$ and
$\Pi_{\mathsf{SD}}^{\mathsf{on}}$ denote the 2PC joint-mask preprocessing and online selected-decoding protocols, respectively, which jointly realize $\mathcal F_{\mathsf{SD}}$. 
Appendix~\ref{app:seldec-realization} gives their concrete composition using standard GMW and mixed-domain conversion primitives~\cite{goldreich1987mental,demmler2015aby,rotaru2019marbled,escudero2020mixed}, following the MPC threshold-FHE decryption paradigm of~\cite{zyskind2025threshold}, and proves its correctness and semi-honest security under the stated preprocessing assumptions. Under BFV decoding correctness at $Q'$, $\mathsf{DecCoeff}_{Q',t}(v_\xi)=z_j[s]\pmod t$. Hence, the value returned to $\mathsf{CSP}$ satisfies:
\begin{equation}
\begin{aligned}
w_\xi
&=
\mathsf{DecCoeff}_{Q',t}(v_\xi)+a_\xi= x_\xi-\alpha_\xi^{\mathsf{PP}} \pmod t
\end{aligned}
\label{eq:presence-masked-decoding}
\end{equation}
Therefore, the assignments $x_\xi^{\mathsf{PP}}=\alpha_\xi^{\mathsf{PP}}$ and $x_\xi^{\mathsf{CSP}}=w_\xi$ satisfy $x_\xi^{\mathsf{PP}}+x_\xi^{\mathsf{CSP}}=x_\xi\pmod t$. Because $\alpha_\xi^{\mathsf{PP}}$ is fresh, uniform, hidden from $\mathsf{CSP}$, and consumed once, $w_\xi$ is uniform from $\mathsf{CSP}$'s view for every $x_\xi$.
Thus, neither party reconstructs an individual correlation coefficient or zero-test input, and all indices in $\mathcal I$ can be processed as one parallel batch.

\noindent\ding{42} \paragraphNew{\DeepUL{Running Example}} Step~9 of Fig.~\ref{fig:Overview} converts the selected correlation coefficient into additive shares of its hidden target difference. Here, $61+132=193\equiv0\pmod{193}$, while neither $\mathsf{PP}$ nor $\mathsf{CSP}$ reconstructs the coefficient individually.

\subsection{One-Sided Zero Testing and Aggregation}
\label{subsec:presence-aggregation}

After selected decoding, $\mathsf{PP}$ and $\mathsf{CSP}$ hold additive shares $x_\xi^{\mathsf{PP}}=\alpha_\xi^{\mathsf{PP}}$ and $x_\xi^{\mathsf{CSP}}=w_\xi$ of each hidden zero-test input $x_\xi$. For every $\xi\in\mathcal I$, they consume its fresh offline token and set $y_\xi^{\mathsf{PP}}=\mu_\xi^{\mathsf{PP}}$ and $y_\xi^{\mathsf{CSP}}=\mu_\xi^{\mathsf{CSP}}+\beta_\xi(x_\xi^{\mathsf{CSP}}-\alpha_\xi^{\mathsf{CSP}})\pmod t$. The token invariant from Section~\ref{subsec:presence-zero-target} gives $y_\xi^{\mathsf{PP}}+y_\xi^{\mathsf{CSP}}=\beta_\xi x_\xi\pmod t$. Because each fresh $\beta_\xi\in\mathbb Z_t^*$ is nonzero and known only to $\mathsf{CSP}$, the shared value $y_\xi$ is zero exactly when $x_\xi$ is zero, without revealing the individual outcome.

For $K=|\mathcal I|$, if $K=0$, the parties set $Y=1$ and skip multiplication. Otherwise, they multiply the shares $\langle y_\xi\rangle$ using a balanced binary tree, consuming $K-1$ fresh triples and $\lceil\log_2K\rceil$ online rounds, producing additive shares of:
\begin{equation}
Y = \prod_{\xi\in\mathcal I}y_\xi = \prod_{\xi\in\mathcal I}\beta_\xi\,\,\, \prod_{\xi\in\mathcal I}x_\xi \pmod t
\label{eq:presence-aggregate-product}
\end{equation}
Only $Y$ is reconstructed to $\mathsf{PP}$, which outputs $\mathsf{present}=\mathbf{1}[Y=0]$. Since $\mathbb Z_t$ is a field and every $\beta_\xi\neq0$, $Y=0$ iff some retained candidate has $x_\xi=0$, equivalently, a keyword match exists. Lemma~\ref{lem:presence-output-privacy} establishes the distribution of $Y$ conditioned on $\mathsf{PP}$'s preceding view. Hence, aggregation reveals no match count, score, or location beyond the prescribed leakage, and presence bit, while $\mathsf{CSP}$ receives no result.

\noindent\ding{42}\paragraphNew{\DeepUL{Running Example}} Steps~10--11 of Fig.~\ref{fig:Overview} blind each target difference with a fresh nonzero scalar and aggregate the results. For $(143,0,61)$, the product is zero, so $\mathsf{PP}$ receives $\mathsf{present}=1$ without learning which factor caused the match.

\subsection{Functional Setup Validation}
\label{subsec:functional-validation}

After completing the epoch setup described in Section~\ref{subsec:setup} and before enabling search, $\mathsf{PP}$ and $\mathsf{CSP}$ freeze their public artifacts, the active modulus $Q'$, and local secret-key shares. Employing a collision-resistant hash over canonical, domain-separated encodings, the parties agree on the sender-visible digest $d_e^{\mathsf S}=\mathsf H(\mathsf{receiverID}\parallel e\parallel\mathsf{params}_e\parallel L_{\max}\parallel\mathsf{pk}_{R,e})$ and the full digest $\mathsf{setupDigest}_e=\mathsf H(d_e^{\mathsf S}\parallel\mathsf{rlk}_{R,e}\parallel Q')$, aborting the protocol upon disagreement. Only after this state is fixed do the parties independently sample and exchange $\eta_{\mathsf{PP}},\eta_{\mathsf{CSP}}\gets\{0,1\}^{\lambda}$ and derive $\mathsf{seed}_{\mathsf{val}}=\eta_{\mathsf{PP}}\oplus\eta_{\mathsf{CSP}}$. A secure pseudorandom generator and bounded sampler expand the seed into $\kappa$ tuples $(m_\nu,a_\nu,b_\nu,u_\nu)$ that are computationally indistinguishable from independent samples, where $m_\nu,a_\nu,b_\nu\in R_t$ are prescribed sparse or dense test polynomials, and $u_\nu\in\{0,\ldots,N-1\}$.
For each tuple, $\mathsf{PP}$ uses fresh randomness to encrypt $(m_\nu,a_\nu,b_\nu)$ under $\mathsf{pk}_{R,e}$ as $(c_\nu,c_\nu^a,c_\nu^b)$. $\mathsf{CSP}$ computes $\overline c_\nu=\mathsf{ModSwitch}_{Q\rightarrow Q'}(c_\nu)$ and $\overline d_\nu=\mathsf{ModSwitch}_{Q\rightarrow Q'}(\mathsf{Relin}_{\mathsf{rlk}_{R,e}}(\mathsf{Eval}_{\mathsf{mult}}(c_\nu^a,c_\nu^b)))$. The protocol then invokes selected decoding from Section~\ref{subsec:presence-selected-decoding}, followed by one-sided zero testing from Section~\ref{subsec:presence-aggregation}, on two singleton instances targeting $m_\nu[u_\nu]$ in $\overline c_\nu$ and $(a_\nu b_\nu)[u_\nu]$ in $\overline d_\nu$. All $2\kappa$ instances utilize fresh validation-only selected-decoding preprocessing and one-time zero-test tokens. $\mathsf{PP}$ reconstructs the blinded singleton outputs, converts them into zero-test bits, and transmits authenticated copies to $\mathsf{CSP}$. The setup is accepted only if all $2\kappa$ bits equal one. Otherwise, $\mathsf R$'s epoch setup is rejected and flagged, and search remains disabled. By the functional correctness of selected decoding and singleton zero testing, all $2\kappa$ known-answer checks return one for a consistent frozen setup (Lemmas~\ref{lem:seldec-realization-correctness} and~\ref{lem:presence-decoding-aggregation}). Passing these checks validates the functional consistency of the frozen BFV artifacts and additive secret-key shares across joint decryption, multiplication, relinearization, modulus switching, selected decoding, and singleton zero testing. The component-level correctness results in Lemmas~\ref{lem:presence-fragment-coverage}--\ref{lem:presence-decoding-aggregation} and Theorem~\ref{thm:presence-e2e-correctness} establish that, for an accepted, consistent setup satisfying the theorem's conditions, these operations compose into a functionally correct post-setup \ProtocolName{} search. Upon acceptance, the parties define 
$\chi_e=(\mathsf{receiverID},e,d_e^{\mathsf S},\mathsf{setupDigest}_e,\mathsf{seed}_{\mathsf{val}},\kappa)$, compute $\tau_e=\mathsf H(\chi_e)$, and issue $\mathsf{cert}_e=(\chi_e,\mathsf{Sig}_{\mathsf{PP}}(\tau_e),\mathsf{Sig}_{\mathsf{CSP}}(\tau_e))$ under an EUF-CMA-secure signature scheme. Before encryption, $\mathsf S$ recomputes $d_e^{\mathsf S}$ from the sender-visible setup and verifies the certificate fields and both signatures. Each authenticated upload includes $(\mathsf{receiverID},e,\mathsf{setupDigest}_e)$, and $\mathsf{CSP}$ forwards and searches only ciphertexts matching its active validated setup. Each party stores $\mathsf{cert}_e$ together with its frozen local state, including its secret-key share, and refuses to use the certificate if that state changes. Any modification requires fresh validation and certification. Thus, a valid certificate binds every upload and search to one frozen, functionally validated setup, preventing wrong artifact execution and preserving the correctness of subsequent \ProtocolName{} searches.

\section{Correctness Analysis of \ProtocolName{}}
\label{sec:seek-correctness}

In this section, we present the component-level correctness lemmas and the E2E correctness theorem for \ProtocolName{}. Full detailed and technical proofs are deferred to Appendix~\ref{app:seek-correctness}.

\begin{lemma}[Fragment Coverage and Unique Retention]
\label{lem:presence-fragment-coverage}
Under the fixed encoding, let $L_P,L_{\max},N\in8\mathbb N$ with $L_P\le L_{\max}\le N$. With the minimal sufficient overlap $L_{\max}-8$, every complete byte-aligned length-$L_P$ window is fully contained in at least one fragment. Moreover, the pruning set $\mathcal I$ of Eq.~\ref{eq:presence-open-list} retains the absolute window-end position exactly once.
\end{lemma}

\begin{lemma}[Correlation Semantics]
\label{lem:presence-correlation-correctness}
For every complete byte-aligned candidate window $\mathsf{Win}_{j,s_0}$ whose correlation score occurs at $s=s_0+L_P-1$, the plaintext coefficient satisfies Eq.~\ref{eq:presence-correlation-semantics}. Consequently, $z_j[s]=L_{\mathsf{fix}}$ if and only if the candidate matches $P$ under the ASCII case-insensitive predicate.
\end{lemma}

\begin{lemma}[No Wraparound and Target Soundness]
\label{lem:presence-no-wrap}
For every $\xi=(j,s)\in\mathcal I$, the correlation satisfies $|z_j[s]|\leq L_{\mathsf{fix}}\leq L_{\max}$. If $s\geq L_P-1$, then $x_\xi=z_j[s]-L_{\mathsf{fix}}\in[-2L_{\mathsf{fix}},0]$ and equals zero exactly for a complete case-insensitive match. Otherwise, the dummy target $L_{\max}+1$ gives $x_\xi\in[-(2L_{\max}+1),-1]$. Therefore, if $t>4L_{\max}+2$, both ranges lie within the centered representatives of $\mathbb Z_t$, so neither an incomplete nor a nonmatching complete candidate becomes zero modulo $t$.
\end{lemma}

\begin{lemma}[Selected-Decoding Correctness]
\label{lem:seldec-realization-correctness}
The preprocessing and online selected-decoding protocols $\Pi_{\mathsf{SD}}^{\mathsf{pre}}$ and $\Pi_{\mathsf{SD}}^{\mathsf{on}}$ jointly realize $\mathcal F_{\mathsf{SD}}$ with functional correctness.
For every $\xi\in\mathcal I$, $\mathsf{PP}$ receives nothing,
while $\mathsf{CSP}$ receives $w_\xi$ (Eq.~\ref{eq:presence-masked-decoding1}).
\end{lemma}

\begin{lemma}[Aggregate Correctness]
\label{lem:presence-decoding-aggregation}
By Lemma~\ref{lem:seldec-realization-correctness} and Eq.~\ref{eq:presence-masked-decoding}, selected decoding yields additive shares of every $x_\xi$, while zero-test conversion yields additive shares of $\beta_\xi x_\xi$:
$$
x_\xi^{\mathsf{PP}}+x_\xi^{\mathsf{CSP}}= x_\xi\pmod t,\,\,y_\xi^{\mathsf{PP}}+y_\xi^{\mathsf{CSP}}= \beta_\xi x_\xi\pmod t
$$
If $\mathcal I\neq\varnothing$, the Beaver tree reconstructs the aggregate in Eq.~\ref{eq:presence-aggregate-product}. If $\mathcal I=\varnothing$, the tree is skipped and the empty product is defined as $Y=1$. In either case, 
$
Y=0
\Longleftrightarrow
\exists\,\xi\in\mathcal I:x_\xi=0.
$
\end{lemma}

\begin{theorem}[End-to-End Correctness of \ProtocolName{}]
\label{thm:presence-e2e-correctness}
Let $M_{\mathsf S\rightarrow\mathsf R}$ denote an $L$-bit encoded message history.
For \ProtocolName{} and BFV parameters $L_P,L_{\max},N\in8\mathbb N$, $L_P\leq L_{\max}\leq N$, prime $t>4L_{\max}+2$, and correct BFV decoding at $Q'$, \ProtocolName{} satisfies:
\begin{equation}
\begin{aligned}
\mathsf{present}=1
\Longleftrightarrow
\exists\,M_{\mathsf S\rightarrow\mathsf R},\
\exists\,u\in\{0,8,\ldots,L-L_P\}:\\
\!\!\!\!\!\!\!\mathsf{Ham}_{\gamma}
(
(
\mathbf m_{\mathsf S\rightarrow\mathsf R}[u],
\ldots,
\mathbf m_{\mathsf S\rightarrow\mathsf R}[u+L_P-1]
),
\mathbf p
)
=0
\end{aligned}
\label{eq:presence-e2e-correctness}
\end{equation}
Therefore, \ProtocolName{} produces neither false positives nor false negatives for byte-aligned case-insensitive keyword search.
\end{theorem}

\section{Security Analysis  of \ProtocolName{}}
\label{sec:seek-security}

We analyze \ProtocolName{} in the standard real-world/ideal-world simulation paradigm under the threat model of Section~\ref{subsec:ThreatModel}. Before keyword search, receiver-provided artifacts undergo the functional setup validation of Section~\ref{subsec:functional-validation} and any test failure rejects the setup and flags $\mathsf{R}$ accordingly. The analysis therefore conditions on an accepted and consistent setup 
as specified in Section~\ref{subsec:setup}. Full proofs are deferred to Appendix~\ref{app:seek-security}.

\noindent \ding{110}
\paragraphNew{\DeepUL{Prescribed Leakage}}
The prescribed leakage includes the BFV public parameters $(N,t,Q,Q')$, public BFV keys, $L_{\max}$, and public corpus geometry, which includes message and fragment counts and identifiers, fragment offsets and lengths, and the retained-candidate count $K=|\mathcal I|$. For $\mathsf{CSP}$, the leakage excludes the query-dependent values $P$, $L_P$, $\gamma$, $L_{\mathsf{fix}}$, and the search result. For $\mathsf{PP}$, $P$ and $L_P$ are ideal inputs and beyond those inputs and the final presence bit, $\mathsf{PP}$ learns no plaintext message, correlation coefficient, match count, or location.

\begin{definition}[Ideal Functionality]
\label{def:seek-ci-functionality}
Ideal functionality $\mathcal F_{\mathsf{SEEK}}$ is initialized with retained plaintext messages and public corpus geometry. Upon receiving keyword $P$ from $\mathsf{PP}$, it computes $\mathsf{present}\in\{0,1\}$ according to Eq.~\ref{eq:presence-e2e-correctness}. It releases only $\mathsf{present}$ to $\mathsf{PP}$, no search result to $\mathsf{CSP}$ or $\mathsf{R}$, and otherwise only the prescribed leakage. Repeated admitted queries release the corresponding presence-bit sequence only to $\mathsf{PP}$.
\end{definition}

\begin{lemma}[Selected-Decoding Privacy]
\label{lem:seldec-privacy}
Under the semi-honest-secure 2PC and preprocessing primitives of Section~\ref{sec:beaver},
$\Pi_{\mathsf{SD}}^{\mathsf{pre}}$ and $\Pi_{\mathsf{SD}}^{\mathsf{on}}$ securely realize $\mathcal F_{\mathsf{SD}}$ against a static semi-honest adversary corrupting at most one of $\mathsf{PP}$ and $\mathsf{CSP}$.
\end{lemma}

\begin{lemma}[Presence-Only Aggregate Leakage]
\label{lem:presence-output-privacy}
Let $V_{\mathsf{PP}}^0$ denote a corrupted $\mathsf{PP}$'s complete view in the input-independent preprocessing-hybrid model after invoking $\mathcal F_{\mathsf{SD}}$ and immediately before the Beaver product tree, using fresh, independent zero-test tokens as specified in Section~\ref{subsec:presence-zero-target}. Conditioned on $V_{\mathsf{PP}}^0$, public $K$, and $\mathsf{present}=\mathbf 1[Y=0]$, the aggregate satisfies:
\[
Y\sim
\begin{cases}
1, & K=0\;(\text{necessarily }\mathsf{present}=0)\\
0, & K\geq1\ \text{and}\ \mathsf{present}=1\\
U(\mathbb Z_t^*), & K\geq1\ \text{and}\ \mathsf{present}=0
\end{cases}
\]
Here $U(\mathbb Z_t^*)$ is uniform over the nonzero field elements. For $K=0$, the aggregation transcript is empty and $Y=1$ is fixed by the protocol definition. For $K\geq1$, semi-honest security of Beaver product~\cite{beaver1991efficient} makes the aggregation transcript and $Y$ computationally simulatable from $V_{\mathsf{PP}}^0$, $K$, and $\mathsf{present}$.
\end{lemma}

\begin{theorem}[Semi-Honest Privacy of \ProtocolName{}]
\label{thm:seek-semi-honest}
Under the adaptive multi-message IND-CPA security of BFV, the semi-honest security of batched OLE~\cite{baum2020ole}, and the component guarantees of Lemmas~\ref{lem:seldec-privacy} and~\ref{lem:presence-output-privacy}, the post-setup execution of \ProtocolName{}, including the encrypted corpus, the concrete $\Pi_{\mathsf{SD}}^{\mathsf{pre}}$ and
$\Pi_{\mathsf{SD}}^{\mathsf{on}}$ protocols, and all subsequent search phases, securely realizes $\mathcal F_{\mathsf{SEEK}}$ with the prescribed leakage against a static semi-honest adversary corrupting at most one of $\mathsf{PP}$ and $\mathsf{CSP}$. The corrupted party's view is simulatable from its local setup state, ideal input/output, and prescribed leakage. Consequently, beyond its query and prescribed leakage, $\mathsf{PP}$ learns only $\mathsf{present}$ bit, while $\mathsf{CSP}$ receives no search result. This guarantee extends to any polynomially bounded and potentially adaptive sequence of authorized queries, provided each execution employs fresh independent one-time preprocessing: the information leakage is limited to the ideal presence-bit sequence and its logical implications.
\end{theorem}

\section{Experimental Evaluation}
\label{sec:eval}

We evaluate \ProtocolName{} across three complementary dimensions. We first microbenchmark the computation and communication costs of individual protocol stages.
We then measure E2E latency using web and mobile application prototypes.
Finally, we assess the 
feasibility of repeated presence-only queries.

\smallskip
\noindent \ding{110}
\paragraphNew{\DeepUL{Implementation and Baseline}}
We implement \ProtocolName{} and the fragmentation-based baseline~\cite{1} using \textsf{Microsoft SEAL}~\cite{sealcrypto} and instantiate all protocol roles $(\mathsf{S}/\mathsf{R}/\mathsf{PP}/\mathsf{CSP})$ 
to isolate protocol design costs from runtime heterogeneity.
All microbenchmarks are executed single-threaded on a machine with two 64-core AMD EPYC processors, 2~TB of DDR4 memory, and Ubuntu~22.04.5~LTS. Both schemes use the same parameter grid: $N\in\{2^{12},2^{13},2^{14},2^{15}\}$ and $L_{\max}\in\{128,256,512,1024,2048\}$. Each configuration uses a 20-bit prime plaintext modulus satisfying $t\equiv1\pmod{2N}$ and $t>4L_{\max}+2$, together with a ciphertext-modulus chain providing at least 128-bit classical security~\cite{HomomorphicEncryptionSecurityStandard}. Appendix~\ref{app:bfv-params} reports the full implementation parameterization $(N,t,Q,Q',q_\ell)$ and BFV noise-budget stress test. 
We select~\cite{1} as our primary baseline because its RLWE-based fragmentation design closely matches our construction and E2EE messaging setting, compared to other HE-based techniques (Appendix~\ref{ComparisonWithHE}).

\smallskip
\noindent \ding{110}
\paragraphNew{\DeepUL{Datasets and Keywords}} \ProtocolName{} is evaluated on the \textsf{SAMSum} messenger-style dataset, comprising $173{,}717$ messages with approximately $1.5$M words between $4{,}416$ distinct users. We use a common parameter setting across all users, which makes the total cost of searching individual conversations comparable to that of searching the entire corpus. This approach demonstrates \ProtocolName{}'s scalability with increasing message volume. Since \textsf{SAMSum} lacks timestamps, we assign interpretive workload durations using population-average estimates of $54.15$ messages per user per day and $14.31$ words per message~\cite{ofcomMessaging2023,onsPopulation2022,lyddy2014analysis}. Under this mapping, the complete corpus represents approximately five years of messaging, while $1$K and $25$K words correspond to about one day and one month, respectively. We evaluate ten increasing workloads spanning these scales under $500$ patterns averaging $16$ bytes, evenly divided between present and absent queries. Present patterns are sampled from \textsf{SAMSum}, while absent patterns are drawn from a frequency-ranked WordFreq/SUBTLEX-US universe of $248{,}266$ keywords~\cite{wordfreq,brysbaert-new-2009}. Section~\ref{sec:eval_repeated_queries} uses this universe to analyze the operational cost of repeated queries.

\smallskip
\noindent \ding{110} {\paragraphNew{{\textbf{\DeepUL{End-to-End Web and Mobile Prototype}}}}}
To measure E2E latency and deployment feasibility on commodity devices, we build a web and cross-platform mobile application with $\mathsf{S}$, $\mathsf{R}$, $\mathsf{PP}$, and $\mathsf{CSP}$ deployed as distinct networked endpoints. We use a modified \textsf{node-seal} backend~\cite{node_seal} supporting the coefficient-domain plaintext construction required by \ProtocolName{} but unavailable in the standard bindings. We run the mobile client on a OnePlus~11R with 8~GB RAM and 256~GB storage. We provide detailed implementation, demonstrations, and interface discussion in Appendix~\ref{applicationImages}.

\begin{figure*}[!t]
  \centering
  \includegraphics[width=\linewidth] {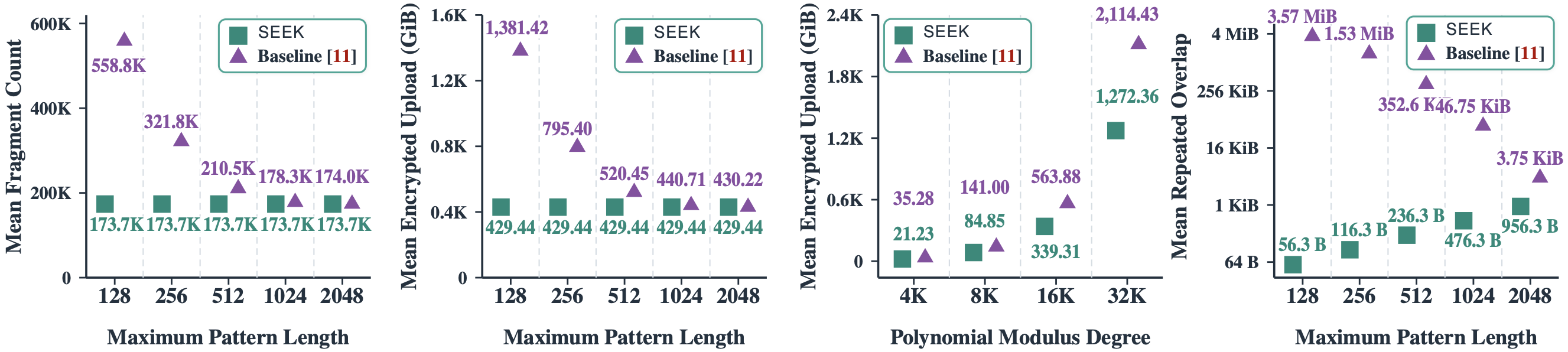}%
    \caption{Fragmentation, encrypted-upload, and overlap overheads of \ProtocolName{} and \cite{1} for the $\approx1.5$M-word message corpus (five year history). Values indexed by $L_{\max}$ are averaged over all $N$ and the $N$-indexed upload is averaged over all $L_{\max}$. }
  \label{fig:seek-rlwe-fragments-upload}
\end{figure*}

\subsection{Protocol Microbenchmarks and Scalability}
\label{sec:eval_microbenchmarks}

\noindent \ding{110}
\paragraphNew{\DeepUL{Correctness}} Across all $20$ $(N,L_{\max})$ pairs, we execute $10{,}000$ encrypted searches: $5000$ positive and $5000$ negative. \ProtocolName{} produces no false positives or false negatives, achieving $100\%$ precision, recall, F1 score, and accuracy. By contrast, the exact-match protocol of~\cite{1} maintains $100\%$ precision but achieves only $60\%$ recall, $75\%$ F1 score, and $80\%$ accuracy because case variations cause $2000$ false negatives. Its wildcard extension supports equivalent ASCII case-insensitive matching but requires a two-ciphertext trapdoor, compared with one ciphertext for \ProtocolName{}. Thus, \ProtocolName{} handles capitalization using negligible local case-mask construction, without additional trapdoor variants or online communication.

\noindent \ding{110}
\paragraphNew{\DeepUL{Fragmentation, Encryption, and Trapdoor Generation}} Fig.~\ref{fig:seek-rlwe-fragments-upload} reports fragmentation, overlap, and encrypted-ingestion overhead for the $1.5$M-word (five-year) workload. Across the four $N$ values, \ProtocolName{} averages $173.7$K fragments at every $L_{\max}$, while the baseline decreases from $558.8$K at $L_{\max}=128$ to $210.5$K at $512$ and $174.0$K at $2048$. At $L_{\max}=128$ and $512$, \ProtocolName{} reduces mean fragment count and encrypted upload by $3.22\times$ and $1.21\times$. Across all $20$ pairs, it reduces the mean fragment count from $288.7$K to $173.7$K and upload/storage from $713.64$ to $429.44$~GiB, a $39.82\%$ reduction. Increasing $N$ removes only the $15$ excess fragments at $N=4096$ but expands \ProtocolName{}'s mean footprint from $21.23$~GiB to $1272.36$~GiB at $N=32768$, favoring the smallest secure degree satisfying the computation and correctness bounds. At the practical $(N,L_{\max})=(4096,512)$ pair, \ProtocolName{} and the baseline require $173{,}732$ and $210{,}535$ fragments, occupy $21.23$ and $25.73$~GiB, and complete fragmentation and encryption in $198.27$ and $230.13$~s, respectively. This corresponds to $1.14$~ms per message for \ProtocolName{}. At $(4096,128)$, the sender costs are $1.12$ and $3.44$~ms per message. Overlap-induced repetition shows an even larger gap. Averaged over $N$, \ProtocolName{} and the baseline process $56.3$~bytes and $3.57$~MiB at $L_{\max}=128$ ($66{,}516\times$ reduction), and $236.3$~bytes and $352.6$~KiB at $L_{\max}=512$ ($1{,}528\times$ reduction). Their grid-wide averages are $368.3$~bytes and $1.10$~MiB, giving a $3{,}129\times$ reduction. At $N=4096$, \ProtocolName{}/baseline trapdoors occupy $128.11/256.22$~KiB and take $15.99/29.70$~ms to construct, approximately halving both costs without penalizing smaller policy-adequate values of $L_{\max}$. Long documents amplify the fragmentation advantage even more. For $N=4096/8192/16384/32768$, $1$K/$10$K-word documents at $L_{\max}=512$ require $84/839$ baseline blocks but only $12/120$, $6/56$, $3/27$, and $2/14$ \ProtocolName{} fragments, yielding $7$--$42\times$ and $6.99$--$59.93\times$ reductions. At $L_{\max}=128$, the baseline requires $335/3{,}353$ blocks, compared with $11/108$, $6/54$, $3/27$, and $2/14$ for \ProtocolName{}, yielding $30.45$--$167.5\times$ and $31.05$--$239.5\times$ reductions. Because the construction of~\cite{1} couples its base-fragment length to the maximum supported pattern length, our primary comparison uses the common $L_{\max}$. Even when the baseline alone is overprovisioned from $512$ bits to its maximum $N/2$, its $1$K/$10$K-word counts of $21/210$, $11/105$, $6/53$, and $3/27$ remain $1.50$--$2.00\times$ and $1.75$--$1.96\times$  those of \ProtocolName{} at $L_{\max}=512$. This bound already supports four times the average pattern length. However, larger values only increase overlap, reduce stride, and generate more fragments for long messages.

\noindent \ding{110}
\paragraphNew{\DeepUL{Correlation Computation}}
Fig.~\ref{fig:seek-rlwe-correlation-time}\footnote{Here and subsequently, costs increase with $N$, history size, or query count, as applicable (Figs.~\ref{fig:seek-rlwe-correlation-time}--\ref{fig:attack}), enabling distinction without varying markers.} compares \ProtocolName{}'s encrypted correlation overhead with the baseline's online adjacent-pair reconstruction and case-insensitive wildcard correlation cost across all histories and $20$ $(N,L_{\max})$ pairs. Following~\cite{1}, adjacent-pair overlaps are reconstructed online for each query. Precomputing and storing them would increase storage by up to $2\times$ while saving only the $0.18\%$ median time spent on reconstruction. \ProtocolName{} is faster in all $200$ observations, achieving median and maximum speedups of $1.99\times$ and $5.47\times$, with a $4.52\times$ median at $L_{\max}=128$. Including preprocessing yields a $1.95\times$ speedup because preprocessing contributes only $1.54\%$ of the combined cost. Computation grows approximately linearly with retained history. From $1$K to $1.5$M words, the median over $L_{\max}$ increases from $0.382/0.784$ to $670.3/1{,}352.6$~s for \ProtocolName{}/baseline at $N=4096$, and from $37.45/69.15$~s to $65.71/119.17$~ks at $N=32768$. For fixed $N$ and history, \ProtocolName{}'s median max--min ratio across the five $L_{\max}$ values is only $1.03\times$, compared with up to $2.9\times$ for the baseline. This advantage follows from fewer instances and lighter arithmetic. \ProtocolName{} uses one ct--ct multiplication per fragment, while each baseline wildcard variant's core correlation uses two ct--ct multiplications, two ct--pt multiplications, and two additions/ciphertext subtractions, with adjacent-pair reconstruction adding one monomial ct--pt multiplication and one ciphertext addition.

\begin{figure}[!t]
  \centering
  \includegraphics[width=\columnwidth] {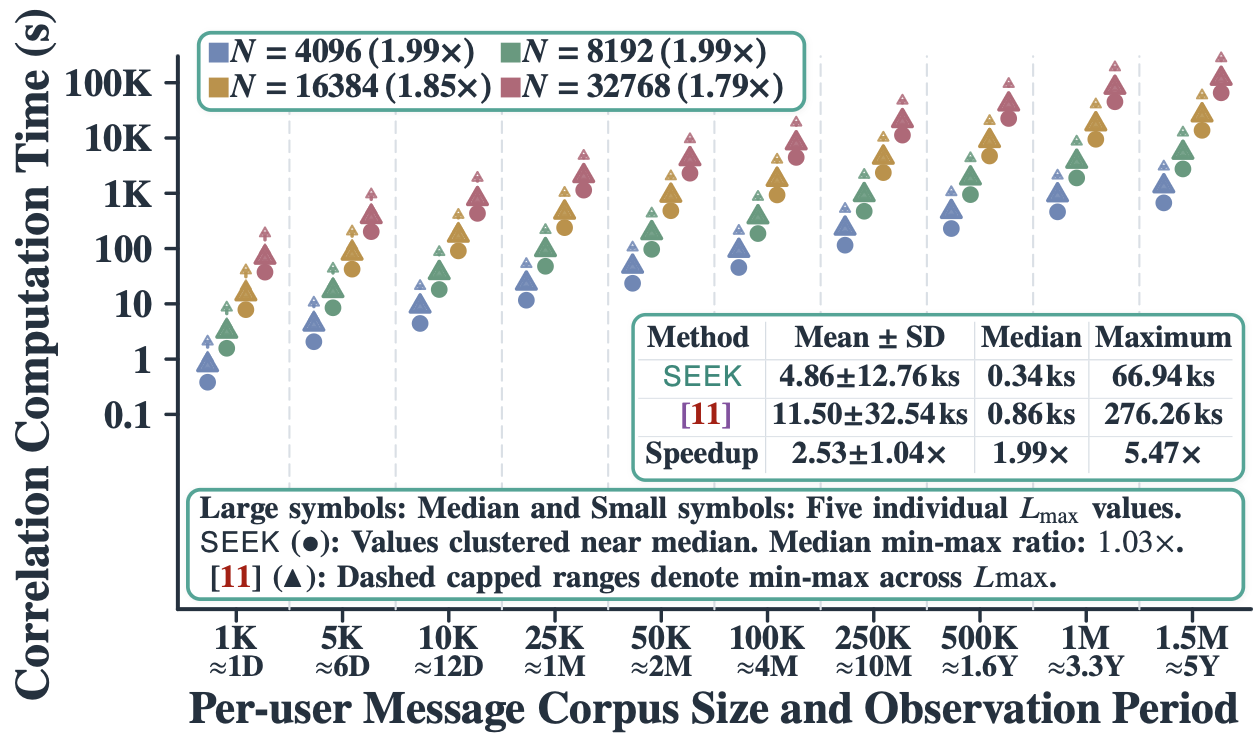}
\caption{Query-wise average correlation cost of \ProtocolName{} and overlap construction with correlation cost of \cite{1}.}
  \label{fig:seek-rlwe-correlation-time}
\end{figure}

\noindent \ding{110}
\paragraphNew{\DeepUL{Online Cost Analysis and Scalability}}
Fig.~\ref{fig:full-protocol-overhead} reports average per-query costs across all message histories and $(N,L_{\max})$ pairs. Computation begins once the encrypted trapdoor, preprocessed corpus, and fresh preprocessing material are available and ends when $\mathsf{PP}$ constructs $\mathsf{present}$. Communication includes trapdoor delivery and all subsequent $\mathsf{PP}$--$\mathsf{CSP}$ exchanges. Across all $(N,L_{\max})$ pairs, daily, monthly, and five-year histories require $0.408$--$42.46$~s and $3.74$--$32.97$~MiB, $12.41$~s--$1.29$~ks and $107.49$--$778.96$~MiB, and $715.50$~s--$74.50$~ks and $6.05$--$43.44$~GiB, respectively. From one month to five years, computation increases by $57.64$--$57.65\times$ and communication by $57.10$--$57.62\times$, closely tracking message growth and confirming near-linear scalability. The bottom panel identifies $N$ as the principal performance determinant. For the five-year history, increasing $N$ from $4096$ to $32768$ raises median computation over $L_{\max}$ from $727.43$~s to $73.19$~ks, a $100.6\times$ increase, and communication from $6.05$ to $43.44$~GiB, a $7.18\times$ increase. At fixed $N$, varying $L_{\max}$ changes computation by at most $1.05\times$ and leaves communication unchanged because overlap pruning limits selected decoding and aggregation to unique byte-aligned positions. Across all $200$ searches, correlation multiplication, relinearization, and modulus switching account for $97.0\%$ of computation, while reduced-$c_1$ transfer and GMW openings account for $96.6\%$ of communication. Consequently, $N=4096$ is the practical sweet spot whenever its decoding and correctness conditions hold. For $L_{\max}=256$, $512$, and $1024$, it requires only $0.408$--$0.424$~s, $12.41$--$12.90$~s, and $715.50$--$743.65$~s for daily, monthly, and five-year histories, with corresponding communication of $3.74$~MiB, $107.49$~MiB, and $6.05$~GiB.

\begin{figure}[!t]
  \centering
  \includegraphics[width=\columnwidth] {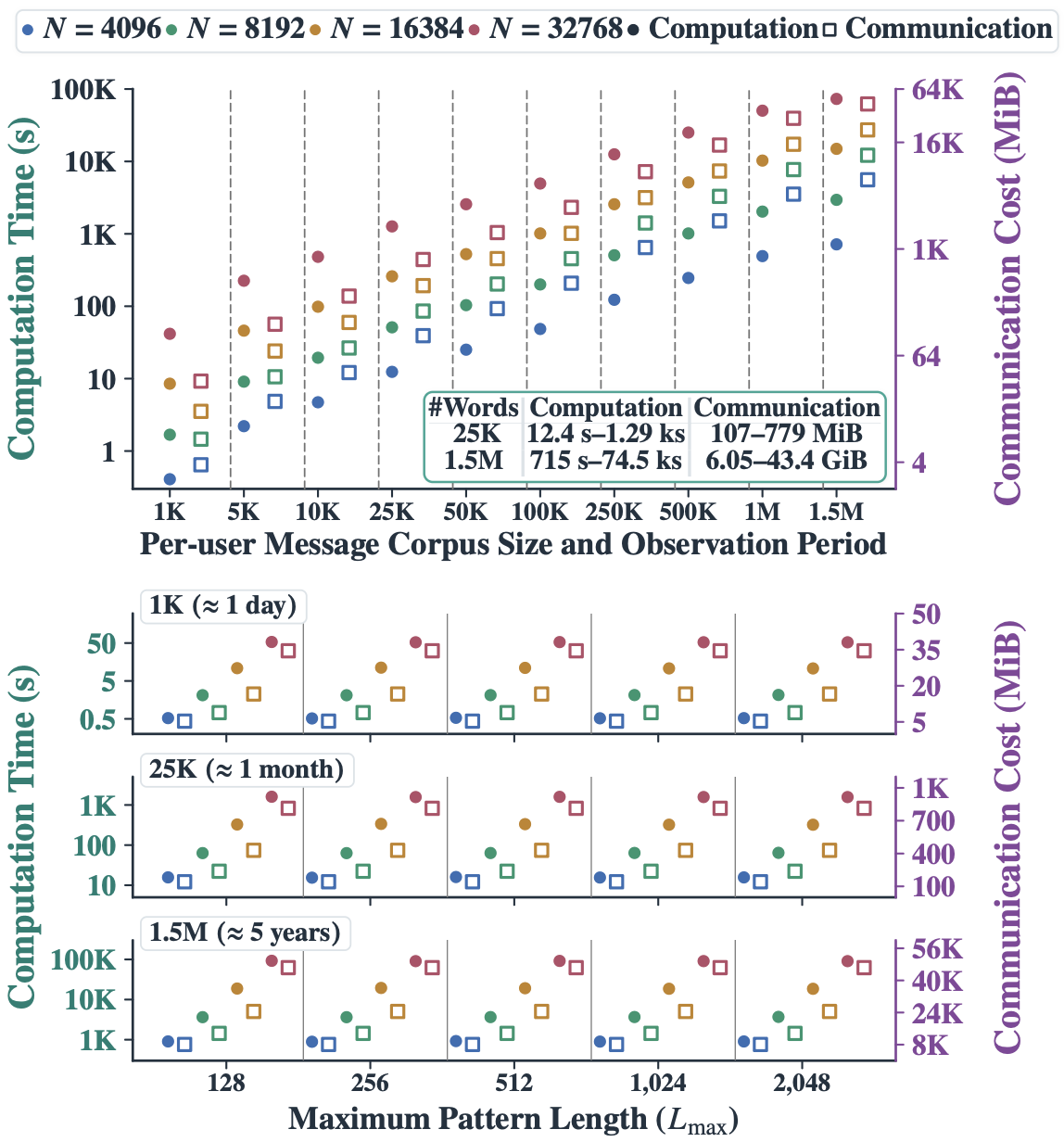}
  \caption{Computation and communication costs of \ProtocolName{} with increasing message history and across all $(N,L_{\max})$ pairs.}
  \label{fig:full-protocol-overhead}
\end{figure}

\noindent \ding{110}
\paragraphNew{\DeepUL{Offline Cost Analysis}} Fig.~\ref{fig:one-time-setup-cost} reports the one-time main BFV setup and fresh per-query $\Pi_{\mathsf{SD}}^{\mathsf{pre}}$ costs, assuming Boolean triples and daBits are available. The setup includes BFV context and key generation, additive secret-key sharing, serialization, and key/share distribution. Its cost increases from $12.518$~ms and $1.501$~MiB at $N=4096$ to $1.053$~s and $272.008$~MiB at $N=32{,}768$ and is amortized over the key-lifecycle period. At the practical $(N,L_{\max})=(4096,512)$ setting, $\Pi_{\mathsf{SD}}^{\mathsf{pre}}$ requires $0.562$~ms and $0.365$~MiB for the daily history, $14.460$~ms and $9.138$~MiB for the monthly history, and $0.817$~s and $530.340$~MiB for the five-year history. The complete fresh per-query bundle for joint-mask preprocessing, selected decoding, zero testing, and aggregation contains $K(8\ell_Q+2)$ Boolean triples, $3K$ daBits, $K-1$ arithmetic Beaver triples, and $K$ scalar OLEs~\cite{demmler2015aby,rotaru2019marbled,escudero2020mixed,beaver1991efficient,baum2020ole}. Under a state-of-the-art $128$-bit semi-honest generation model using batches of $10^7$ and an amortized rate of $0.118$ communicated bits per random oblivious transfer~\cite{bruggemann2023flute}, generating this bundle is projected to communicate $1.666$~MiB, $41.745$~MiB, and $2.366$~GiB for the daily, monthly, and five-year histories, respectively. Since these resources depend only on public $K$ and the BFV parameters, rather than the corpus, keyword, or result, they can be batch-generated before query arrival.

\begin{figure}[!t]
\centering
\includegraphics[width=0.8\columnwidth] {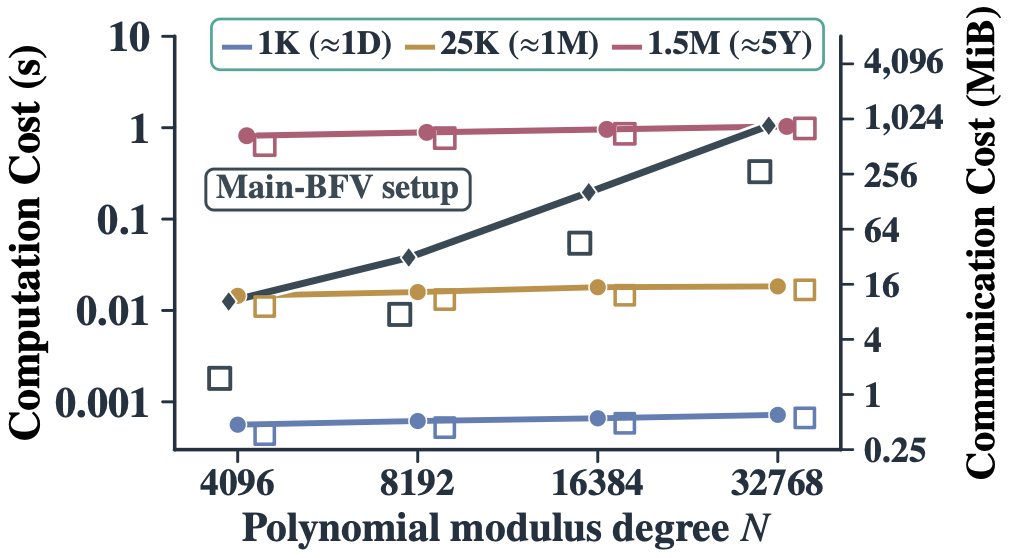}
\caption{One-time BFV setup and fresh per-query selected-decoding preprocessing overheads of \ProtocolName{} at $L_{\max}=512$.}
\label{fig:one-time-setup-cost}
\end{figure}

\noindent \ding{42} \paragraphNew{\DeepUL{Example}} At $(4096,512)$, with 
preprocessing available, the five-year online search from trapdoor delivery to final $\mathsf{present}$ bit requires $743.65$~s and $6.05$~GiB, or $13.26$~min with serialized computation and ideal $1$~Gbps transmission. 
Even the largest evaluated corpus, $1.243$~TiB at $(32768,2048)$, is searched online in approximately $20.39$~h, demonstrating \ProtocolName{}'s feasibility at substantially larger scales.

\subsection{Repeated-Query Attack Analysis}
\label{sec:eval_repeated_queries}

Fig.~\ref{fig:attack} demonstrates the conservative lower-bound costs for repeated-query abuse over one-month and five-year \textsf{SAMSum} histories~\cite{gliwa-etal-2019-samsum}, considering only the online correlation-to-presence pipeline under the practical $(N,L_{\max})=(4096,128)$ configuration. Following the fixed, target-independent external WordFreq/SUBTLEX-US frequency ranking~\cite{wordfreq,brysbaert-new-2009}, the first $1334/1290$, $4655/4521$, and $21{,}537/20{,}145$ queries cover at least $70\%$, $80\%$, and $90\%$ of word occurrences in the one-month/five-year histories, respectively. The complete $248{,}266$-query sweep reaches only $92.50\%/92.57\%$ coverage. These evaluator-side oracle measurements overstate actual disclosure because each query reveals only one history-wide substring-presence bit without counts, message identity, position, or order. Consequently, even $90\%$ occurrence coverage neither recovers $90\%$ of the plaintext nor provides sufficient structure to reconstruct conversations. At WhatsApp scale \cite{metaWhatsAppThreeBillion2026}, applying a $5$K-query sweep, which exceeds both $80\%$ thresholds, to three billion five-year histories is projected to require $346.5$ million compute-years and $97.4$~ZB. Even one platform-wide query requires $69{,}300$ compute-years and $19.5$~EB, taking approximately $36$ weeks with $10^5$ perfectly balanced workers. Within one week, this cluster can process at most approximately $83$ million user-keyword evaluations, equivalent to one keyword over $83$ million users or $5$K keywords over only $16{,}600$ users, before communication further reduces this bound. An attacker must therefore trade keyword breadth against population and history length. Narrow targeted probing remains feasible, but platform-scale reconstruction is computationally and communicationally infeasible.

\begin{figure}[!t]
\includegraphics[width=\columnwidth] {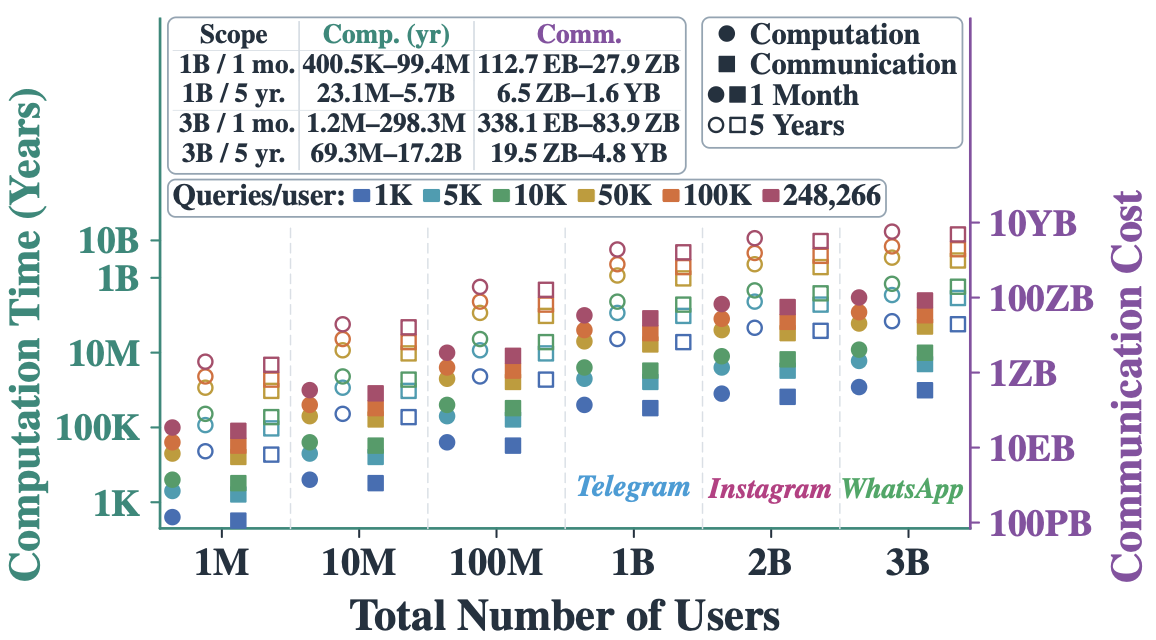}
\caption{Repeated-Query Overhead.}
\label{fig:attack}
\end{figure}

\subsection{End-to-End Prototype Evaluation}
\label{sec:eval_prototype}

We evaluate the complete \ProtocolName{} application prototype under $(4096,512)$ using a one-week history of approximately $5$K words represented by $91$ fragments. The complete setup and functional validation ($\kappa=4$) takes $24.62$~s on the web client and $24.81$~s on mobile, with $6.37$~MB of communication, indicating that server-side $\mathsf{PP}$--$\mathsf{CSP}$ operations dominate. Message packing and BFV encryption cost approximately $0.4$~ms per word on the web client and $4$~ms per word on mobile. Fresh input-independent preprocessing requires $67.76$~s of computation in the current serialized implementation, with $36.1$~s incurred by $\mathsf{PP}$ and $31.66$~s by $\mathsf{CSP}$. Computation is dominated by Boolean-triple generation ($49.91$~s) and base OT ($12.15$~s), while preprocessing exchanges $622.72$~MB of payload, primarily for Boolean triples ($549.38$~MB) and arithmetic Beaver triples ($46.31$~MB), corresponding to an ideal serialization time of $4.98$~s over a $1$~Gbps link. \ProtocolName{} additionally exchanges $15.84$~MB across corpus processing and online 2PC ($11.65$~MB). A fresh execution including preprocessing therefore corresponds to $5.11$~s of communication time at $1$~Gbps. With input-independent preprocessing off the query path, trapdoor encryption, encrypted $\pm1$ mapping, and correlation-to-bit reconstruction require $1.59$~ms, $8.51$~ms, and $1.91$~s, respectively. The latter includes $319.8$~ms for homomorphic multiplication, $626.4$~ms for joint-mask processing, and $631.3$~ms for selected decoding. Overall, the query-time cryptographic pipeline requires $1.92$~s, demonstrating practical E2E encrypted messaging and privacy-preserving search across web and commodity mobile devices.

\section{Conclusion}
\label{sec:conclusion}

We presented \ProtocolName{}, a privacy-preserving keyword-search protocol that combines minimum-sufficient byte-aligned fragmentation, packed BFV correlation, and 2PC-based selected decoding, blinded zero testing, and secure aggregation. It supports ASCII case-insensitive search using one fixed-size encrypted trapdoor and one query-dependent homomorphic multiplication per fragment, releasing only a corpus-wide presence bit to $\mathsf{PP}$. We established search correctness and post-setup privacy against a static semi-honest adversary under the stated assumptions. 
\ProtocolName{} achieved 100\% observed accuracy, reduced mean 
fragment 
overhead by 39.82\%, and provided maximum correlation speedup of $5.47\times$ over the 
wildcard baseline. Our web/mobile prototype and scalability analysis further demonstrate practical and efficient presence-only search over retained encrypted messages.

\bibliographystyle{ACM-Reference-Format}
\bibliography{references}

@article{1,
  title={Efficient post-quantum pattern matching on encrypted data},
  author={Bkakria, Anis and Izabach{\`e}ne, Malika},
  journal={IACR Communications in Cryptology},
  volume={1},
  number={2},
  year={2024}
}

@inproceedings{24,
  title={$\{$MUSES$\}$: Efficient $\{$Multi-User$\}$ Searchable Encrypted Database},
  author={Le, Tung and Behnia, Rouzbeh and Guajardo, Jorge and Hoang, Thang},
  booktitle={33rd USENIX Security Symposium (USENIX Security 24)},
  pages={2581--2598},
  year={2024}
}

@inproceedings{25,
  title={$\{$d-DSE$\}$: Distinct Dynamic Searchable Encryption Resisting Volume Leakage in Encrypted Databases},
  author={Liu, Dongli and Wang, Wei and Xu, Peng and Yang, Laurence T and Luo, Bo and Liang, Kaitai},
  booktitle={33rd USENIX Security Symposium (USENIX Security 24)},
  pages={2563--2580},
  year={2024}
}

@article{1:CS15,
  title={Substring-searchable symmetric encryption},
  author={Chase, Melissa and Shen, Emily},
  journal={Proceedings on Privacy Enhancing Technologies},
  year={2015}
}

@inproceedings{2,
  title={Authorized keyword search on encrypted data},
  author={Shi, Jie and Lai, Junzuo and Li, Yingjiu and Deng, Robert H and Weng, Jian},
  booktitle={European symposium on research in computer security},
  pages={419--435},
  year={2014},
  organization={Springer}
}

@inproceedings{3,
  title={Efficient encrypted keyword search for multi-user data sharing},
  author={Kiayias, Aggelos and Oksuz, Ozgur and Russell, Alexander and Tang, Qiang and Wang, Bing},
  booktitle={European symposium on research in computer security},
  pages={173--195},
  year={2016},
  organization={Springer}
}

@inproceedings{7,
  title={CIPHERMATCH: Accelerating Homomorphic Encryption-Based String Matching via Memory-Efficient Data Packing and In-Flash Processing},
  author={Kabra, Mayank and Nadig, Rakesh and Gupta, Harshita and Bera, Rahul and Frouzakis, Manos and Arulchelvan, Vamanan and Liang, Yu and Mao, Haiyu and Sadrosadati, Mohammad and Mutlu, Onur},
  booktitle={Proceedings of the 30th ACM International Conference on Architectural Support for Programming Languages and Operating Systems, Volume 2},
  pages={111--130},
  year={2025}
}

@inproceedings{8,
  title={Authorized private keyword search over encrypted data in cloud computing},
  author={Li, Ming and Yu, Shucheng and Cao, Ning and Lou, Wenjing},
  booktitle={2011 31st international conference on distributed computing systems},
  pages={383--392},
  year={2011},
  organization={IEEE}
}

@inproceedings{11,
  title={$\{$I/O-Efficient$\}$ dynamic searchable encryption meets forward \& backward privacy},
  author={Mondal, Priyanka and Chamani, Javad Ghareh and Demertzis, Ioannis and Papadopoulos, Dimitrios},
  booktitle={33rd USENIX Security Symposium (USENIX Security 24)},
  pages={2527--2544},
  year={2024}
}

@inproceedings{12,
  title={Pattern matching on encrypted streams},
  author={Desmoulins, Nicolas and Fouque, Pierre-Alain and Onete, Cristina and Sanders, Olivier},
  booktitle={International Conference on the Theory and Application of Cryptology and Information Security},
  pages={121--148},
  year={2018},
  organization={Springer}
}

@inproceedings{15,
  title={Public key encryption with flexible pattern matching},
  author={Bouscati{\'e}, Elie and Castagnos, Guilhem and Sanders, Olivier},
  booktitle={International Conference on the Theory and Application of Cryptology and Information Security},
  pages={342--370},
  year={2021},
  organization={Springer}
}

@article{20,
  title={SWiSSSE: System-wide security for searchable symmetric encryption},
  author={Gui, Zichen and Paterson, Kenneth G and Patranabis, Sikhar and Warinschi, Bogdan},
  journal={Proceedings on Privacy Enhancing Technologies},
  year={2024}
}

@inproceedings{22,
  title={Dynamic searchable encryption with optimal search in the presence of deletions},
  author={Chamani, Javad Ghareh and Papadopoulos, Dimitrios and Karbasforushan, Mohammadamin and Demertzis, Ioannis},
  booktitle={31st USENIX Security Symposium (USENIX Security 22)},
  pages={2425--2442},
  year={2022}
}

@inproceedings{23,
  title={$\{$FEASE$\}$: Fast and Expressive Asymmetric Searchable Encryption},
  author={Meng, Long and Chen, Liqun and Tian, Yangguang and Manulis, Mark and Liu, Suhui},
  booktitle={33rd USENIX Security Symposium (USENIX Security 24)},
  pages={2545--2562},
  year={2024}
}

@inproceedings{1:SLPR15,
  title={Blindbox: Deep packet inspection over encrypted traffic},
  author={Sherry, Justine and Lan, Chang and Popa, Raluca Ada and Ratnasamy, Sylvia},
  booktitle={Proceedings of the 2015 ACM conference on special interest group on data communication},
  pages={213--226},
  year={2015}
}

@inproceedings{1:CDK+17,
  title={BlindIDS: Market-compliant and privacy-friendly intrusion detection system over encrypted traffic},
  author={Canard, S{\'e}bastien and Diop, A{\"\i}da and Kheir, Nizar and Paindavoine, Marie and Sabt, Mohamed},
  booktitle={Proceedings of the 2017 ACM on Asia Conference on Computer and Communications Security},
  pages={561--574},
  year={2017}
}

@inproceedings{1:DFOS18,
  title={Pattern matching on encrypted streams},
  author={Desmoulins, Nicolas and Fouque, Pierre-Alain and Onete, Cristina and Sanders, Olivier},
  booktitle={International Conference on the Theory and Application of Cryptology and Information Security},
  pages={121--148},
  year={2018},
  organization={Springer}
}

@inproceedings{1:BCC20,
  title={Privacy-preserving pattern matching on encrypted data},
  author={Bkakria, Anis and Cuppens, Nora and Cuppens, Fr{\'e}d{\'e}ric},
  booktitle={International Conference on the Theory and Application of Cryptology and Information Security},
  pages={191--220},
  year={2020},
  organization={Springer}
}

@inproceedings{1:BCS21,
  title={Public key encryption with flexible pattern matching},
  author={Bouscati{\'e}, Elie and Castagnos, Guilhem and Sanders, Olivier},
  booktitle={International Conference on the Theory and Application of Cryptology and Information Security},
  pages={342--370},
  year={2021},
  organization={Springer}
}

@inproceedings{1:BCS23,
  title={Pattern matching in encrypted stream from inner product encryption},
  author={Bouscati{\'e}, {\'E}lie and Castagnos, Guilhem and Sanders, Olivier},
  booktitle={IACR International Conference on Public-Key Cryptography},
  pages={774--801},
  year={2023},
  organization={Springer}
}

@misc{N19,
  author       = {{National Crime Agency}},
  title        = {{NCA and police smash thousands of criminal conspiracies after infiltration of encrypted communication platform in UK's biggest ever law enforcement operation}},
  howpublished = {\url{https://www.nationalcrimeagency.gov.uk/news/operation-venetic}},
  year         = {2020},
}

@misc{N20,
  author       = {{US Department of Justice}},
  title        = {{FBI's Encrypted Phone Platform Infiltrated Hundreds of Criminal Syndicates, Resulting in Massive Worldwide Takedown}},
  howpublished = {\url{https://www.justice.gov/usao-sdca/pr/fbi-s-encrypted-phone-platform-infiltrated-hundreds-criminal-syndicates-result-massive}},
  year         = {2021},
}

@misc{sealcrypto,
    title = {{M}icrosoft {SEAL} (release 4.1)},
    howpublished = {\url{https://github.com/Microsoft/SEAL}},
    month = jan,
    year = 2023,
    note = {Microsoft Research, Redmond, WA.},
    key = {SEAL}
}

@misc{node_seal,
  author ={Nick Angelou},
  title = {{node-seal: Homomorphic Encryption for TypeScript or JavaScript - Microsoft SEAL}},
  howpublished = {\url{https://github.com/s0l0ist/node-seal/tree/main}},
  Date ={2025},
}

@misc{expo,
  author= {Expo},
  title = {{Expo: An open-source framework for making universal native apps with React.}},
  howpublished = {\url{https://github.com/expo/expo}},
  Date ={2025},
}

@article{fan2012somewhat,
  title={Somewhat practical fully homomorphic encryption},
  author={Fan, Junfeng and Vercauteren, Frederik},
  journal={Cryptology ePrint Archive},
  year={2012}
}

@inproceedings{brakerski2012fully,
  title={Fully homomorphic encryption without modulus switching from classical GapSVP},
  author={Brakerski, Zvika},
  booktitle={Annual cryptology conference},
  pages={868--886},
  year={2012},
  organization={Springer}
}

@misc{N21,
  author       = {{Reuters}},
  title        = {{UK orders Apple to open up users' encrypted cloud data, report says}},
  howpublished = {\url{https://www.reuters.com/world/uk/uk-asks-apple-let-it-spy-users-encrypted-accounts-washington-post-reports-2025-02-07/}},
  year         = {2025}
}

@misc{N22,
  author       = {{Reuters}},
  title        = {{Apple pulls data protection feature in UK amid government demands}},
  howpublished = {\url{https://www.reuters.com/technology/apple-removing-end-to-end-cloud-encryption-feature-uk-bloomberg-news-reports-2025-02-21/}},
  year         = {2025}
}

@misc{N23,
  author       = {{European Commission}},
  title        = {{Proposal for a Regulation laying down rules to prevent and combat child sexual abuse}},
  howpublished = {\url{https://eur-lex.europa.eu/legal-content/EN/TXT/?uri=celex\%3A52022PC0209}},
  year         = {2022}
}

@misc{N24,
  author       = {{Reuters}},
  title        = {{EU fails to extend rules on child abuse content detection by online platforms}},
  howpublished = {\url{https://www.reuters.com/legal/litigation/eu-fails-extend-rules-child-abuse-content-detection-by-online-platforms-2026-03-16/}},
  year         = {2026}
}

@misc{N25,
  author       = {{Rest of World}},
  title        = {{WhatsApp gives India an ultimatum on encryption}},
  howpublished = {\url{https://restofworld.org/2024/exporter-whatsapp-encryption-india/}},
  year         = {2024}
}

@misc{N26,
  author       = {{Australian Department of Home Affairs}},
  title        = {{The Assistance and Access Act 2018}},
  howpublished = {\url{https://www.homeaffairs.gov.au/about-us/our-portfolios/national-security/lawful-access-telecommunications/data-encryption}},
  year         = {2023}
}

@misc{N27,
  author       = {{Parliament of Australia}},
  title        = {{Telecommunications and Other Legislation Amendment (Assistance and Access) Bill 2018}},
  howpublished = {\url{https://www.aph.gov.au/Parliamentary_Business/Bills_Legislation/Bills_Search_Results/Result?bId=r6195}},
  year         = {2018}
}

@online{N28,
  author       = {{The Guardian}},
  title        = {{Instagram to remove end-to-end encryption for private messages in May}},
  url = {{https://www.theguardian.com/technology/2026/mar/18/instagram-to-remove-end-to-end-encryption-for-private-messages-in-may}},
  year         = {2026}
}

@online{N18,
  author       = {{Europol}},
  title        = {{Dismantling encrypted criminal EncroChat communications leads to over 6,500 arrests and close to EUR 900 million seized}},
  url = {{https://www.europol.europa.eu/media-press/newsroom/news/dismantling-encrypted-criminal-encrochat-communications-leads-to-over-6-500-arrests-and-close-to-eur-900-million-seized}},
  year         = {2023},
}

@article{mouchet2021multiparty,
  title={Multiparty homomorphic encryption from ring-learning-with-errors},
  author={Mouchet, Christian and Troncoso-Pastoriza, Juan and Bossuat, Jean-Philippe and Hubaux, Jean-Pierre},
  journal={Proceedings on Privacy Enhancing Technologies},
  volume={2021},
  number={4},
  pages={291--311},
  year={2021}
}

@inproceedings{beaver1991efficient,
  title={Efficient multiparty protocols using circuit randomization},
  author={Beaver, Donald},
  booktitle={Annual international cryptology conference},
  pages={420--432},
  year={1991},
  organization={Springer}
}

@misc{N29,
  author       = {Nicas, Jack and Zhong, Raymond and Wakabayashi, Daisuke},
  title        = {{Censorship, Surveillance and Profits: A Hard Bargain for Apple in China}},
  howpublished = {\url{https://www.nytimes.com/2021/05/17/technology/apple-china-censorship-data.html}},
  year         = {2021},
  note         = {The New York Times}
}

@misc{N30,
  author       = {Canetti, Ran and Kaptchuk, Gabriel},
  title        = {{The Broken Promise of Apple's Announced Forbidden-photo Reporting System -- And How To Fix It}},
  howpublished = {\url{https://www.bu.edu/riscs/2021/08/10/apple-csam/}},
  year         = {2021}
}

@article{N31,
  title={Bugs in our pockets: the risks of client-side scanning},
  author={Abelson, Harold and Anderson, Ross and Bellovin, Steven M and Benaloh, Josh and Blaze, Matt and Callas, Jon and Diffie, Whitfield and Landau, Susan and Neumann, Peter G and Rivest, Ronald L and others},
  journal={Journal of Cybersecurity},
  volume={10},
  number={1},
  pages={tyad020},
  year={2024},
  publisher={Oxford University Press}
}

@inproceedings{GrubbsLR17,
  author    = {Grubbs, Paul and Lu, Jiahui and Ristenpart, Thomas},
  title     = {{Message Franking via Committing Authenticated Encryption}},
  booktitle = {Advances in Cryptology -- CRYPTO 2017},
  series    = {Lecture Notes in Computer Science},
  volume    = {10403},
  pages     = {66--97},
  publisher = {Springer},
  year      = {2017},
  doi       = {10.1007/978-3-319-63697-9_3},
  url       = {https://eprint.iacr.org/2017/664}
}

@inproceedings{issa2022hecate,
  title={Hecate: Abuse reporting in secure messengers with sealed sender},
  author={Issa, Rawane and Alhaddad, Nicolas and Varia, Mayank},
  booktitle={31st USENIX Security Symposium (USENIX Security 22)},
  pages={2335--2352},
  year={2022}
}

@article{N32,
    author = {Abelson, Harold and Anderson, Ross and Bellovin, Steven M. and Benaloh, Josh and Blaze, Matt and Diffie, Whitfield and Gilmore, John and Green, Matthew and Landau, Susan and Neumann, Peter G. and Rivest, Ronald L. and Schiller, Jeffrey I. and Schneier, Bruce and Specter, Michael A. and Weitzner, Daniel J.},
    title = { Keys under doormats: mandating insecurity by requiring government access to all data and communications },
    journal = {Journal of Cybersecurity},
    volume = {1},
    number = {1},
    pages = {69-79},
    year = {2015},
    month = {09},
    issn = {2057-2085},
    doi = {10.1093/cybsec/tyv009},
    url = {https://doi.org/10.1093/cybsec/tyv009}
}

@misc{N33,
  author       = {{England and Wales Court of Appeal, Criminal Division}},
  title        = {{A, B, D and C v. Regina, [2021] EWCA Crim 128}},
  howpublished = {\url{https://www.judiciary.uk/judgments/a-b-d-c-v-regina/}},
  year         = {2021}
}

@misc{N34,
  author       = {{Court of Justice of the European Union}},
  title        = {{Judgment of the Court in Case C-670/22, M.N. (EncroChat)}},
  howpublished = {\url{https://eur-lex.europa.eu/legal-content/EN/TXT/?uri=CELEX:62022CJ0670_RES}},
  year         = {2024}
}

@misc{N37,
  author       = {{Meta}},
  title        = {{Information for law enforcement authorities}},
  howpublished = {\url{https://www.meta.com/safety/communities/law/guidelines/}},
  year         = {2026}
}

@misc{N38,
  author       = {{Signal Messenger}},
  title        = {{Government Requests}},
  howpublished = {\url{https://signal.org/bigbrother/}},
  year         = {2026}
}

@misc{N39,
  author       = {The Guardian},
  title        = {{Facebook gave police their private data. Now, this duo face abortion charges}},
  howpublished = {\url{https://www.theguardian.com/us-news/2022/aug/10/facebook-user-data-abortion-nebraska-police}},
  year         = {2022}
}

@misc{N40,
  author       = {{Apple}},
  title        = {{Legal Process Guidelines: Government and Law Enforcement within the United States}},
  howpublished = {\url{https://www.apple.com/legal/privacy/law-enforcement-guidelines-us.pdf}},
  year         = {2025}
}

@misc{N41,
  author       = {Tech Crunch},
  title        = {{US Government Loses Bid to Force Facebook to Wiretap Messenger Calls}},
  howpublished = {\url{https://techcrunch.com/2018/09/28/us-government-loses-bid-to-force-facebook-to-wiretap-messenger-calls/}},
  year         = {2018}
}

@misc{N42,
  author       = {{Electronic Frontier Foundation}},
  title        = {{EFF, ACLU v. DOJ -- Facebook Messenger Unsealing}},
  howpublished = {\url{https://www.eff.org/cases/eff-aclu-v-doj-facebook-messenger-unsealing?language=en}},
  year         = {2018}
}

@misc{N43,
  author       = {{National Institute of Standards and Technology}},
  title        = {{Security and Privacy Controls for Information Systems and Organizations}},
  howpublished = {\url{https://csrc.nist.gov/pubs/sp/800/53/r5/upd1/final}},
  year         = {2020}
}

@misc{N44,
  title        = {{RFC 9162: Certificate Transparency Version 2.0}},
  howpublished = {\url{https://www.rfc-editor.org/rfc/rfc9162}},
  author={Laurie, Ben and Messeri, Eran and Stradling, Rob},
  year={2021},
  publisher={RFC Editor}
}

@inproceedings{N45,
  title={Leakage-abuse attacks against searchable encryption},
  author={Cash, David and Grubbs, Paul and Perry, Jason and Ristenpart, Thomas},
  booktitle={Proceedings of the 22nd ACM SIGSAC conference on computer and communications security},
  pages={668--679},
  year={2015}
}

@inproceedings{N46,
  title={Generic attacks on secure outsourced databases},
  author={Kellaris, Georgios and Kollios, George and Nissim, Kobbi and O'neill, Adam},
  booktitle={Proceedings of the 2016 ACM SIGSAC Conference on Computer and Communications Security},
  pages={1329--1340},
  year={2016}
}

@misc{N47,
  author       = {Perrin, Trevor and Marlinspike, Moxie and Schmidt, Rolfe},
  title        = {{The Double Ratchet Algorithm}},
  howpublished = {\url{https://signal.org/docs/specifications/doubleratchet/}},
  year         = {2025},
  note         = {accessed July 30, 2026}
}

@misc{N48,
  author       = {Barnes, Richard and Beurdouche, Benjamin and Robert, Raphael and Millican, Jon and Omara, Emad and Cohn-Gordon, Katriel},
  title        = {{RFC 9420: The Messaging Layer Security (MLS) Protocol}},
  howpublished = {\url{https://www.rfc-editor.org/rfc/rfc9420}},
  year         = {2023},
  note         = {RFC 9420}
}

@incollection{goldreich1987mental,
  title={How to play any mental game, or a completeness theorem for protocols with honest majority},
  author={Goldreich, Oded and Micali, Silvio and Wigderson, Avi},
  booktitle={Providing sound foundations for cryptography: on the work of Shafi Goldwasser and Silvio Micali},
  pages={307--328},
  year={2019}
}

@inproceedings{demmler2015aby,
  title={ABY-A framework for efficient mixed-protocol secure two-party computation.},
  author={Demmler, Daniel and Schneider, Thomas and Zohner, Michael},
  booktitle={Ndss},
  year={2015}
}

@inproceedings{rotaru2019marbled,
  title={Marbled circuits: Mixing arithmetic and boolean circuits with active security},
  author={Rotaru, Dragos and Wood, Tim},
  booktitle={International Conference on Cryptology in India},
  pages={227--249},
  year={2019},
  organization={Springer}
}

@inproceedings{escudero2020mixed,
  title={Improved primitives for MPC over mixed arithmetic-binary circuits},
  author={Escudero, Daniel and Ghosh, Satrajit and Keller, Marcel and Rachuri, Rahul and Scholl, Peter},
  booktitle={Annual international cryptology conference},
  pages={823--852},
  year={2020},
  organization={Springer}
}

@inproceedings{zyskind2025threshold,
  title={High-Throughput Universally Composable Threshold FHE Decryption},
  author={Zyskind, Guy and Zarchy, Doron and Leibovich, Max and Peikert, Chris},
  booktitle={Proceedings of the 2025 ACM SIGSAC Conference on Computer and Communications Security},
  pages={2339--2353},
  year={2025}
}

@article{gliwa-etal-2019-samsum,
  author       = {Bogdan Gliwa and
                  Iwona Mochol and
                  Maciej Biesek and
                  Aleksander Wawer},
  title        = {SAMSum Corpus: {A} Human-annotated Dialogue Dataset for Abstractive
                  Summarization},
  journal      = {CoRR},
  volume       = {abs/1911.12237},
  year         = {2019},
  url          = {http://arxiv.org/abs/1911.12237},
  eprinttype   = {arXiv},
  eprint       = {1911.12237},
  bibsource    = {dblp computer science bibliography, https://dblp.org}
}

@incollection{HomomorphicEncryptionSecurityStandard,
  title={Homomorphic encryption standard},
  author={Albrecht, Martin and Chase, Melissa and Chen, Hao and Ding, Jintai and Goldwasser, Shafi and Gorbunov, Sergey and Halevi, Shai and Hoffstein, Jeffrey and Laine, Kim and Lauter, Kristin and others},
  booktitle={Protecting privacy through homomorphic encryption},
  pages={31--62},
  year={2022},
  publisher={Springer}
}

@misc{ofcomMessaging2023,
  author       = {{Ofcom}},
  title        = {WhatsAppening in the World of Online Communications?},
  year         = {2023},
  howpublished = {Ofcom research report},
  url          = {https://www.ofcom.org.uk/internet-based-services/technology/whatsappening-in-the-world-of-online-communications}
}

@misc{onsPopulation2022,
  author       = {{Office for National Statistics}},
  title        = {Population Estimates for the UK, England, Wales,
                  Scotland, and Northern Ireland: Mid-2022},
  year         = {2024},
  url          = {https://www.ons.gov.uk/peoplepopulationandcommunity/populationandmigration/populationestimates/bulletins/annualmidyearpopulationestimates/mid2022}
}

@article{lyddy2014analysis,
  title={An analysis of language in university students' text messages},
  author={Lyddy, Fiona and Farina, Francesca and Hanney, James and Farrell, Lynn and Kelly O'Neill, Niamh},
  journal={Journal of Computer-Mediated Communication},
  volume={19},
  number={3},
  pages={546--561},
  year={2014},
  publisher={Oxford University Press Oxford, UK}
}

@software{wordfreq,
  author       = {Robyn Speer},
  title        = {rspeer/wordfreq: v3.0},
  month        = sep,
  year         = 2022,
  publisher    = {Zenodo},
  version      = {v3.0.2},
  doi          = {10.5281/zenodo.7199437},
  url          = {https://doi.org/10.5281/zenodo.7199437}
}

@article{brysbaert-new-2009,
  title={Moving beyond Ku{\v{c}}era and Francis: A critical evaluation of current word frequency norms and the introduction of a new and improved word frequency measure for American English},
  author={Brysbaert, Marc and New, Boris},
  journal={Behavior research methods},
  volume={41},
  number={4},
  pages={977--990},
  year={2009},
  publisher={Springer}
}

@misc{metaWhatsAppThreeBillion2026,
  author       = {{Meta}},
  title        = {{It's Time to Reserve Your WhatsApp Username}},
  year         = {2026},
  howpublished = {\url{https://about.fb.com/news/2026/06/its-time-to-reserve-your-whatsapp-username/}},
  note         = {Accessed July 30, 2026}
}

@inproceedings{keller2016mascot,
  title={MASCOT: faster malicious arithmetic secure computation with oblivious transfer},
  author={Keller, Marcel and Orsini, Emmanuela and Scholl, Peter},
  booktitle={Proceedings of the 2016 ACM SIGSAC conference on computer and communications security},
  pages={830--842},
  year={2016}
}

@article{baum2020ole,
  title={Efficient protocols for oblivious linear function evaluation from ring-LWE},
  author={Baum, Carsten and Escudero, Daniel and Pedrouzo-Ulloa, Alberto and Scholl, Peter and Troncoso-Pastoriza, Juan Ram{\'o}n},
  journal={Journal of Computer Security},
  volume={30},
  number={1},
  pages={39--78},
  year={2022},
  publisher={SAGE Publications Sage UK: London, England}
}

@inproceedings{bruggemann2023flute,
  title={FLUTE: fast and secure lookup table evaluations},
  author={Br{\"u}ggemann, Andreas and Hundt, Robin and Schneider, Thomas and Suresh, Ajith and Yalame, Hossein},
  booktitle={2023 IEEE Symposium on Security and Privacy (SP)},
  pages={515--533},
  year={2023},
  organization={IEEE}
}

@article{gui2024leakage,
  title={Analyze Your Leakage! Security Analysis of Encryption Schemes for Substring Search},
  author={Gui, Zichen and Paterson, Kenneth G and Patranabis, Sikhar},
  journal={Cryptology ePrint Archive},
  year={2024}
}

@inproceedings{bonte2020homomorphic,
  title={Homomorphic string search with constant multiplicative depth},
  author={Bonte, Charlotte and Iliashenko, Ilia},
  booktitle={Proceedings of the 2020 ACM SIGSAC Conference on Cloud Computing Security Workshop},
  pages={105--117},
  year={2020}
}

@article{narisada2026tfhe, 
title={Efficient Homomorphic String Search via TFHE}, 
volume={1}, 
ISSN={3083-2454}, 
DOI={10.53941/pc.2026.100013}, 
number={2}, 
journal={Pragmatic Cybersecurity}, 
publisher={Scilight Press}, 
author={Narisada, Shintaro and Okada, Hiroki and Nishide, Takashi and Fukushima, Kazuhide}, 
year={2026},
pages={13} }

@article{secoaguirre2026text,
  title={Extending homomorphic algorithms for encrypted text comparison},
  author={Seco-Aguirre, I{\~n}aki and Regueiro, Cristina and Bernab{\'e}-Rodr{\'\i}guez, Julen and Jacob, Eduardo},
  journal={Scientific Reports},
  url={https://doi.org/10.1038/s41598-026-48255-2},
  year={2026},
  publisher={Nature Publishing Group UK London}
}

@inproceedings{yasuda2013secure,
  title={Secure pattern matching using somewhat homomorphic encryption},
  author={Yasuda, Masaya and Shimoyama, Takeshi and Kogure, Jun and Yokoyama, Kazuhiro and Koshiba, Takeshi},
  booktitle={Proceedings of the 2013 ACM workshop on Cloud computing security workshop},
  pages={65--76},
  year={2013}
}

@inproceedings{yasuda2014practical,
  author       = {Masaya Yasuda and
                  Takeshi Shimoyama and
                  Jun Kogure and
                  Kazuhiro Yokoyama and
                  Takeshi Koshiba},
  editor       = {Joaqu{\'{\i}}n Garc{\'{\i}}a{-}Alfaro and
                  Georgios V. Lioudakis and
                  Nora Cuppens{-}Boulahia and
                  Simon N. Foley and
                  William M. Fitzgerald},
  title        = {Practical Packing Method in Somewhat Homomorphic Encryption},
  booktitle    = {Data Privacy Management and Autonomous Spontaneous Security - 8th
                  International Workshop, {DPM} 2013, and 6th International Workshop,
                  {SETOP} 2013, Egham, UK, September 12-13, 2013, Revised Selected Papers},
  series       = {Lecture Notes in Computer Science},
  volume       = {8247},
  pages        = {34--50},
  publisher    = {Springer},
  year         = {2013},
  url          = {https://doi.org/10.1007/978-3-642-54568-9\_3},
  doi          = {10.1007/978-3-642-54568-9\_3},
  bibsource    = {dblp computer science bibliography, https://dblp.org}
}

@inproceedings{ishimaki2017private,
  title={Private substring search on homomorphically encrypted data},
  author={Ishimaki, Yu and Imabayashi, Hiroki and Yamana, Hayato},
  booktitle={2017 IEEE International Conference on Smart Computing (SMARTCOMP)},
  pages={1--6},
  year={2017},
  organization={IEEE}
}

\appendix

\begin{figure*}[t]
    \centering
    \begin{minipage}[t]{0.64\textwidth}
        \centering
        \begin{subfigure}[t]{0.35\textwidth}
            \centering
            \includegraphics[width=\textwidth]
                {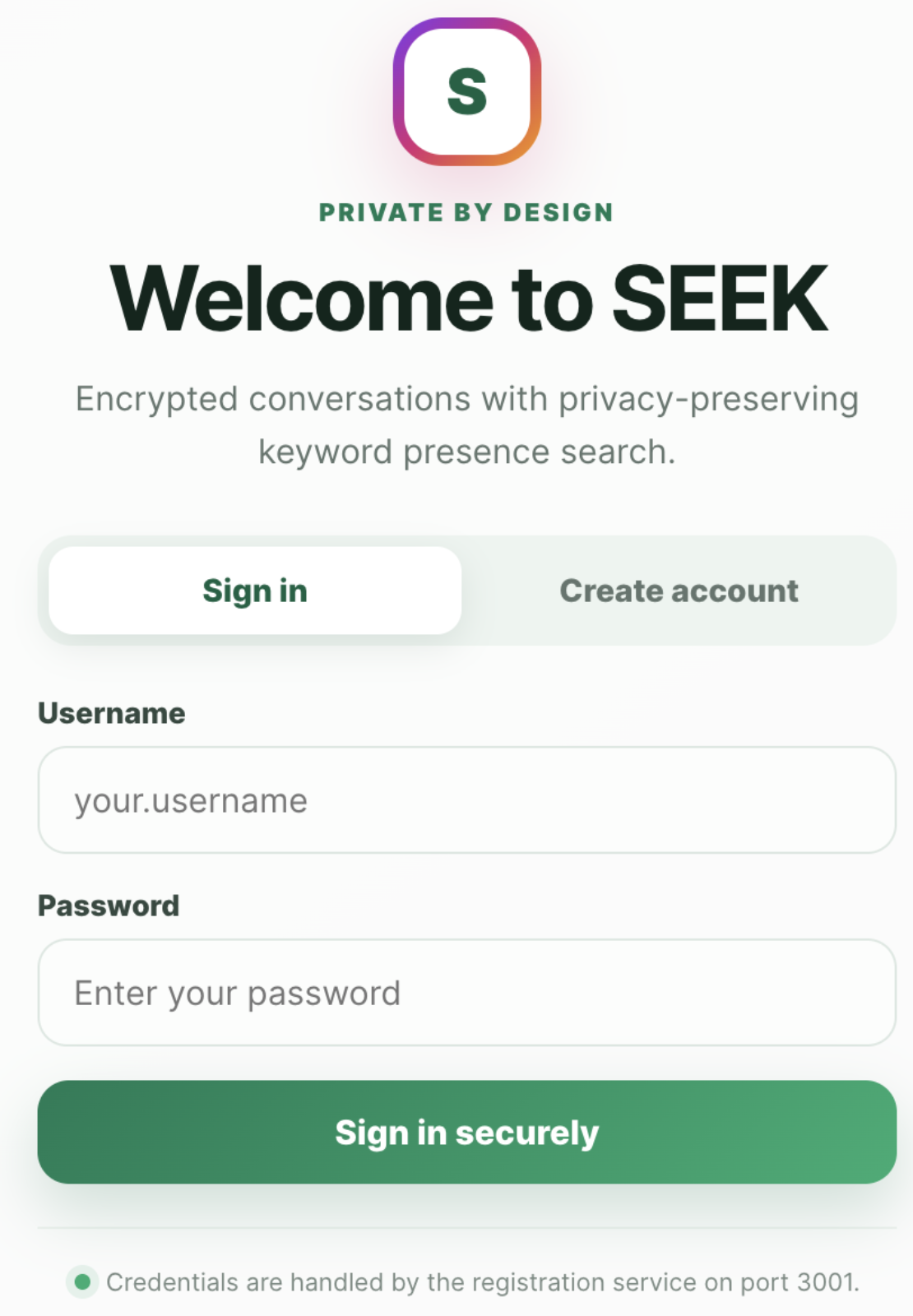}
            \caption{Login Interface.}
            \label{fig:login-interface}
        \end{subfigure}
        \hfill
        \begin{subfigure}[t]{0.315\textwidth}
            \centering
            \includegraphics[width=\textwidth]
                {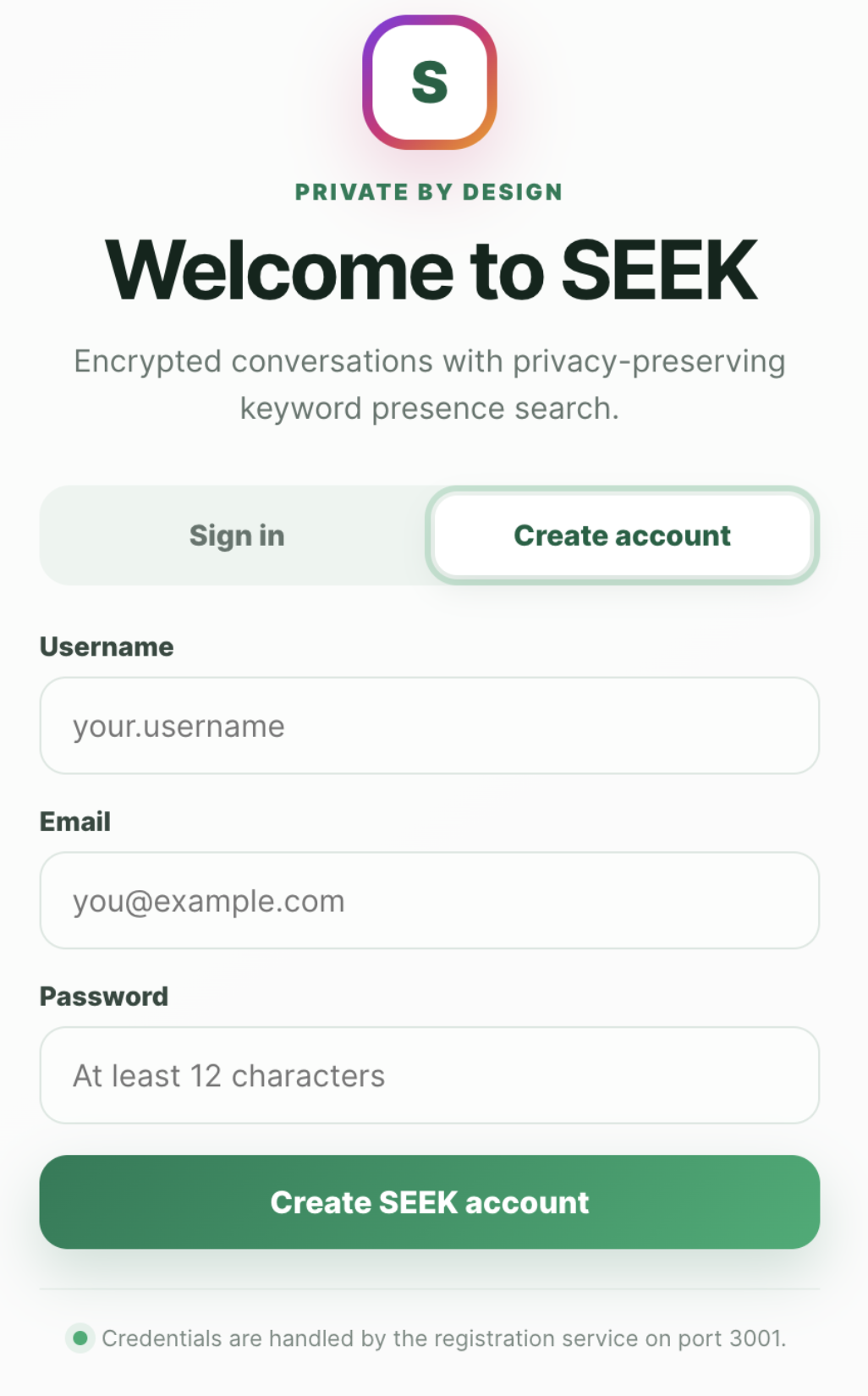}
            \caption{Registration Interface.}
            \label{fig:registration-interface}
        \end{subfigure}
        \hfill
        \begin{subfigure}[t]{0.315\textwidth}
            \centering
            \includegraphics[width=\textwidth]
                {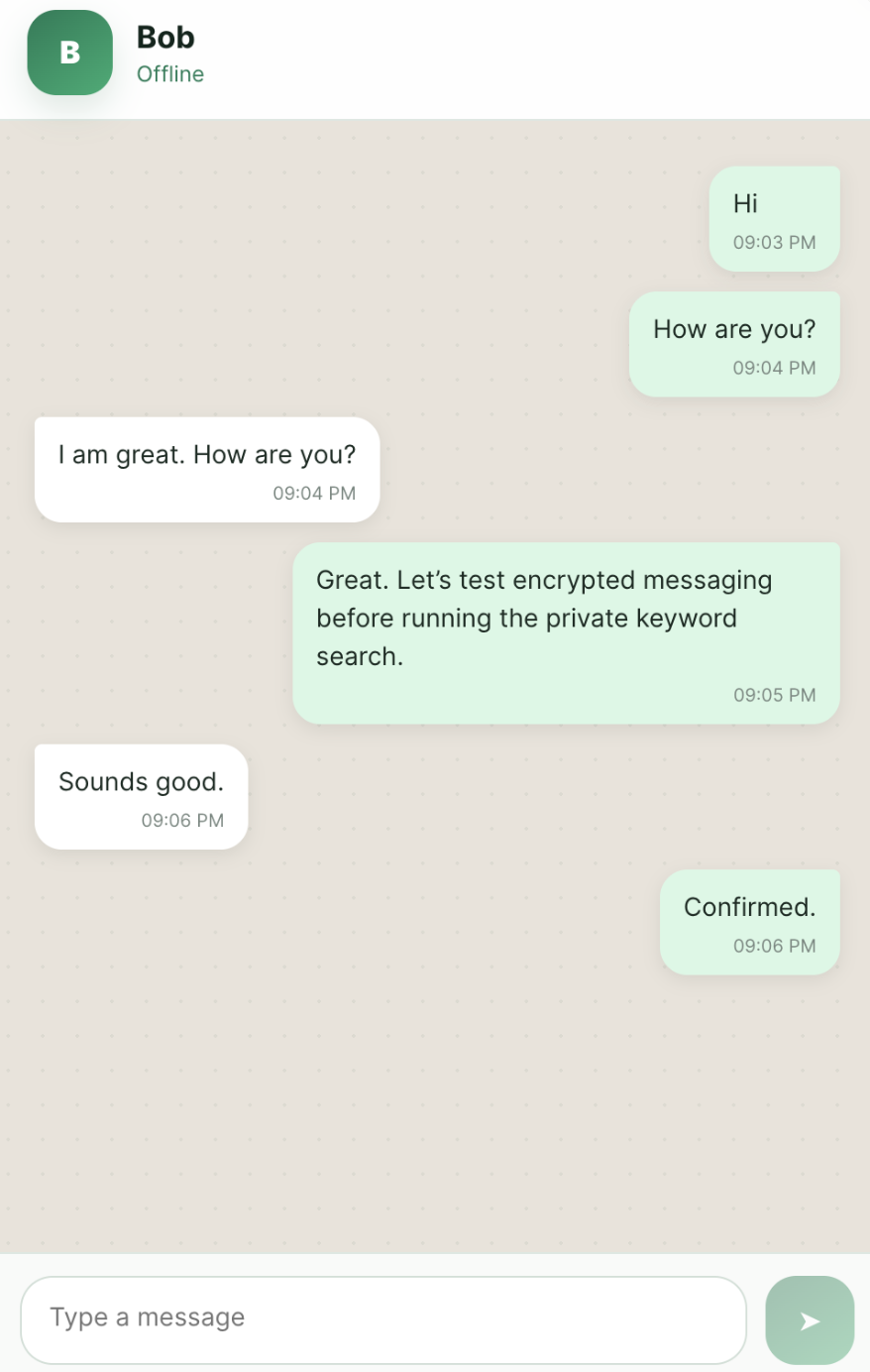}
            \caption{Chat Interface.}
            \label{fig:chat-interface}
        \end{subfigure}
    \caption{Client interfaces for login, registration, and encrypted messaging.}
    \label{fig:client-interfaces}
    \end{minipage}
    \hfill
    \begin{minipage}[t]{0.35\textwidth}
        \centering
        \includegraphics[width=\linewidth]
            {\detokenize{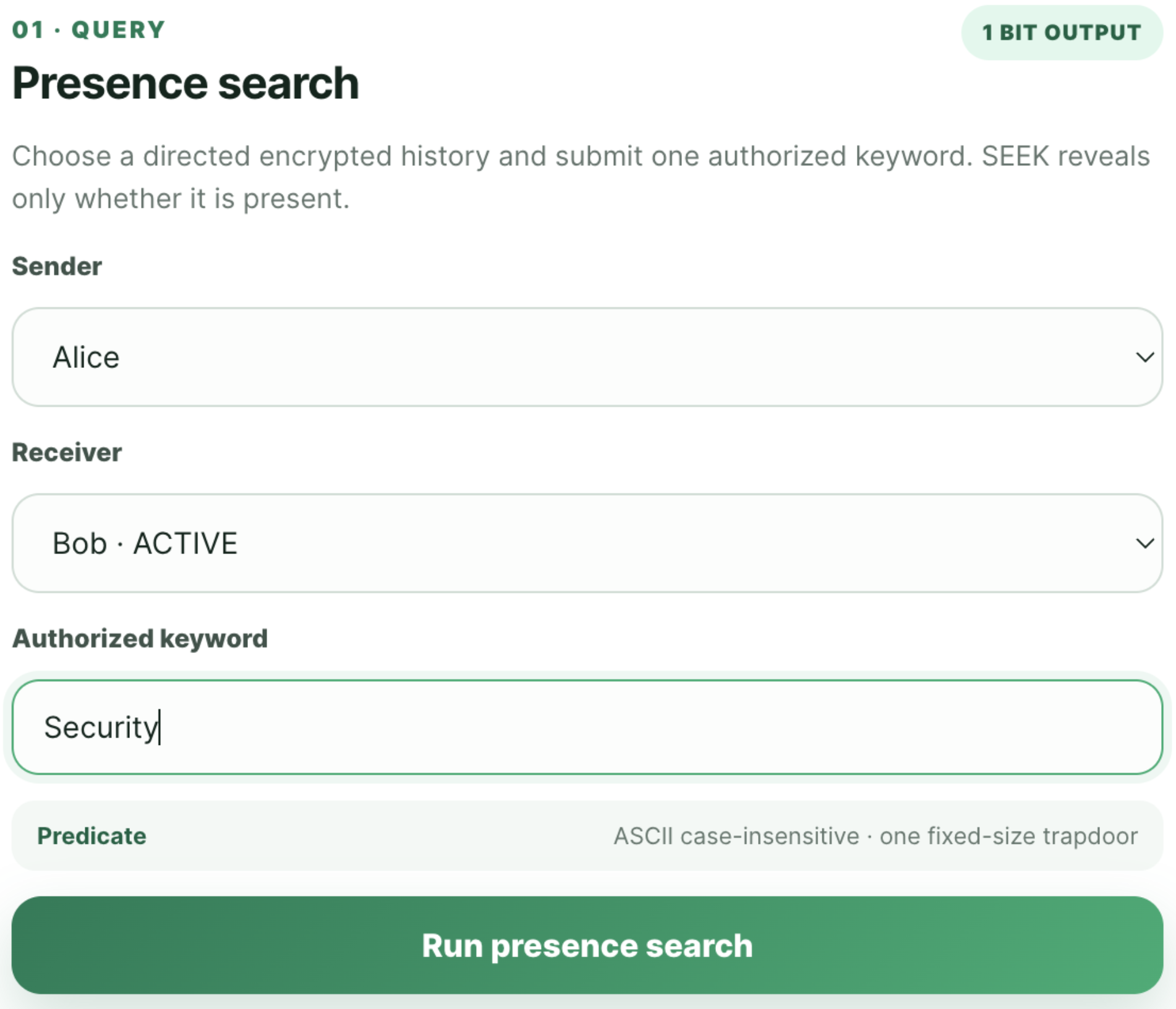}}
        \caption{Search Interface: Requires a directed message history of targeted end users, and an ASCII case-insensitive keyword.}
        \label{fig:search-interface}
    \end{minipage}
\end{figure*}

\begin{figure*}[t]
    \centering
    \begin{subfigure}[t]{0.17\linewidth}
        \centering
        \includegraphics[width=\linewidth]
            {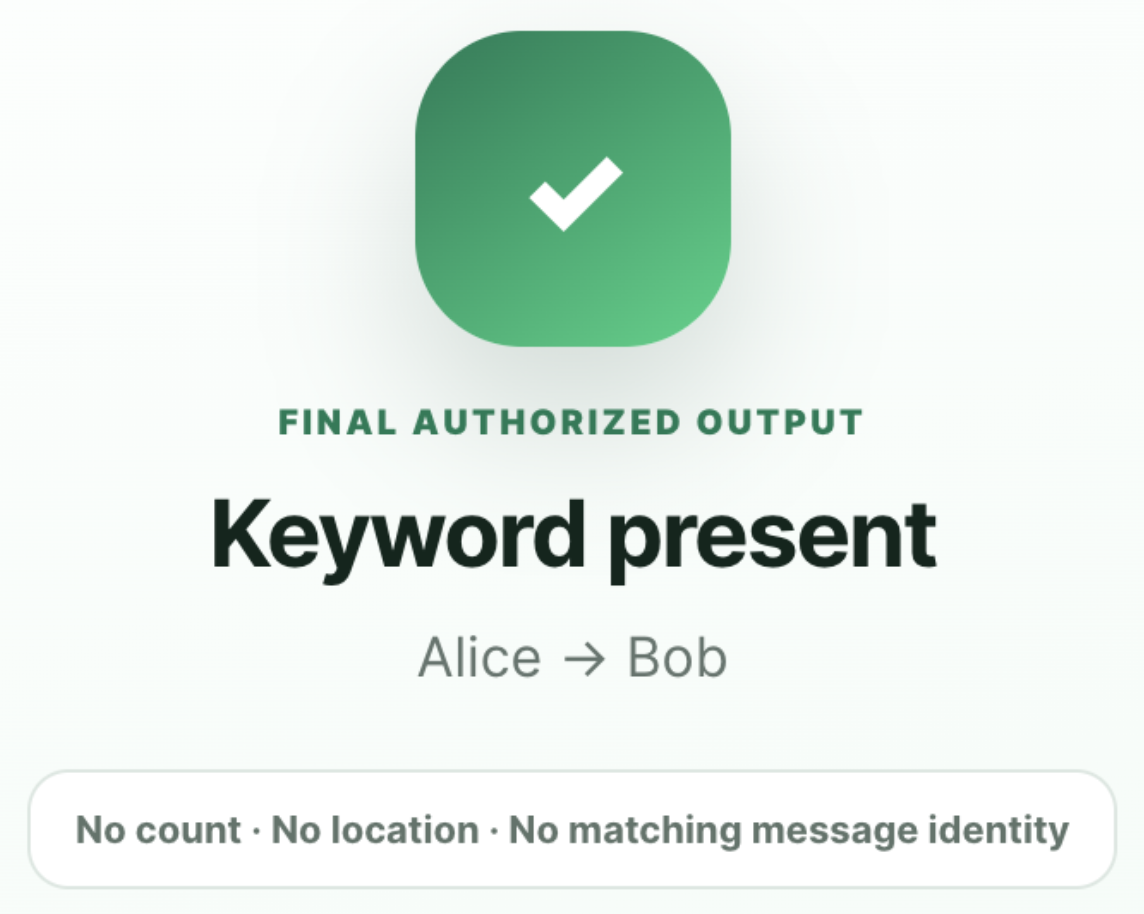}
        \caption{Successful Search.}
        \label{fig:keyword-present}
    \end{subfigure}
    \hfill
    \begin{subfigure}[t]{0.17\linewidth}
        \centering
        \includegraphics[width=\linewidth]
            {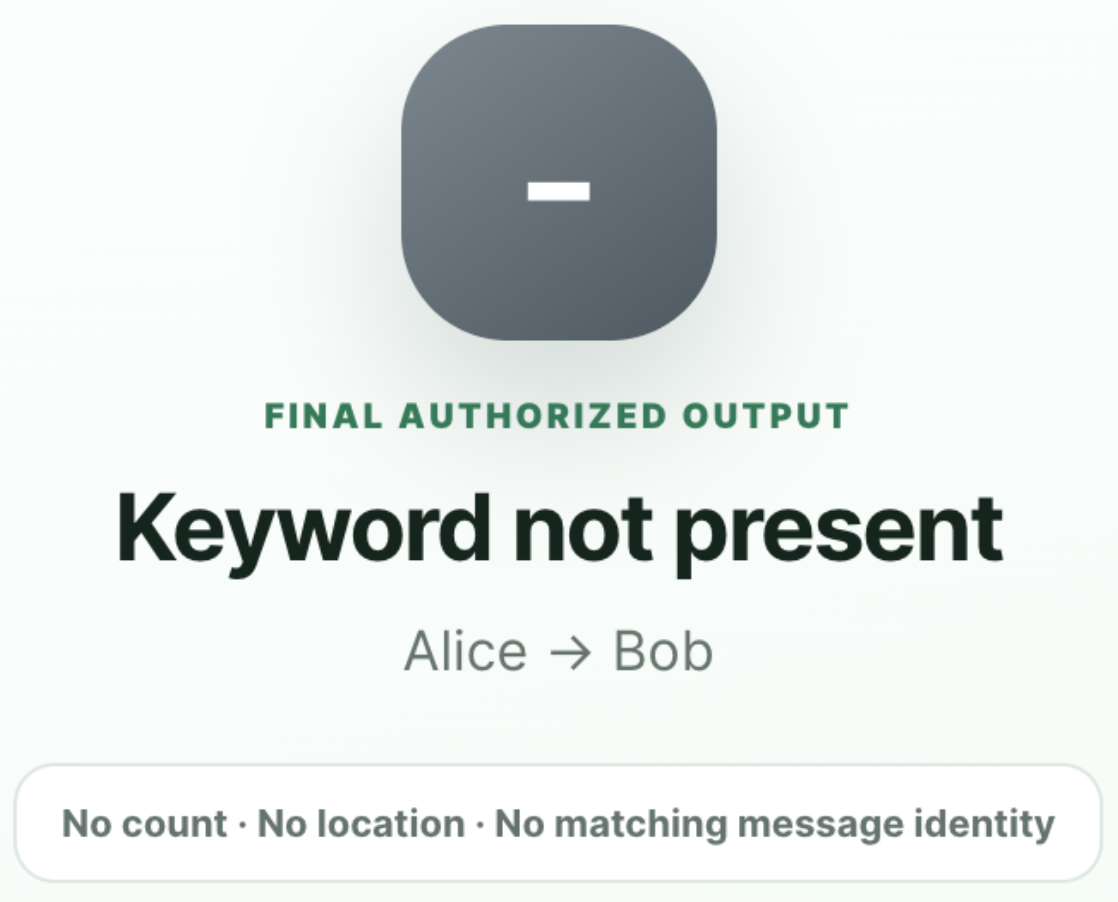}
        \caption{Unsuccessful Search.}
        \label{fig:keyword-absent}
    \end{subfigure}
    \hfill
    \begin{subfigure}[t]{0.65\linewidth}
        \centering
        \includegraphics[width=\linewidth]
            {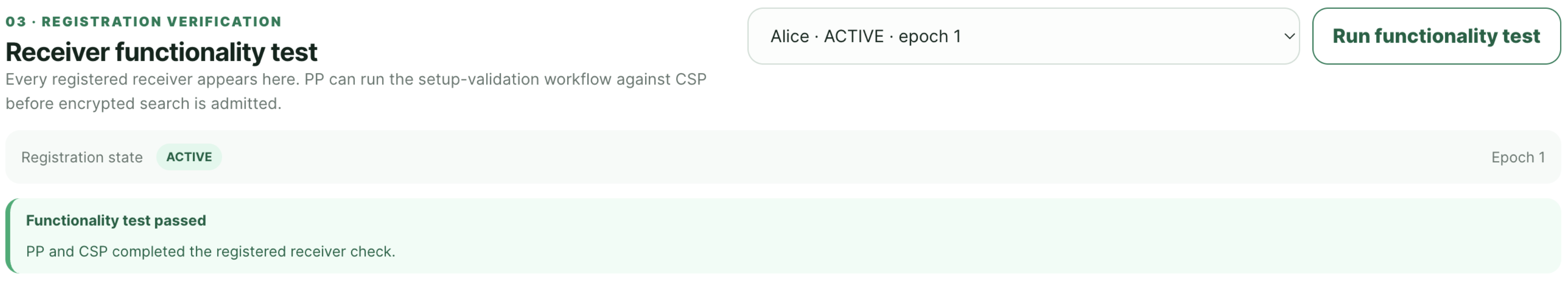}
        \caption{Functional validation of the end users' setup artifacts, as described in Section~\ref{subsec:functional-validation}.}
        \label{fig:setup-functional-validation}
    \end{subfigure}

    \caption{\ProtocolName{} application outputs for keyword-presence search and functional validation of end-user setup artifacts.}
    \label{fig:application-search-validation}
\end{figure*}

\section{Prototype \ProtocolName{} Application Details}
\label{applicationImages}

\noindent \ding{110} \paragraphNew{\textbf{\DeepUL{Prototype Implementation}}}
To quantify end-to-end latency on commodity devices and validate deployability within an encrypted-messaging workflow, a proof-of-concept web and cross-platform mobile application was developed for the \ProtocolName{} protocol. The prototype consists of four logical components: (i) a combined Chat/Registration layer for authentication, receiver setup, certified public-key distribution, and encrypted-message communication; (ii) a $\mathsf{PP}$ service for trapdoor generation, search coordination, and final-result reconstruction; (iii) a $\mathsf{CSP}$ service for encrypted storage and search evaluation; and (iv) sender/receiver endpoints implemented as web and mobile messaging clients. The React web client utilizes a bundled \textsf{node-seal} backend~\cite{node_seal}, which exposes the coefficient-domain BFV plaintext operations required by \ProtocolName{} through WebAssembly. The cross-platform mobile client is built in React Native using Expo SDK~\cite{expo}. Since React Native does not execute WebAssembly directly, the mobile client invokes the same embedded SEAL WASM stack inside a hidden WebView via a structured message-passing interface. Both clients verify the receiver's certified public key, encrypt outgoing message fragments locally, and decrypt received messages exclusively at the endpoint. The $\mathsf{PP}$ and $\mathsf{CSP}$ components are implemented as Python services using FastAPI, with core cryptographic operations executed by native helpers built on \textsf{Microsoft SEAL}. The components communicate through authenticated REST APIs, and Socket.IO supports real-time messaging. For each authorized search, $\mathsf{PP}$ pins the selected directed corpus and generates a fixed-size encrypted trapdoor, while $\mathsf{CSP}$ applies the encrypted transformation $m\mapsto1-2m$, evaluates BFV correlation, performs relinearization and modulus switching, and retains its local BFV shares. The parties directly generate fresh Boolean multiplication triples~\cite{demmler2015aby}, doubly authenticated bits (daBits)~\cite{rotaru2019marbled,escudero2020mixed}, scalar oblivious linear evaluation (OLE) zero-test tokens~\cite{baum2020ole}, and arithmetic Beaver multiplication triples using RSA-based oblivious transfers (OTs) and extensions. These resources support joint-mask GMW, selected decoding with daBits, scalar-OLE zero testing, and balanced Beaver product aggregation. Only the final product-tree root is reconstructed at $\mathsf{PP}$ and mapped to a single presence bit, whereas no match position, count, correlation score, or matching-message identity is revealed, and $\mathsf{CSP}$ receives no result. The four components communicate as distinct networked roles within a containerized local deployment with different endpoints, and the unified client architecture supports both iOS and Android.

\noindent \ding{110}
\paragraphNew{\textbf{\DeepUL{Interfaces and Workflow}}}
Fig.~\ref{fig:client-interfaces} and Fig.~\ref{fig:search-interface} illustrate the authentication, encrypted messaging, and authorized search interfaces. Fig.~\ref{fig:application-search-validation} displays the search outputs and setup diagnostics. Upon successful registration, a receiver epoch is automatically provisioned and validated, enabling the client to access the chat interface. Each outgoing message is independently fragmented and encrypted locally using the receiver's certified BFV public key, transmitted via Chat services to $\mathsf{CSP}$, and decrypted and reassembled exclusively by the receiver. As a result, searches do not span message boundaries. In Fig.~\ref{fig:search-interface}, the $\mathsf{PP}$ operator selects a sender $\mathsf{S}$, a receiver $\mathsf{R}$, and an ASCII case-insensitive keyword. The direction $\mathsf{S}\rightarrow\mathsf{R}$ refers solely to the ordered history of messages sent from $\mathsf{S}$ to $\mathsf{R}$ within the selected receiver epoch, excluding the reverse history. $\mathsf{PP}$ submits a single fixed-size encrypted trapdoor for this history, after which the online search is conducted with $\mathsf{CSP}$ without receiver involvement. As depicted in Fig.~\ref{fig:keyword-present} and Fig.~\ref{fig:keyword-absent}, the only corpus-dependent output is the bit $\mathsf{present}\in\{0,1\}$, shown as \emph{Keyword present} or \emph{Keyword not present}. No match count, location, or correlation score is revealed, and $\mathsf{CSP}$ obtains no result. Finally, Fig.~\ref{fig:setup-functional-validation} presents an optional diagnostic for rerunning the $\mathsf{PP}$-$\mathsf{CSP}$ validation of an active receiver setup. 

\section{Comparison with Prior Encrypted Schemes}
\label{ComparisonWithHE}

This appendix expands upon Section~\ref{RelatedWork} by evaluating \ProtocolName{} in comparison with representative HE-based pattern and keyword-matching constructions. The fragmentation-based RLWE scheme~\cite{1} is identified as the closest construction-level baseline to \ProtocolName{}, enabling direct comparison of fragmentation and operation counts. Other constructions, such as \texttt{CipherMatch}~\cite{7}, constant-depth BGV~\cite{bonte2020homomorphic}, BWT-based FHE substring search~\cite{ishimaki2017private}, suffix-array TFHE~\cite{narisada2026tfhe}, and CKKS-based comparison~\cite{secoaguirre2026text}, utilize distinct encodings, evaluation models, correctness regimes, and output interfaces. Therefore, their suitability for \ProtocolName{}'s incremental messaging history, role-separated, presence-only setting is assessed qualitatively rather than through direct runtime comparison.

Fragmentation-based RLWE matching~\cite{1} extends the boundary-coverage strategy introduced in~\cite{1:BCC20}. Here, $\phi$ represents the base-fragment length, and $\eta=\lceil L/\phi\rceil$ for a message of length $L$. The scheme in~\cite{1:BCC20} requires $\phi\geq2(L_{\max}-1)$ and generates $\eta$ base fragments along with $\eta-1$ boundary fragments, resulting in a total of $2\eta-1$ logical fragments. In the padded two-grid instantiation, these fragments form two length-$\phi$ grids offset by $\phi/2$. 
The RLWE refinement~\cite{1} uploads only the $\eta$ non-overlapping base ciphertexts. Its matching algorithm constructs $\eta-1$ adjacent-pair ciphertexts, and applies packed Hamming-distance matching to each pair. In contrast, under its fixed 8-bit encoding and byte-aligned fragmentation, \ProtocolName{} fragments each message independently and employs the minimal sufficient overlap $L_{\max}-8$. For a message of length $L_i$, it stores $\mathcal J_i$ ciphertext fragments, defined as follows (refer Section \ref{subsec:fragEncrypt}): 
$
  \mathcal J_i=
  \max\!\left\{
    1,
    \left\lceil
      {(L_i-L_{\max}+8)}/
           ({N-L_{\max}+8})
    \right\rceil
  \right\}
$
Therefore, for messages with lengths $L_1,\ldots,L_M$, the total number of fragments is given by $\mathcal J=\sum_{i=1}^{M}\mathcal J_i$.

\ding{42} \paragraphNew{\DeepUL{{Example}}} Consider a message with $L=131072$ and $L_{\max}=512$. For~\cite{1}, our comparison uses $\phi=L_{\max}=512$, so $\eta=256$: scheme~\cite{1} uploads $256$ base ciphertexts and evaluates $255$ adjacent-pair instances. For the padded instantiation of~\cite{1:BCC20}, let $\phi=2L_{\max}=1024$, so $\eta=128$ and the scheme materializes $2\eta-1=255$ logical fragments. \ProtocolName{} uses $\mathcal J=37$ for $N=4096$ and $\mathcal J=5$ for $N=2^{15}$. Relative to~\cite{1}, this reduces the stored-ciphertext count by $6.92\times$ and $51.2\times$, respectively, and the per-query matching-instance count by $6.89\times$ and $51\times$. Relative to~\cite{1:BCC20}, the corresponding logical-fragment-footprint reductions are $6.89\times$ and $51\times$.

As specified in~\cite{1}, forming each adjacent-pair ciphertext uses one ct-pt multiplication by the public monomial and one ciphertext addition. Its normal match algorithm then uses one ct-ct multiplication, two ct-pt polynomial multiplications, and two ciphertext additions/subtractions per pair.
The published wildcard match used in our case-insensitive comparison (Section~\ref{sec:eval_microbenchmarks}) has a core of two ct--ct multiplications, two ct--pt multiplications, including the public scalar-by-2 multiplication, and two ciphertext additions/subtractions per pair. Including adjacent-pair construction, the complete evaluated wildcard path therefore performs two ct--ct multiplications, three ct--pt multiplications, and three ciphertext additions/subtractions per pair. This wildcard extension can realize the same ASCII case-insensitive predicate. The distinction of \ProtocolName{} is therefore not the predicate itself, but its combined construction using one trapdoor ciphertext, one query-dependent ct-ct multiplication per stored fragment, fewer evaluated instances, and presence-only aggregation. Adding a presence-only output layer to~\cite{1} can align its leakage with \ProtocolName{}, but does not reduce its per-query adjacent-pair evaluations.

\texttt{CipherMatch}~\cite{7} addresses client-server exact binary string matching by packing $b$ bits per BFV plaintext coefficient, complementing and replicating the query, encrypting multiple left-shifted query polynomials, and adding each variant to every encrypted database block. A matching $b$-bit chunk produces the all-ones value $2^b-1$. Let $S(L_P)$ represent the number of shifted variants and $B$ the number of encrypted database blocks, where $B=\lceil L/(Nb)\rceil$ for a single $L$-bit sequence. The matching core uploads $S(L_P)$ query ciphertexts once and reuses them across the $B$ blocks, resulting in $S(L_P)B$ ciphertext additions at the server. Since~\cite{7} specifies only eight variants for its 8-bit example and does not provide a general formula for $S(L_P)$, a query-length-hiding adaptation must pad to a public bound $S_{\max}$. This approach requires $\Theta(S_{\max})$ query ciphertexts and $\Theta(S_{\max}B)$ server additions. Furthermore,~\cite{7} does not address matches that cross separately encrypted database polynomials, and its server natively generates and returns exact match locations. Consequently, a presence-only E2EE adaptation necessitates explicit boundary handling and secure result aggregation. \textit{In contrast, \ProtocolName{} employs a single fixed-size encrypted trapdoor and retrieves all within-fragment offsets using one ct-ct multiplication per fragment, thereby avoiding shifted-query amplification.}

Constant-depth BGV search~\cite{bonte2020homomorphic} encodes one character per SIMD slot and represents an $N_T$-character text using an $(M,k)$-cover, where $k$ is the slot count and $M<k$ is the server-visible pattern-length parameter. It stores $r=\lceil(N_T-M+1)/(k-M+1)\rceil$ length-$k$ chunks with $M-1$ characters of overlap, ensuring unique coverage of length-$M$ candidates. The search invokes its randomized $\mathsf{HomEQ}$ circuit $rM$ times. Thus, although its multiplicative depth is independent of $M$, its operation count remains pattern-length dependent through repeated multiplications, rotations, and Frobenius maps. An unequal candidate is falsely accepted with probability $t^{-d}$, giving whole-search correctness of at least $(1-t^{-d})^{r(k-M+1)}$. Its wildcard extension adds one ct-ct multiplication per equality evaluation but does not natively implement ASCII case-insensitive matching. The compressed output preserves all occurrence locations for client decryption. Although the server can derive a new cover from existing ciphertexts, this requires additional selections, rotations, and additions. The construction provides neither a length-hiding fixed-bound execution nor canonical pruning of duplicate shorter-pattern candidates. \textit{In contrast, \ProtocolName{} fixes fragmentation and public pruning using $L_{\max}$, privately masks candidates incomplete for the hidden $L_P$, and releases only a corpus-wide presence-or-absence bit.}

The private substring search of~\cite{ishimaki2017private} evaluates encrypted substring queries over encrypted data using SIMD batching and a BWT-based representation. However, this construction does not specify mechanisms for updating independently encrypted, incrementally arriving messages or for supporting a role-separated, corpus-wide presence-only interface. Supporting \ProtocolName{}'s setting would therefore require rebuilding or updating the BWT representation, or maintaining separate structures and privately aggregating their outputs. Moreover, TFHE-based search~\cite{narisada2026tfhe} replaces exhaustive alignment testing with homomorphic binary search over an encrypted suffix array. Prior to encryption, the data owner pads the entire text, constructs its suffix array in plaintext, and encrypts both the text and the decomposed suffix-array tables. Considering the padded pattern and text as $\widetilde P$ and $\widetilde T$, CMux-tree lookups and bootstrapped lexicographic comparisons achieve $O(|\widetilde P|\log|\widetilde T|)$ matching complexity while concealing the actual lengths within publicly known padded bounds. Similarly, the construction specifies no dynamic index-update procedure, so incrementally arriving messages would require rebuilding/updating the suffix array, or maintaining separate indexes and aggregating their outputs. Its query contains one LWE ciphertext per encoded pattern character. Under the proposed extension, each character expands into four encoded characters and therefore four query ciphertexts, resulting in increased runtime and memory use. The native output is one encrypted matching position or nothing, rather than role-separated corpus-wide presence bit. \textit{\ProtocolName{} instead directly supports incrementally encrypted histories, native case-insensitive matching, and one fixed-size trapdoor ciphertext.}

CKKS-based comparison~\cite{secoaguirre2026text} packs an ASCII-encoded target and a replicated, zero-padded pattern into CKKS slots. The evaluator receives the target bound $n$, pattern length $L_P$, and padding metadata in the clear, then examines the resulting $n-L_P+1$ alignments using rotations. Across these alignments, two stages of Chebyshev-polynomial approximation implement absolute-value and nonzero testing, after which a balanced multiplication tree produces an encrypted containment value intended to be zero for presence and one for absence. Because this predicate is approximate, near-ASCII mismatches can produce values undesirably close to zero and increase false-positive risk. Increasing the polynomial degree and multiplicative depth improves the distinguishability but raises runtime and memory. A single invocation is also limited by the available CKKS slots, so longer or incrementally accumulated histories require additional fragmentation, boundary handling, and cross-ciphertext aggregation. \textit{In contrast, \ProtocolName{} deterministically evaluates overlap-complete fragments, hides $L_P$ within public $L_{\max}$, and releases corpus-wide presence through role-separated two-party computation.}

\section{Preprocessing and Selected Decoding}
\label{app:seldec-realization}

This section realizes the ideal selected-decoding preprocessing relation $\mathcal F_{\mathsf{SD}}^{\mathsf{pre}}$ through $\Pi_{\mathsf{SD}}^{\mathsf{pre}}$ and, using its outputs, realizes $\mathcal F_{\mathsf{SD}}$ from Section~\ref{subsec:presence-selected-decoding} through $\Pi_{\mathsf{SD}}^{\mathsf{on}}$. The former combines local scaled decomposition, semi-honest GMW, and one daBit conversion per selected coefficient, while the latter uses GMW comparisons and daBit-based Boolean-to-arithmetic conversion~\cite{goldreich1987mental,demmler2015aby,rotaru2019marbled,escudero2020mixed}. Thus, Boolean triples and daBits are the only standard input-independent resources required within selected decoding. The construction follows the multiparty-BFV paradigm~\cite{mouchet2021multiparty} and the masked-opening and private-rounding techniques of MPC threshold-FHE decryption~\cite{zyskind2025threshold}.

\noindent \ding{110} \textbf{\DeepUL{Scaled Quotient and Remainder.}} For $u\in\mathbb Z_{Q'}$, let $u^+=\operatorname{can}_{Q'}(u)\in\{0,\ldots,Q'-1\}$ denote its canonical integer representative. We define the \emph{scaled quotient} $k_u$ and the
\emph{scaled remainder} $\rho_u$ as the quotient and remainder obtained by
dividing the integer $t u^+$ by $Q'$, such that: $t u^+=k_u Q'+\rho_u$, for $0\leq\rho_u<Q'$. Equivalently, $k_u = \lfloor{t u^+}/{Q'}\rfloor$ and $\rho_u = t u^+-k_uQ'$. We denote this quotient-remainder pair by $\mathsf{ScaleDiv}_{Q',t}(u)= (k_u,\rho_u)$. Since $0\leq u^+<Q'$, these values satisfy $0\leq k_u<t$ and 
$0\leq\rho_u<Q'$. For a nonnegative integer $\rho<Q'$, the
notation $\llbracket\rho\rrbracket_{\mathsf B}$ denotes componentwise Boolean
XOR shares of its fixed-length binary representation. Let $\ell_Q=\lceil\log_2 Q'\rceil$, so every $\rho\in\{0,\ldots,Q'-1\}$ has an $\ell_Q$-bit representation.

\noindent \ding{110} \textbf{\DeepUL{Two-party Offline
Preprocessing.}} For parameters $(Q',t)$, every $\xi\in\mathcal I$, $\mathcal F_{\mathsf{SD}}^{\mathsf{pre}}$ receives $r_\xi^{\mathsf{PP}},r_\xi^{\mathsf{CSP}}\in\mathbb Z_{Q'}$, and computes:
$$r_\xi=r_\xi^{\mathsf{PP}}+r_\xi^{\mathsf{CSP}}\!\!\!\!\!\pmod{Q'},(k_{r,\xi},\,\,\rho_{r,\xi})=\mathsf{ScaleDiv}_{Q',t}(r_\xi)$$
To share these outputs, it samples $k_{r,\xi}^{\mathsf{PP}}{\leftarrow}\mathbb Z_t$ and $\rho_{r,\xi,i}^{\mathsf{PP}}{\leftarrow}\{0,1\}$ independently for every bit $i$, sets $k_{r,\xi}^{\mathsf{CSP}}$ and $\rho_{r,\xi,i}^{\mathsf{CSP}}$ as follows and returns only the local shares. 
\[
k_{r,\xi}^{\mathsf{CSP}}
 =k_{r,\xi}-k_{r,\xi}^{\mathsf{PP}}\!\!\!\pmod t,\,\,\,
\rho_{r,\xi,i}^{\mathsf{CSP}}
 =\rho_{r,\xi,i}\oplus\rho_{r,\xi,i}^{\mathsf{PP}}
\]
All sampling is independent across components, coefficients, and executions. Protocol $\Pi_{\mathsf{SD}}^{\mathsf{pre}}$ realizes this relation as follows. 
\begin{enumerate}
    \item Each party $A\in\{\mathsf{PP},\mathsf{CSP}\}$ independently samples $r_\xi^A\gets\mathbb Z_{Q'}$ and locally computes $(k_\xi^A,\rho_\xi^A)=\mathsf{ScaleDiv}_{Q',t}(r_\xi^A)$.
    \item Using GMW on private $\ell_Q$-bit values $\rho_\xi^{\mathsf{PP}}$ and $\rho_\xi^{\mathsf{CSP}}$, the parties compute Boolean shares of $c_\xi$ and $\rho_{r,\xi}$, as: 
    \[S_\xi=\rho_\xi^{\mathsf{PP}}+\rho_\xi^{\mathsf{CSP}},\, c_\xi=\mathbf 1[S_\xi\geq Q'],\,\rho_{r,\xi}=S_\xi-c_\xi Q'\]
    \item Parties convert only $\llbracket c_\xi\rrbracket_{\mathsf B}$ to arithmetic shares $\langle c_\xi\rangle_t$ using one daBit and locally set $k_{r,\xi}^A=k_\xi^A+c_\xi^A\pmod t$.
\end{enumerate}

\noindent The resulting local state is $(r_\xi^A,k_{r,\xi}^A, \llbracket\rho_{r,\xi}\rrbracket_{\mathsf B}^A)$ and neither party reconstructs $r_\xi$, $c_\xi$, $k_{r,\xi}$, or $\rho_{r,\xi}$.

\begin{lemma}[Correctness and Security of $\Pi_{\mathsf{SD}}^{\mathsf{pre}}$]
\label{lem:seldec-prep-realization}
Given fresh, non-reused Boolean triples and daBits and semi-honest-secure GMW and Boolean-to-arithmetic conversion, $\Pi_{\mathsf{SD}}^{\mathsf{pre}}$ securely realizes $\mathcal F_{\mathsf{SD}}^{\mathsf{pre}}$ against a static semi-honest adversary corrupting at most one of $\mathsf{PP}$ and $\mathsf{CSP}$.
\end{lemma}

\begin{proof}
For $A\in\{\mathsf{PP},\mathsf{CSP}\}$, let
$(r_\xi^A)^+=\operatorname{can}_{Q'}(r_\xi^A)$. Local decomposition gives
$t(r_\xi^A)^+=k_\xi^A Q'+\rho_\xi^A$. Let
\[
\omega_\xi=
\mathbf 1[(r_\xi^{\mathsf{PP}})^+
+(r_\xi^{\mathsf{CSP}})^+\geq Q'] \,\,\,\text{and}\,\,\, r_\xi^+=\operatorname{can}_{Q'}(r_\xi)
\]
Then
$r_\xi^+=(r_\xi^{\mathsf{PP}})^++(r_\xi^{\mathsf{CSP}})^+
-\omega_\xi Q'$. Because
$S_\xi=c_\xi Q'+\rho_{r,\xi}$, adding the two local decompositions gives:
\[
t r_\xi^+=(k_\xi^{\mathsf{PP}}+k_\xi^{\mathsf{CSP}}+c_\xi-t\omega_\xi)Q'+\rho_{r,\xi}
\]
Uniqueness of Euclidean division therefore implies:
\[
k_{r,\xi}
=
k_\xi^{\mathsf{PP}}+k_\xi^{\mathsf{CSP}}
+c_\xi-t\omega_\xi
\equiv
k_\xi^{\mathsf{PP}}+k_\xi^{\mathsf{CSP}}+c_\xi
\pmod t
\]
This is exactly the reconstruction of the arithmetic shares, while the GMW output reconstructs the required $\rho_{r,\xi}$. All interactions occur within the GMW circuit, and the daBit-based conversion of $c_\xi$ is local. Mask sampling, scaled decomposition, and quotient-share formation are performed locally. The semi-honest simulators, together with sequential and parallel composition, therefore simulate the batched protocol view. Fresh GMW output masks ensure uniformity of the Boolean remainder shares, while the independent daBit mask guarantees uniformity of the arithmetic quotient shares, resulting in the distribution specified by $\mathcal F_{\mathsf{SD}}^{\mathsf{pre}}$.
\end{proof}

\noindent \ding{110} \textbf{\DeepUL{Online Masked Opening and Secure Rounding.}}
Using the local state produced by $\Pi_{\mathsf{SD}}^{\mathsf{pre}}$, for each $\xi\in\mathcal I$ and $A\in\{\mathsf{PP},\mathsf{CSP}\}$, the parties compute and reconstruct only:
\begin{equation}
\begin{aligned}
\bar v_\xi^A
&=p_\xi^A+r_\xi^A\pmod{Q'}\\
\bar v_\xi
&=\bar v_\xi^{\mathsf{PP}}+\bar v_\xi^{\mathsf{CSP}}
=v_\xi+r_\xi\pmod{Q'}
\end{aligned}
\label{eq:seldec-masked-opening}
\end{equation}
Because $r_\xi$ is uniform, hidden, and used once, the opened value $\bar v_\xi$ is uniform in $\mathbb Z_{Q'}$ for every fixed $v_\xi$. Since $\bar v_\xi$ is public, both parties locally compute $(k_{\bar v,\xi},\rho_{\bar v,\xi})=\mathsf{ScaleDiv}_{Q',t}(\bar v_\xi)$. Using the GMW and mixed-domain primitives of Section~\ref{sec:beaver}, they then privately compute Boolean shares of:
\begin{equation}
\begin{aligned}
b_\xi
&=
\mathbf 1
[
\rho_{\bar v,\xi}<\rho_{r,\xi}
],\,\, \rho_{v,\xi}
=
(
\rho_{\bar v,\xi}-\rho_{r,\xi}
)
\bmod Q'
\nonumber
\end{aligned}
\end{equation}
$$
h_\xi
=
\mathbf 1
[
\rho_{v,\xi}
\geq
\lceil{Q'}/{2}\rceil
]
$$
The bit $b_\xi$ is the borrow in modular subtraction, while $h_\xi$ implements BFV nearest-integer rounding. Because $Q'$ is odd, the rounding comparison has no tie. Using the two daBits, the parties convert the Boolean shares of $b_\xi$ and $h_\xi$ into arithmetic shares $\langle b_\xi\rangle_t$ and $\langle h_\xi\rangle_t$. They share the public value $k_{\bar v,\xi}$ as $(k_{\bar v,\xi}^{\mathsf{PP}},k_{\bar v,\xi}^{\mathsf{CSP}})= (k_{\bar v,\xi}\bmod t,0)$ and locally compute:
\begin{equation}
m_\xi^A
=
k_{\bar v,\xi}^A-k_{r,\xi}^A-b_\xi^A+h_\xi^A
\pmod t
\label{eq:seldec-decoded-shares}
\end{equation}
Finally, $\mathsf{PP}$ sends the rebased share $\mathsf{reb}_\xi = m_\xi^{\mathsf{PP}}+a_\xi \pmod t$ to $\mathsf{CSP}$, which outputs $w_\xi$ as follows:
\begin{equation}
w_\xi
=
m_\xi^{\mathsf{CSP}}+\mathsf{reb}_\xi
\pmod t
\label{eq:seldec-rebased-output}
\end{equation}
The protocol gives no designated output to $\mathsf{PP}$ and processes all indices independently in one batch.

\begin{lemma}[Semi-Honest Security of Selected Decoding]
\label{lem:seldec-realization-security}
In the $\mathcal F_{\mathsf{SD}}^{\mathsf{pre}}$-hybrid model, for odd $Q'$ and semi-honest-secure GMW comparison and Boolean-to-arithmetic conversion, $\Pi_{\mathsf{SD}}^{\mathsf{on}}$ securely realizes $\mathcal F_{\mathsf{SD}}$ against a static semi-honest adversary corrupting at most one party.
\end{lemma}

\begin{proof}
For each fixed $v_\xi$, the honest party's fresh mask share ensures that $\bar v_\xi=v_\xi+r_\xi\pmod{Q'}$ is uniformly distributed in $\mathbb Z_{Q'}$. Therefore, the simulator samples $\bar v_\xi$ uniformly and selects the honest opening share. The GMW and Boolean-to-arithmetic views are simulatable under their assumed semi-honest security. If $\mathsf{PP}$ is corrupted, $\mathcal F_{\mathsf{SD}}$ produces no output, and its local values and outgoing rebased share are determined by its input and the simulated subprotocol views. If $\mathsf{CSP}$ is corrupted, the simulator receives $w_\xi$, simulates $m_\xi^{\mathsf{CSP}}$, and sets $\mathsf{reb}_\xi = w_\xi - m_\xi^{\mathsf{CSP}} \pmod t$. This approach preserves the real conditional distribution and reconstructs the prescribed $w_\xi$. Independence across indices and parallel composition together yield the claimed batched realization.
\end{proof}

\noindent \ding{110} \textbf{\DeepUL{Separate Zero-Test Token Generation.}}
The zero-test-token construction and its share invariant are defined in Section~\ref{subsec:presence-zero-target}. Its $K=|\mathcal I|$ OLE instances are generated batch-wise and independently of $\Pi_{\mathsf{SD}}^{\mathsf{pre}}$. Each batch is bound to
$(\mathsf{receiverID},e,Q',t,\mathcal I,\mathsf{batchID})$, delivered before its associated online execution, consumed once, and then erased. The input-independent backend also supplies the fresh Boolean triples and two daBits required to compute $b_\xi$ and $h_\xi$, together with the arithmetic triples required by the aggregation tree.

\noindent \ding{110} \textbf{\DeepUL{Cost Analysis.}}
Let $\ell_t=\lceil\log_2t\rceil$ and $K=|\mathcal I|$. The following counts cover $\Pi_{\mathsf{SD}}^{\mathsf{pre}}$, $\Pi_{\mathsf{SD}}^{\mathsf{on}}$, and online Beaver aggregation. Generation of the required OLEs, Boolean triples, daBits, and arithmetic triples is accounted for separately in the evaluation. For each retained coefficient, the ripple-carry implementation of $\Pi_{\mathsf{SD}}^{\mathsf{pre}}$ uses $2\ell_Q$ Boolean AND gates for addition, $\ell_Q+1$ for public-$Q'$ subtraction, and $\ell_Q$ for conditional reduction. It therefore consumes $4\ell_Q+1$ Boolean triples and one daBit. Since each GMW AND opening transmits two logical bits per party and the daBit conversion reveals one masked bit, this stage transmits $K(8\ell_Q+3)$ logical bits per party.
Protocol $\Pi_{\mathsf{SD}}^{\mathsf{on}}$ consumes another $4\ell_Q+1$ Boolean triples and two daBits per coefficient, bringing the selected-decoding totals to $K(8\ell_Q+2)$ Boolean triples and $3K$ daBits. Before fixed-width serialization and batch-wise byte padding, the online stage transmits approximately $K(9\ell_Q+4+\ell_t)$ bits from $\mathsf{PP}$ and $K(9\ell_Q+4)$ bits from $\mathsf{CSP}$. These totals include one $\ell_Q$-bit masked-value share from each party, two bits per party for each of the $4\ell_Q+1$ GMW AND gates, two daBit-opening bits per party, and one $\ell_t$-bit rebased share from $\mathsf{PP}$. Including $\Pi_{\mathsf{SD}}^{\mathsf{pre}}$ increases the respective totals to $K(17\ell_Q+7+\ell_t)$ and $K(17\ell_Q+7)$ bits before aggregation. For $K\geq1$, the Beaver tree protocol transmits an additional $2(K-1)\ell_t$ bits per party and one final $\ell_t$-bit root share from $\mathsf{CSP}$.

\section{Correctness Analysis of \ProtocolName{}}
\label{app:seek-correctness}

This appendix contains the proofs of Lemma \ref{lem:presence-fragment-coverage}, \ref{lem:presence-correlation-correctness}, \ref{lem:presence-no-wrap}, \ref{lem:seldec-realization-correctness}, \ref{lem:presence-decoding-aggregation}, and Theorem \ref{thm:presence-e2e-correctness}, as stated in Section \ref{sec:seek-correctness}.

\begin{proof}[\DeepUL{\textbf{Proof of Lemma~\ref{lem:presence-fragment-coverage}}}]
Consider a message of length $L$ and a complete byte-aligned query window beginning at $u\in\{0,8,\ldots,L-L_P\}$. By Section~\ref{subsec:fragEncrypt}, fragment $j$ begins at $\mathsf{off}_j=j\cdot\mathsf{stride}$, where $\mathsf{stride}=N-L_{\max}+8$. We define the relevant fragment index as $j = \min \{\lfloor{u}/{\mathsf{stride}}\rfloor,\mathcal J-1\}$.  This definition selects the latest available fragment whose starting offset does not exceed $u$, and therefore $\mathsf{off}_j\leq u$. We first consider the case in which $j<\mathcal J-1$. By definition, the window starts before the next fragment offset, so $u<\mathsf{off}_j+\mathsf{stride}$. Since both positions are byte aligned, $u\leq\mathsf{off}_j+\mathsf{stride}-8$, we have:
\begin{align*}
u+L_P-1
&\leq u+L_{\max}-1 \,\,(\text{as} \,\,L_P\leq L_{\max})\\
&\leq \mathsf{off}_j+\mathsf{stride}-8+L_{\max}-1 = \mathsf{off}_j+N-1
\end{align*}
Thus, the window is contained in fragment $j$. If $j=\mathcal J-1$, validity of the window and $T_j=L-\mathsf{off}_j$ give:
\begin{align*}
u+L_P-1
&\leq L-1 = \mathsf{off}_j+T_j-1
\end{align*}
So the final fragment also contains it. The overlap is minimal within the class of byte-aligned
fragmentations. Suppose instead that the overlap were
$O<L_{\max}-8$, and let two consecutive fragments begin at $a$ and
$b=a+N-O$. A length-$L_{\max}$ window beginning at $b-8$ starts
before the second fragment, while its final bit satisfies the following:
\begin{align*}
(b-8)+(L_{\max}-1)
&=a+N-1+(L_{\max}-O-8)>a+N-1
\end{align*}
It therefore extends beyond the first fragment and begins before the second, so no fragment contains it.

It remains to prove exact-once retention. Let $e=u+L_P-1$ be the window's absolute end position. Because $u$ and $L_P$ are multiples of eight, $e\equiv7\pmod 8$. Hence, $e$ satisfies
the byte-alignment condition in Eq.~\ref{eq:presence-open-list}. Fragment $0$ retains all its eligible byte ends, whereas each fragment $j>0$ retains those with $s\geq L_{\max}-1$. Every fragment preceding another is full, and the definition of $\mathcal J$ ensures that every noninitial final fragment has $T_j\geq L_{\max}$. Moreover, the first retained absolute end of any fragment $j>0$ is:
\begin{align*}
\mathsf{off}_j+L_{\max}-1
&=\mathsf{off}_{j-1}+\mathsf{stride}+L_{\max}-1 \\
&=\mathsf{off}_{j-1}+N+7
\end{align*}
This is the next byte-aligned end after the preceding fragment's last eligible end $\mathsf{off}_{j-1}+N-1$. Thus, the retained
byte-end sets are disjoint and consecutive and partition
$\{7,15,\ldots,L-1\}$. Consequently, $e$ has a unique representation $e=\mathsf{off}_{j^\star}+s^\star$, for $(j^\star,s^\star)\in\mathcal I$. If $j^\star=0$, then $s^\star=e\geq L_P-1$. Otherwise,
$s^\star\geq L_{\max}-1\geq L_P-1$. Therefore,
\begin{align*}
s_0^\star
&=s^\star-(L_P-1)\geq0, 
\mathsf{off}_{j^\star}+s_0^\star=e-(L_P-1)=u
\end{align*}
Since $s^\star<T_{j^\star}$, this unique retained coefficient
represents precisely the original complete window beginning at $u$.
\end{proof}

\begin{proof}[\DeepUL{\textbf{Proof of Lemma~\ref{lem:presence-correlation-correctness}}}]
Consider a complete candidate window beginning at relative offset
$s_0$ in fragment $j$. The reversed query polynomial places query
coefficient $i$ at degree $L_P-i-1$. Therefore, the products of the
aligned message and query coefficients contribute to output
coefficient $s=s_0+L_P-1.$ The unreduced product
$u_j(X)P_{\mathsf{CI}}(X)$ has degree at most
$N+L_P-2$, where $u_j(X)=\sum_{i=0}^{T_j-1}(1-2b_i^{(j)})X^i$. Since $s\geq L_P-1$, we have
$s+N\geq N+L_P-1$, which exceeds the maximum degree of the unreduced
product. Thus, no term of degree $s+N$ exists, and reduction modulo
$X^N+1$ introduces no negacyclic contribution at coefficient $s$.
Accordingly, its plaintext value is defined as in Eq. \ref{eq:presence-correlation-semantics}.
For every constrained position with $\gamma_i=1$, the corresponding
signed product is $+1$ when the two bits agree and $-1$ when they
differ. Positions with $\gamma_i=0$ contribute zero. Let $h= \mathsf{Ham}_{\gamma} (\mathsf{Win}_{j,s_0}, \mathbf p )$. Among the $L_{\mathsf{fix}}$ constrained positions, exactly $h$ differ
and $L_{\mathsf{fix}}-h$ agree. It follows that
$z_j[s]=
(L_{\mathsf{fix}}-h)-h=
L_{\mathsf{fix}}-2h$, which is Eq.~\ref{eq:presence-correlation-semantics}. Therefore,
$z_j[s]=L_{\mathsf{fix}}$ if and only if $h=0$. For an alphabetic query byte, $\gamma_i$ removes only the bit of
numerical weight $32$. Uppercase and lowercase encodings of the same ASCII letter differ exactly in this bit, while all other bits remain equal. Hence, $h=0$ holds precisely when the candidate matches $P$ under the ASCII case-insensitive predicate.
\end{proof}

\begin{proof}[\DeepUL{\textbf{Proof of Lemma~\ref{lem:presence-no-wrap}}}]
Each coefficient of the negacyclic correlation polynomial is a signed
sum containing at most one contribution from every constrained query
position. Each nonzero contribution has magnitude one. Therefore, for
every retained coefficient, including coefficients affected by
negacyclic reduction, we have
$
|z_j[s]|
\leq
L_{\mathsf{fix}}
\leq
L_{\max}.
$
Suppose first that $s\geq L_P-1$, so that $\xi=(j,s)$ represents a
complete candidate. By
Lemma~\ref{lem:presence-correlation-correctness} and the real target
$\theta_\xi=L_{\mathsf{fix}}$, the zero-test value reduces to
$x_\xi=
z_j[s]-L_{\mathsf{fix}}=
-2
\mathsf{Ham}_{\gamma}
(
\mathsf{Win}_{j,s_0},
\mathbf p
)$.
Consequently,
$x_\xi\in[-2L_{\mathsf{fix}},0]\subseteq[-2L_{\max},0]$, and
$x_\xi=0$ if and only if the candidate is a valid case-insensitive
match.

Now suppose that $s<L_P-1$. The protocol assigns the dummy target
$\theta_\xi=L_{\max}+1$. Combining this target with
$-L_{\mathsf{fix}}\leq z_j[s]\leq L_{\mathsf{fix}}$ provides the following observation:
\[
\begin{aligned}
x_\xi
&=
z_j[s]-(L_{\max}+1)\\
&\in
[-(L_{\mathsf{fix}}+L_{\max}+1),
\,L_{\mathsf{fix}}-L_{\max}-1]\\ &\subseteq
[-(2L_{\max}+1),-1]
\end{aligned}
\]
Thus, an early coefficient is always nonzero as an integer. Because $t>4L_{\max}+2$ and $t$ is odd, the centered representative
range of $\mathbb Z_t$ contains the complete interval $[-(2L_{\max}+1),\,2L_{\max}+1].$ All correlation and zero-test values considered above therefore have
their intended centered representatives modulo $t$. In particular,
modular reduction cannot map any nonzero incomplete or nonmatching
candidate to zero. This proves both no-wrap and target-soundness.
\end{proof}

\begin{proof}[\DeepUL{\textbf{Proof of Lemma~\ref{lem:seldec-realization-correctness}}}]
Let $\xi\in\mathcal I$, $v_\xi^+=\operatorname{can}_{Q'}(v_\xi)$,
$r_\xi^+=\operatorname{can}_{Q'}(r_\xi)$, and
$\bar v_\xi^+=\operatorname{can}_{Q'}(\bar v_\xi)$. From
Eq.~\ref{eq:seldec-masked-opening}, there exists
$\lambda_\xi\in\{0,1\}$ such that:
$v_\xi^+
=
\bar v_\xi^+-r_\xi^++\lambda_\xi Q'
$
By the definitions of scaled quotient, remainder and the borrow bit, we have:
\[
t\bar v_\xi^+
=
k_{\bar v,\xi}Q'+\rho_{\bar v,\xi},\,\,\
tr_\xi^+
=
k_{r,\xi}Q'+\rho_{r,\xi}
\]
\[
\rho_{v,\xi}
=
\rho_{\bar v,\xi}-\rho_{r,\xi}+b_\xi Q'
=
\left(
\rho_{\bar v,\xi}-\rho_{r,\xi}
\bmod Q'
\right)
\]
Therefore $0\leq\rho_{v,\xi}<Q'$. Substituting the two Euclidean decompositions gives the following:
\[
tv_\xi^+
=
\left(
k_{\bar v,\xi}
-
k_{r,\xi}
-
b_\xi
+
\lambda_\xi t
\right)Q'
+
\rho_{v,\xi}
\]
Since $0\leq\rho_{v,\xi}<Q'$, the definition of $h_\xi$ shows
exactly whether nearest-integer rounding adds one to the quotient.
Moreover, replacing the canonical representative $v_\xi^+$ by the
centered representative $\operatorname{ctr}_{Q'}(v_\xi)$ changes the
scaled value by an integer multiple of $t$, which vanishes modulo $t$.
Therefore, applying the definition of
$\mathsf{DecCoeff}_{Q',t}$ (Section~\ref{subsec:bfv}) yields:
\begin{equation}
\mathsf{DecCoeff}_{Q',t}(v_\xi)
=
k_{\bar v,\xi}
-
k_{r,\xi}
-
b_\xi
+
h_\xi
\pmod t
\label{eq:presence-secure-decoding-identity}
\end{equation}
By functional correctness of the preprocessing shares, GMW comparisons, and
daBit conversions, Eq.~\ref{eq:seldec-decoded-shares} therefore satisfies:
\[
m_\xi^{\mathsf{PP}}+m_\xi^{\mathsf{CSP}}
=
\mathsf{DecCoeff}_{Q',t}(v_\xi)
\pmod t
\]
Combining this equality with Eq.~\ref{eq:seldec-rebased-output} gives $w_\xi$ as:
\[
\begin{aligned}
w_\xi
&=
m_\xi^{\mathsf{CSP}}+\mathsf{reb}_\xi=
m_\xi^{\mathsf{CSP}}+m_\xi^{\mathsf{PP}}+a_\xi\\
&=
\mathsf{DecCoeff}_{Q',t}(v_\xi)+a_\xi
\pmod t
\end{aligned}
\]
This is exactly the output prescribed by
$\mathcal F_{\mathsf{SD}}$.
\end{proof}

\begin{proof}[\DeepUL{\textbf{Proof of
Lemma~\ref{lem:presence-decoding-aggregation}}}]
If $\mathcal I=\varnothing$, all per-index claims are vacuous, the Beaver tree
is skipped, and the protocol defines the empty product as $Y=1$. Assume
henceforth that $\mathcal I\neq\varnothing$, and fix an arbitrary
$\xi=(j,s)\in\mathcal I$. By the additive BFV secret-key sharing relation and the local definitions in Section~\ref{subsec:presence-selected-decoding}:
\[
\begin{aligned}
p_\xi^{\mathsf{PP}}+p_\xi^{\mathsf{CSP}}
&=
[
 c_{0,j}^{Q'}
 +
 c_{1,j}^{Q'}
 (
 \mathsf{sk}_{R,e}^{\mathsf{PP}}
 +
 \mathsf{sk}_{R,e}^{\mathsf{CSP}}
)
]_s\\
&=
[
 c_{0,j}^{Q'}+c_{1,j}^{Q'}\mathsf{sk}_{R,e}
]_s
\pmod{Q'}
\end{aligned}
\]
Thus,
$v_\xi=p_\xi^{\mathsf{PP}}+p_\xi^{\mathsf{CSP}}\pmod{Q'}$ is precisely the
selected BFV decryption-phase coefficient. By coefficient-wise BFV
correctness at $Q'$:
$
\mathsf{DecCoeff}_{Q',t}(v_\xi)
=
z_j[s]
\pmod t.
$
By Lemma~\ref{lem:seldec-realization-correctness}, the selected-decoding
protocol returns no designated output to $\mathsf{PP}$ and returns to
$\mathsf{CSP}$:
$
w_\xi
=
\mathsf{DecCoeff}_{Q',t}(v_\xi)+a_\xi
\pmod t.
$
Substituting
$a_\xi=-\theta_\xi-\alpha_\xi^{\mathsf{PP}}\pmod t$ and
$x_\xi=z_j[s]-\theta_\xi\pmod t$ gives:
\[
\begin{aligned}
w_\xi
&=
z_j[s]-\theta_\xi-\alpha_\xi^{\mathsf{PP}}=
x_\xi-\alpha_\xi^{\mathsf{PP}}
\pmod t
\end{aligned}
\]
Therefore, the assignments
$
x_\xi^{\mathsf{PP}}
=
\alpha_\xi^{\mathsf{PP}},\,x_\xi^{\mathsf{CSP}}
=
w_\xi
$
satisfy:
$
x_\xi^{\mathsf{PP}}+x_\xi^{\mathsf{CSP}}
=
x_\xi
\pmod t.
$
For token conversion, Section~\ref{subsec:presence-aggregation} defines
$
y_\xi^{\mathsf{PP}}
=
\mu_\xi^{\mathsf{PP}}$, and
$y_\xi^{\mathsf{CSP}}
=
\mu_\xi^{\mathsf{CSP}}
+
\beta_\xi
(
 x_\xi^{\mathsf{CSP}}-\alpha_\xi^{\mathsf{CSP}}
)
\pmod t.
$
Using the token invariant from
Section~\ref{subsec:presence-zero-target}:
$\mu_\xi^{\mathsf{PP}}+\mu_\xi^{\mathsf{CSP}}
=
\beta_\xi
(
 \alpha_\xi^{\mathsf{PP}}+\alpha_\xi^{\mathsf{CSP}}
)
\pmod t
$, we obtain the following:
\[
\begin{aligned}
y_\xi^{\mathsf{PP}}+y_\xi^{\mathsf{CSP}}
&=
\mu_\xi^{\mathsf{PP}}+\mu_\xi^{\mathsf{CSP}}
+
\beta_\xi
(
 x_\xi^{\mathsf{CSP}}-\alpha_\xi^{\mathsf{CSP}}
)\\
&=
\beta_\xi
(
 \alpha_\xi^{\mathsf{PP}}+x_\xi^{\mathsf{CSP}}
)=
\beta_\xi x_\xi
\pmod t
\end{aligned}
\]
Thus, the parties hold additive shares of
$y_\xi=\beta_\xi x_\xi\pmod t$. Applying correctness of Beaver multiplication
\cite{beaver1991efficient} inductively at every internal node of the balanced
product tree yields Eq.~\ref{eq:presence-aggregate-product}:
Because $t$ is prime,
$\mathbb Z_t$ is a field and because every
$\beta_\xi\in\mathbb Z_t^*$, every multiplier is nonzero. Therefore,
$$
Y=0
\Longleftrightarrow
\prod_{\xi\in\mathcal I}x_\xi=0
\Longleftrightarrow
\exists\,\xi\in\mathcal I:x_\xi=0
$$
This proves the claim.
\end{proof}

\begin{proof}[\DeepUL{\textbf{Proof of Theorem~\ref{thm:presence-e2e-correctness}}}]
We prove the two directions of
Eq.~\ref{eq:presence-e2e-correctness}, e.g., \textit{Soundness} and \textit{Completeness} separately.

\textbf{\emph{Soundness.}}
Assume $\mathsf{present}=1$. By the output rule in
Section~\ref{subsec:presence-aggregation}, this implies $Y=0$.
Lemma~\ref{lem:presence-decoding-aggregation} then guarantees the
existence of a retained coefficient
$\xi=(j,s)\in\mathcal I$ satisfying $x_\xi=0$.
Lemma~\ref{lem:presence-no-wrap} excludes every early dummy coefficient
and every nonmatching complete candidate. Hence, $\xi$ represents a
complete candidate whose correlation value is
$z_j[s]=L_{\mathsf{fix}}$.

Let the absolute starting position of this candidate be
$u
=
\mathsf{off}_j+s-(L_P-1)$.
Because $\xi\in\mathcal I$, its absolute end position satisfies
$\mathsf{off}_j+s\equiv7\pmod8$. Since
$L_P-1\equiv7\pmod8$, it follows that $u\equiv0\pmod8$.
Moreover, completeness of the candidate gives
$0\leq u\leq L-L_P$. By
Lemma~\ref{lem:presence-correlation-correctness}, the corresponding
message window has restricted Hamming distance zero from $\mathbf p$.
Therefore, the right-hand side of
Eq.~\ref{eq:presence-e2e-correctness} holds.

\textbf{\emph{Completeness.}}
Assume the right-hand side of
Eq.~\ref{eq:presence-e2e-correctness} holds. Thus, some retained
message contains a byte-aligned window beginning at a valid offset $u$
whose restricted Hamming distance from $\mathbf p$ is zero.
Lemma~\ref{lem:presence-fragment-coverage} guarantees that this window
is contained in a fragment and that its absolute end position is
represented by exactly one retained coefficient
$\xi\in\mathcal I$.
Lemma~\ref{lem:presence-correlation-correctness} then gives
$z_j[s]=L_{\mathsf{fix}}$. Because $\xi$ is a complete candidate, its
target is $\theta_\xi=L_{\mathsf{fix}}$, and hence $x_\xi=0$.
By Lemma~\ref{lem:presence-decoding-aggregation}, this zero is preserved
through token conversion and the aggregate satisfies $Y=0$.
The output rule therefore gives $\mathsf{present}=1$.

The soundness and completeness implications establish
Eq.~\ref{eq:presence-e2e-correctness}. If
$\mathcal I=\varnothing$, the protocol defines $Y=1$ and hence
$\mathsf{present}=0$. Moreover, by
Lemma~\ref{lem:presence-fragment-coverage}, no complete byte-aligned
candidate exists in this case.
\end{proof}

\section{Detailed Security Proof for \ProtocolName{}}
\label{app:seek-security}

This appendix proves Lemmas~\ref{lem:seldec-privacy} and~\ref{lem:presence-output-privacy} and Theorem~\ref{thm:seek-semi-honest} in its stated input-independent preprocessing-hybrid model. Ideal batched OLE and ideal generators provide fresh, independent, non-reused Boolean triples, daBits, and arithmetic triples, while $\Pi_{\mathsf{SD}}^{\mathsf{pre}}$ remains a concrete part of \ProtocolName{} using the ideal Boolean-triple and daBit interfaces. 
The privacy proof begins after all the functional validation (Section \ref{subsec:functional-validation}) tests are passed. 
The public setup transcript and the corrupted party's local setup state, including its additive secret-key share, are auxiliary inputs to the simulation. The receiver $\mathsf R$ does not participate in the online protocol, and the adversary statically corrupts at most one of $\mathsf{PP}$ and $\mathsf{CSP}$.

\begin{proof}[\DeepUL{\textbf{Proof of Lemma~\ref{lem:seldec-privacy}}}]
The result follows from Lemma~\ref{lem:seldec-prep-realization}, Lemma \ref{lem:seldec-realization-correctness}, and Lemma~\ref{lem:seldec-realization-security} by sequential composition.
\end{proof}

\noindent \ding{110}
\paragraphNew{\DeepUL{Simulation Notation}}
Let $A\in\{\mathsf{PP},\mathsf{CSP}\}$. Lemma~\ref{lem:seldec-privacy} guarantees a simulator $\mathcal S_{\mathsf{SD}}^A$ for $\Pi_{\mathsf{SD}}^{\mathsf{pre}}$ followed by $\Pi_{\mathsf{SD}}^{\mathsf{on}}$. It uses the public information, the corrupted party's local input and ideal-resource state, and the prescribed $\mathcal F_{\mathsf{SD}}$ output. This output is only $(w_\xi)_{\xi\in\mathcal I}$ for $\mathsf{CSP}$. The ideal OLE and correlated-resource generators produce no backend transcript. Similarly, semi-honest Beaver-multiplication security guarantees a simulator $\mathcal S_{\mathsf{Prod}}^A$ for the product-tree transcript using the corrupted party's leaf shares, ideal arithmetic-triple state, and prescribed output~\cite{beaver1991efficient}. The prescribed output is only $Y$ for $\mathsf{PP}$.

\begin{proof}[\DeepUL{\textbf{Proof of
Lemma~\ref{lem:presence-output-privacy}}}]
For each $\xi\in\mathcal I$, token generation in the ideal
batched-OLE hybrid samples
$\alpha_\xi^{\mathsf{PP}},\alpha_\xi^{\mathsf{CSP}},
\mu_\xi^{\mathsf{CSP}}\xleftarrow{}\mathbb Z_t$ and
$\beta_\xi\xleftarrow{}\mathbb Z_t^*$ independently, and sets
$\mu_\xi^{\mathsf{PP}}
=
\beta_\xi
(
\alpha_\xi^{\mathsf{PP}}+\alpha_\xi^{\mathsf{CSP}}
)
-
\mu_\xi^{\mathsf{CSP}}
\pmod t.
$
Because $\mu_\xi^{\mathsf{CSP}}$ is uniform, thus for fixed $a,m\in\mathbb Z_t$ and $b\in\mathbb Z_t^*$, we have:
$$
\begin{aligned}
&\Pr[\mu_\xi^{\mathsf{PP}}=m
 \mid \alpha_\xi^{\mathsf{PP}}=a,\beta_\xi=b]\\
&\quad=\sum_{u\in\mathbb Z_t}
 \Pr[\alpha_\xi^{\mathsf{CSP}}=u,
 \mu_\xi^{\mathsf{CSP}}=b(a+u)-m]
 =\frac1t
\end{aligned}
$$
Hence the pair
$(\alpha_\xi^{\mathsf{PP}},\mu_\xi^{\mathsf{PP}})$ observed by
$\mathsf{PP}$ is independent of $\beta_\xi$. Since tokens are independent
across retained indices and selected decoding does not reveal or use a
$\mathsf{CSP}$-side token component, conditioning on
$V_{\mathsf{PP}}^0$ leaves
$(\beta_\xi)_{\xi\in\mathcal I}$ uniformly distributed over
$(\mathbb Z_t^*)^K$. By Eq.~\ref{eq:presence-aggregate-product}, if $\mathsf{present}=1$, correctness gives $x_\xi=0$ for some $\xi$, so
$Y=0$. If $\mathsf{present}=0$ and $K\geq1$, every $x_\xi$ is nonzero.
Because $\mathbb Z_t$ is a field and the product of independent uniform
elements of $\mathbb Z_t^*$ is uniform in $\mathbb Z_t^*$, the aggregate
$Y$ is uniform in $\mathbb Z_t^*$. For $K=0$, the protocol defines the
empty product as $Y=1$. The simulator therefore samples
\[
\widetilde Y\gets
\begin{cases}
1, & K=0\\
0, & K\geq1\ \text{and}\ \mathsf{present}=1\\
U(\mathbb Z_t^*)
& K\geq1\ \text{and}\ \mathsf{present}=0
\end{cases}
\]
For $K=0$, the simulator sets $\widetilde Y=1$ and emits the empty aggregation transcript, exactly as the real protocol. Otherwise, it invokes $\mathcal S_{\mathsf{Prod}}^{\mathsf{PP}}$ on the local leaf shares $(\mu_\xi^{\mathsf{PP}})_{\xi\in\mathcal I}$ and prescribed output
$\widetilde Y$. This simulates the complete aggregation view from
$V_{\mathsf{PP}}^0$, $K$, and the final $\mathsf{present}$ bit.
\end{proof}

\begin{proof}[\DeepUL{\textbf{Proof of
Theorem~\ref{thm:seek-semi-honest}}}] 
Let $A\in\{\mathsf{PP},\mathsf{CSP}\}$ be the corrupted party. The simulator $\mathcal S_A$ receives its ideal input, local setup state, prescribed leakage, and ideal output. For each query, the ideal preprocessing interfaces supply its OLE-token and local preprocessing shares, while $\mathcal S_{\mathsf{SD}}^A$ and $\mathcal S_{\mathsf{Prod}}^A$ simulate the composed selected-decoding and aggregation views, respectively.
The simulator maintains one persistent dummy corpus for the epoch. Whenever
the public leakage identifies a stored fragment $j$, it samples one fresh
encryption: $\widetilde c_{M,j}
\xleftarrow{}
\mathsf{Enc}_{\mathsf{pk}_{R,e}}(0)$ and applies the public preprocessing to obtain
$\widetilde c_{M,j}^{\pm1}$. If $A=\mathsf{CSP}$, it also sets
$\widetilde c_P\xleftarrow{}\mathsf{Enc}_{\mathsf{pk}_{R,e}}(0)$ and if
$A=\mathsf{PP}$, it generates the trapdoor $c_P$ from the corrupted party's query. Considering $c_P^\star \in \{\widetilde c_P,c_P\}$ as the resulting trapdoor, it then computes the correlation, relinearization and modulus switching operations, exactly as in the real protocol:
\[
\widetilde c_{\mathsf{corr},j}^{Q'}
\leftarrow
\mathsf{ModSwitch}_{Q\rightarrow Q'}
(
\mathsf{Relin}_{\mathsf{rlk}_{R,e}}
(
\mathsf{Eval}_{\mathsf{mult}}
(
\widetilde c_{M,j}^{\pm1},c_P^\star
)
))
\]
By the setup condition in
Theorem~\ref{thm:seek-semi-honest}, the corrupted party's additive
secret-key share is distributed as prescribed in Section~\ref{subsec:setup}
and is independent of the complete BFV secret key and published keys. BFV
IND-CPA security therefore permits a polynomial-length hybrid replacing the
hidden corpus ciphertexts, and for a corrupted $\mathsf{CSP}$ the hidden
trapdoor, by these dummy encryptions. Public homomorphic evaluation, modulus
switching, and local share computation are efficient post-processing and
preserve indistinguishability.

If $K=0$, the simulator emits empty selected-decoding, zero-testing, and
aggregation transcripts. It sets $Y=1$ for a corrupted $\mathsf{PP}$, and
$\mathsf{present}=0$ follows from the ideal functionality. The following two
cases therefore assume $K\geq1$.

\noindent \ding{110}
\DeepUL{\textbf{Case 1: Corrupted $\mathsf{CSP}$.}}
For each $\xi=(j,s)\in\mathcal I$, the simulator computes the corrupted
party's local contribution from the evaluated dummy ciphertext as:
\[
\widetilde p_\xi^{\mathsf{CSP}}
=
[
\widetilde c_{0,j}^{Q'}
+
\widetilde c_{1,j}^{Q'}
(
\mathsf{sk}_{R,e}^{\mathsf{CSP}}\bmod Q'
)
]_s
\pmod{Q'}
\]
By the ideal batched-OLE token distribution, $\alpha_\xi^{\mathsf{PP}}$ remains hidden and uniform in $\mathbb Z_t$. Hence $w_\xi=x_\xi-\alpha_\xi^{\mathsf{PP}}$ is uniform in $\mathbb Z_t$ for every fixed $x_\xi$, and $(w_\xi)_{\xi\in\mathcal I}\equiv U(\mathbb Z_t^K)$. The simulator samples $\widetilde{\mathbf w}\gets\mathbb Z_t^K$ and invokes $\mathcal S_{\mathsf{SD}}^{\mathsf{CSP}}$ with $(\widetilde p_\xi^{\mathsf{CSP}})_{\xi\in\mathcal I}$ as the corrupted $\mathsf{CSP}$'s local selected-decoding input and $\widetilde{\mathbf w}$ as its prescribed output.
Using its local token state from the ideal batched-OLE hybrid, it then
computes $\widetilde y_\xi^{\mathsf{CSP}}$ as:
\[
\widetilde y_\xi^{\mathsf{CSP}}
=
\mu_\xi^{\mathsf{CSP}}
+
\beta_\xi
(
\widetilde w_\xi-\alpha_\xi^{\mathsf{CSP}}
)
\pmod t
\]
It invokes $\mathcal S_{\mathsf{Prod}}^{\mathsf{CSP}}$ on
$(\widetilde y_\xi^{\mathsf{CSP}})_{\xi\in\mathcal I}$ with no prescribed
output. Thus the simulated $\mathsf{CSP}$ view uses neither the plaintext
corpus nor the keyword, its length, the search result, or any
match value.

\noindent \ding{110}
\DeepUL{\textbf{Case 2: Corrupted $\mathsf{PP}$.}}
The simulator uses the corrupted party's genuine query to construct the
private targets and the genuine trapdoor, evaluates it against the persistent
dummy corpus, and computes:
\[
\widetilde p_\xi^{\mathsf{PP}}
=
[
\widetilde c_{1,j}^{Q'}
(
\mathsf{sk}_{R,e}^{\mathsf{PP}}\bmod Q'
)
]_s
\pmod{Q'}
\]
It sets $a_\xi=-\theta_\xi-\alpha_\xi^{\mathsf{PP}}\pmod t$ as in the real protocol and invokes $\mathcal S_{\mathsf{SD}}^{\mathsf{PP}}$ with $((\widetilde p_\xi^{\mathsf{PP}},a_\xi))_{\xi\in\mathcal I}$ as the corrupted $\mathsf{PP}$'s local selected-decoding inputs and $\mathsf{PP}$ receives no designated output.
The local zero-test leaf share is
$y_\xi^{\mathsf{PP}}=\mu_\xi^{\mathsf{PP}}$. Using $K$ and the ideal bit
$\mathsf{present}$, the simulator samples $\widetilde Y$ according to
Lemma~\ref{lem:presence-output-privacy} and invokes
$\mathcal S_{\mathsf{Prod}}^{\mathsf{PP}}$ on
$(\mu_\xi^{\mathsf{PP}})_{\xi\in\mathcal I}$ with prescribed output
$\widetilde Y$. The simulated party reconstructs $\widetilde Y$ and outputs
$\mathbf 1[\widetilde Y=0]=\mathsf{present}$. Here, no individual hit, count, location, or match identity is selected or revealed by the simulator.

To justify the construction formally, let $H_0^A$ denote the real execution in the stated preprocessing-hybrid model. In $H_1^A$, the execution of $\Pi_{\mathsf{SD}}^{\mathsf{pre}}$ followed by $\Pi_{\mathsf{SD}}^{\mathsf{on}}$ is replaced with $\mathcal F_{\mathsf{SD}}$ and its simulator $\mathcal S_{\mathsf{SD}}^A$, as justified by Lemma~\ref{lem:seldec-privacy}. For a corrupted $\mathsf{CSP}$, the prescribed output may be sampled as $\widetilde{\mathbf w}\gets\mathbb Z_t^K$, since $w_\xi=x_\xi-\alpha_\xi^{\mathsf{PP}}$ is exactly uniform under the ideal batched-OLE distribution and $\mathsf{PP}$ receives no designated output. In $H_2^A$, the Beaver-tree transcript is replaced by $\mathcal S_{\mathsf{Prod}}^A$, using Lemma~\ref{lem:presence-output-privacy} to sample the output delivered to a corrupted $\mathsf{PP}$. In $H_3^A$, every BFV ciphertext whose plaintext is hidden from $A$ is replaced by the corresponding dummy encryption, recompute the public evaluations and local contributions, and regenerate the simulated views using the same prescribed outputs. Thus,
\[
\mathsf{View}_A^{\mathsf{real}}
\equiv H_0^A
\approx_c H_1^A
\approx_c H_2^A
\approx_c H_3^A
\equiv\mathsf{View}_{\mathcal S_A}^{\mathsf{ideal}}
\]
The transitions follow from Lemma~\ref{lem:seldec-privacy} and the uniform-mask argument, Beaver-product security together with Lemma~\ref{lem:presence-output-privacy}, and BFV IND-CPA security under efficient post-processing.

For a polynomially bounded, sequentially adaptive sequence of authorized queries over the same fixed retained corpus within same epoch, the simulator retains the same dummy corpus and simulates each query as above using fresh, disjoint, one-time preprocessing. For a corrupted $\mathsf{CSP}$, hidden trapdoors are replaced as they are received. Sequential composition applies because every query consumes fresh independent one-time preprocessing. The only cumulative output is the ideal presence-bit sequence and its logical implications, exactly as stated in Theorem~\ref{thm:seek-semi-honest}.
\end{proof}

\section{BFV Evaluation Parameters}
\label{app:bfv-params}

Table~\ref{tab:appendix-bfv-parameters} shows the BFV parameter sets used in our evaluation. For each of the 20 $(N,L_{\max})$ profiles, we ran three independent trials, testing a total of 540 adversarial correlation instances. The test cases include both pattern-length extremes $L_P\in\{8,L_{\max}\}$, dense all $+1$ and all $-1$ inputs, maximum Hamming distance, earliest complete coefficients, fragment boundaries, and cases with the minimum final fragment. Each instance tests signed-fragment preprocessing, ct--ct multiplication and relinearization, every intermediate level of sequential modulus switching, and the direct $Q\rightarrow Q'$ path. In $6615$ stage-level observations, all $N$ coefficients of every decrypted polynomial matched an independent negacyclic-polynomial oracle, every selected coefficient decoded as expected, and every observed noise budget stayed positive. Table~\ref{tab:bfv-noise-budget} gives the minimum remaining BFV invariant noise budget at $Q'$. Both the sequential and direct modulus-switching paths yielded the same minimum for every profile, with a globally observed minimum of 8 bits. These results provide implementation-level evidence for the tested stress cases.

\begin{table}[!h]
\centering
\caption{BFV parameters used in the evaluation.}
\label{tab:appendix-bfv-parameters}
\scriptsize
\setlength{\tabcolsep}{2.5pt}
\renewcommand{\arraystretch}{1.12}
\begin{tabular}{@{}r r c c l@{}}
\toprule
\multicolumn{1}{c}{$N$}
&
\multicolumn{1}{c}{$t$}
&
\multicolumn{1}{c}{\shortstack{$Q$\\(bits)}}
&
\multicolumn{1}{c}{\shortstack{$\{q_\ell\}$\\(bits)}}
&
\multicolumn{1}{c}{\shortstack{$Q'$\\(bits)}}
\\
\midrule

$4096$
& $1{,}032{,}193$
& $72$
& $36+36$
& $36$ \\

$8192$
& $1{,}032{,}193$
& $174$
& $43+43+44+44$
& $43$ \\

$16384$
& $786{,}433$
& $389$
& $3{\times}48+5{\times}49$
& $48$ \\

$32768$
& $786{,}433$
& $825$
& $15{\times}55$
& $55$ \\

\bottomrule
\end{tabular}
\end{table}

\begin{table}[!h]
\centering
\caption{Minimum observed remaining BFV invariant noise budget
at $Q'$ (bits).}
\label{tab:bfv-noise-budget}
\small
\setlength{\tabcolsep}{4pt}
\begin{tabular}{rccccc}
\toprule
$N\backslash L_{\max}$ & 128 & 256 & 512 & 1024 & 2048 \\
\midrule
4096  & 9  & 9  & 8  & 8  & 9  \\
8192  & 15 & 15 & 15 & 15 & 15 \\
16384 & 20 & 20 & 20 & 20 & 20 \\
32768 & 26 & 26 & 26 & 26 & 26 \\
\bottomrule
\end{tabular}
\end{table}

\end{document}